\documentclass[journal,twoside,web]{ieeecolor}

\usepackage[table,dvipsnames]{xcolor} 
\PassOptionsToPackage{table,dvipsnames}{xcolor}
\usepackage{amsmath,amssymb,amsfonts}
\usepackage{algorithm}
\usepackage{algpseudocode}

\usepackage{bm}

\usepackage{cite}

\usepackage{eucal}

\usepackage{generic}
\usepackage{graphicx}

\usepackage{mathtools}
\usepackage{mathrsfs}
\usepackage{newtxmath} 

\usepackage{textcomp}
\usepackage{times} 
\usepackage{tikz}
\usetikzlibrary{patterns}

\usepackage{pgfplots}
\pgfplotsset{compat=1.18}  

\makeatletter
\def\endfigure{\end@float}
\def\endtable{\end@float}   
\makeatother

\usepackage[caption=false,font=footnotesize]{subfig}

\usepackage{hyperref}
\usepackage[acronym]{glossaries}
\usepackage[nameinlink]{cleveref} 

\creflabelformat{equation}{#2#1#3}
\crefname{equation}{Equation}{Equations}
\crefname{theorem}{Theorem}{Theorems}
\crefname{corollary}{Corollary}{Corollaries}
\crefname{lemma}{Lemma}{Lemmas}

\crefname{section}{Section}{Sections}

\crefname{table}{Table}{Tables}
\crefname{figure}{Figure}{Figures}
\crefname{algorithm}{Algorithm}{Algorithms}

\Crefname{table}{Table}{Tables}
\Crefname{figure}{Figure}{Figures}
\Crefname{algorithm}{Algorithm}{Algorithms}

\crefname{ineq}{Inequality}{Inequalities}
\Crefname{ineq}{Inequality}{Inequalities}
\creflabelformat{ineq}{#2{\upshape#1}#3}

\crefname{prob}{Problem}{Problems}
\Crefname{prob}{Problem}{Problems}
\creflabelformat{prob}{#2{\upshape#1}#3} 

\crefname{ass}{Assumption}{Assumptions}
\Crefname{ass}{Assumption}{Assumptions}
\creflabelformat{ass}{#2{\upshape#1}#3}

\newtheorem{definition}{Definition}
\newtheorem{assumption}{Assumption}

\newtheorem{theorem}{Theorem}
\newtheorem{lemma}{Lemma}
\newtheorem{corollary}{Corollary}

\DeclarePairedDelimiter\set{\{}{\}}

\newcommand{\sthat}{\text{s.t.}}
\newcommand{\epi}{\text{epi}\,}

\newcommand{\col}{\text{col}}
\newcommand{\row}{\text{row}}
\newcommand{\diag}{\text{diag}}
\newcommand{\ve}{\text{vec}}
\newcommand{\ive}{\text{vec}^{-1}}

\newcommand{\abs}[1]{\left| #1 \right|}

\newcommand{\brackets}[1]{\left( #1 \right)}

\newcommand{\dbx}{\dot{\textbf{x}}}

\newcommand{\vertex}[1]{\mathcal{V}(#1)}

\newcommand{\edge}[1]{\mathcal{E}(#1)}
\newcommand{\arc}[1]{\mathcal{A}(#1)}
\newcommand{\outneigh}[2]{\mathcal{N}_{#1}^+(#2)}
\newcommand{\inneigh}[2]{\mathcal{N}_{#1}^-(#2)}
\newcommand{\neigh}[2]{\mathcal{N}_{#1}(#2)}

\newcommand{\cA}{\mathcal{A}}

\newcommand{\cC}{\mathcal{C}}
\newcommand{\cD}{\mathcal{D}}
\newcommand{\cE}{\mathcal{E}}

\newcommand{\cG}{\mathcal{G}}
\newcommand{\cH}{\mathcal{H}}

\newcommand{\cL}{\mathcal{L}}

\newcommand{\cN}{\mathcal{N}}
\newcommand{\cO}{\mathcal{O}}

\newcommand{\cQ}{\mathcal{Q}}
\newcommand{\cR}{\mathcal{R}}

\newcommand{\cV}{\mathcal{V}}

\newcommand{\cZ}{\mathcal{Z}}

\newcommand{\bbA}{\mathbb{A}}

\newcommand{\bbL}{\mathbb{L}}

\newcommand{\bbN}{\mathbb{N}}

\newcommand{\bbP}{\mathbb{P}}
\newcommand{\bbQ}{\mathbb{Q}}
\newcommand{\bbR}{\mathbb{R}}
\newcommand{\bbS}{\mathbb{S}}

\newcommand{\bbX}{\mathbb{X}}
\newcommand{\bbY}{\mathbb{Y}}

\newcommand{\bone}{\mathbf{1}}
\newcommand{\bzero}{\mathbf{0}}

\newcommand{\bA}{\mathbf{A}}
\newcommand{\bB}{\mathbf{B}}
\newcommand{\bC}{\mathbf{C}}
\newcommand{\bD}{\mathbf{D}}
\newcommand{\bE}{\mathbf{E}}

\newcommand{\bG}{\mathbf{G}}
\newcommand{\bH}{\mathbf{H}}
\newcommand{\bI}{\mathbf{I}}
\newcommand{\bJ}{\mathbf{J}}
\newcommand{\bK}{\mathbf{K}}

\newcommand{\bM}{\mathbf{M}}

\newcommand{\bP}{\mathbf{P}}
\newcommand{\bQ}{\mathbf{Q}}
\newcommand{\bR}{\mathbf{R}}
\newcommand{\bS}{\mathbf{S}}
\newcommand{\bT}{\mathbf{T}}
\newcommand{\bU}{\mathbf{U}}
\newcommand{\bV}{\mathbf{V}}
\newcommand{\bW}{\mathbf{W}}
\newcommand{\bX}{\mathbf{X}}
\newcommand{\bY}{\mathbf{Y}}
\newcommand{\bZ}{\mathbf{Z}}

\newcommand{\ba}{\mathbf{a}}

\newcommand{\be}{\mathbf{e}}

\newcommand{\bj}{\mathbf{j}}

\newcommand{\bq}{\mathbf{q}}
\newcommand{\br}{\mathbf{r}}

\newcommand{\bu}{\mathbf{u}}
\newcommand{\bv}{\mathbf{v}}
\newcommand{\bw}{\mathbf{w}}
\newcommand{\bx}{\mathbf{x}}
\newcommand{\by}{\mathbf{y}}

\newcommand{\barbH}{\overline{\bH}}

\newcommand{\barbJ}{\overline{\bJ}}

\newcommand{\barbQ}{\overline{\bQ}}
\newcommand{\barbR}{\overline{\bR}}

\newcommand{\barbX}{\overline{\bX}}
\newcommand{\barbY}{\overline{\bY}}

\newcommand{\barbbQ}{\overline{\bbQ}}

\newcommand{\barcQ}{\overline{\cQ}}

\newcommand{\hatbj}{\widehat{\bj}}

\newcommand{\hatbJ}{\widehat{\bJ}}

\def\hatbQ{\widehat{\bQ}}
\newcommand{\hatbR}{\widehat{\bR}}

\newcommand{\hatbX}{\widehat{\bX}}
\newcommand{\hatbY}{\widehat{\bY}}

\newcommand{\hatbbQ}{\widehat{\bbQ}}

\newcommand{\tilbJ}{\widetilde{\bJ}}

\newcommand{\tilbR}{\widetilde{\bR}}

\newcommand{\barN}{\overline{N}}
\newcommand{\barn}{\overline{n}}
\newcommand{\barm}{\overline{m}}

\newcommand{\thh}{\mathrm{th}}

\newacronym{vndt}{VNDT}{Vidyasagar's Network Dissipativity Theorem}
\newacronym{kyp}{KYP Lemma}{Kalman-Yakubovixh-Popov Lemma}
\newacronym{io}{IO}{input-output}
\newacronym{osp}{OSP}{output strictly passive}

\newacronym{lti}{LTI}{linear time-invariant}
\newacronym{lmi}{LMI}{linear matrix inequality}
\newacronym{iqc}{IQC}{integral quadratic constraint}

\newacronym{admm}{ADMM}{alternating direction method of multipliers}
\newacronym{sdp}{SDP}{semidefinite programming}
\newacronym{ssa}{SSA}{sequential stability analysis}

\newacronym{uav}{UAV}{unmanned aerial vehicles}

\def\BibTeX{{\rm B\kern-.05em{\sc i\kern-.025em b}\kern-.08em
    T\kern-.1667em\lower.7ex\hbox{E}\kern-.125emX}}
\begin{document}
\title{Dissipativity-Based Distributed Stability Analysis for Networks with Heterogeneous Nonlinear Agents}
\author{Ingyu Jang, \IEEEmembership{Graduate Student Member, IEEE}, Ethan J. LoCicero, and Leila Bridgeman, \IEEEmembership{Member, IEEE}
\thanks{This paragraph of the first footnote will contain the date on 
which you submitted your paper for review. It will also contain support 
information, including sponsor and financial support acknowledgment. For 
example, ``This work was supported in part by the U.S. Department of 
Commerce under Grant 123456.'' }
\thanks{*This work is supported by ONR Grant No. N00014-23-1-2043.}
\thanks{Ingyu Jang (PhD Student), and Leila Bridgeman (Assistant Professor) are with the Department of Mechanical Engineering and Material Science, Duke University, Durham, NC 27708 USA
        (email: {\tt\small ij40@duke.edu; ljb48@duke.edu}, phone: 919-225-4215). }
\thanks{Ethan J. LoCicero (Fellow) is with TerraPower, LLC, Bellevue, WA 98008 USA
        (email: {\tt\small ethanlocicero@gmail.com}).}}

\maketitle

\begin{abstract}
Stabilizing large networks of nonlinear agents is challenging; decomposition and distributed analysis of these networks are crucial for computational tractability and information security. Vidyasagar's Network Dissipativity Theorem enables both properties concurrently in distributed network analysis. This paper explored combining it with the alternating direction methods of multipliers to develop distributed stability analysis for networks of inhomogeneous, nonlinear agents. One algorithm enhances information security by requiring agents to share only a dissipativity characterization, not a dynamical model, for stability analysis. A second algorithm further restricts this information sharing to their clique, thereby enhancing security, and can also reduce the computational burden of stability analysis if the network allows chordal decomposition. The convergence of the proposed algorithms is demonstrated, and criteria are identified for decomposable networks facilitating chordal decomposition. The effectiveness of the proposed methods is demonstrated through numerical examples involving a swarm of linearized unmanned aerial vehicles and networks beyond linear time-invariant agents.
\end{abstract}


\section{INTRODUCTION}

\IEEEPARstart{I}{n} recent decades, the analysis and control of large-scale multi-agent networked systems have been recognized as key challenges in system design \cite{sztipanovits2011toward}.
However, centralized control of such systems is hindered by the heavy computational and communication burdens associated with large-scale systems. 
To mitigate these limitations, various controller synthesis methods have been developed, which promote sparse inter-agent communication during online operation \cite{jovanovic2016controller,lian2018sparsity,fardad2011sparsity,lin2013design}.
Although these sparse controllers are advantageous during online operation, they pose significant difficulties, as these methods overlook intellectual property concerns \cite{isik2023impact} and involve high computational costs in the offline design and analysis phase. Distributed stability analysis techniques are needed that accommodate heterogeneous, nonlinear agents.
This paper presents methods that use \gls{vndt} \cite{vidyasagar1981input} for distributed stability analysis of networked nonlinear systems using consensus-based algorithms, which avoid sharing agent information.

Among large-scale stability analysis theories, we adopt \gls{vndt} \cite{vidyasagar1981input}, as it does not rely on specific network structures or agent homogeneity. Various \gls{io} stability theorems, including the Passivity \cite{lozano2013dissipative}, Small Gain \cite{zames1966input1}, and Conic Sector Theorems \cite{zames1966input1}, are special cases of \gls{vndt}. 
Each one establishes closed-loop \gls{io} stability of a network from agents' \textit{open-loop} dissipativity properties. Dissipativity analysis of individual agents is generally easier than analyzing interconnected dynamics; for example, bounding the gain of $N$ interconnected \gls{lti} $n$-state systems via the Bounded Real Lemma scales as $\mathcal{O}((Nn)^6)$, while finding their individual gains scales as an $N\mathcal{O}(n^6)$ \cite{nesterov1994interior}.
Moreover, analyzing their closed-loop dynamics requires a unified model framework, but their open-loop dynamics can be analyzed using completely disparate methods. Hence, \gls{vndt} can incorporate agents with varied nonlinearities, such as time delays, parametric uncertainty, and stochastic behavior, which can be characterized using dissipativity \cite{li2002delay, mahmoud2009dissipativity, haddad2022dissipativity}.
For example, the dissipativity of an agent with a reliable \gls{lti} model can be analyzed using \cite{gupta1996robust}, while agents with sufficient data can employ \cite{romer2019one}.
Consequently, extensive research has employed \gls{vndt} to analyze and control large-scale systems \cite{rojas2009dynamic,hioe2014decentralized,ghanbari2016large,liu2023incremental,hatanaka2015passivity,chopra2006passivity,chopra2012output,miyano2020continuous,stan2007analysis,yao2009passive,hirche2008stabilizing,anderson2011dynamical,meissen2015compositional,agarwal2020distributed}.

Although \gls{vndt} suits large-scale system analysis, applying it directly still requires solving a matrix inequality compatibility involving all agents' dissipativity properties \cite{rojas2009dynamic,hioe2014decentralized,ghanbari2016large,liu2023incremental}. 
Several distributed methods to satisfy this compatibility exist \cite{hatanaka2015passivity,chopra2006passivity,chopra2012output,miyano2020continuous,stan2007analysis,yao2009passive,hirche2008stabilizing,anderson2011dynamical,meissen2015compositional,agarwal2020distributed}, but most approaches~\cite{hatanaka2015passivity,chopra2006passivity,chopra2012output,miyano2020continuous,stan2007analysis,yao2009passive} impose passivity on all agents, ensuring the compatibility a priori and creating a fully decentralized analysis. However, passivity is more conservative than general dissipativity. 
Some studies \cite{hirche2008stabilizing,anderson2011dynamical} reduce this conservatism by identifying fixed dissipativity structures of the agents a priori, but infinitely many valid dissipativity characterizations of an agent exist, making the optimal choice unavailable in advance. A poor selection can violate the compatibility, while the right selection would affirm it.
In \cite{agarwal2020distributed}, this was mitigated by characterizing agents and verifying compatibility sequentially, but this can still fix inopportune characterizations for agents analyzed earlier in the sequence, and the process applies only to \gls{lti} agents. Instead, \cite{meissen2015compositional} tackled the combined problem of performing local analysis and verifying network compatibility using distributed optimization. This is the strategy explored here too, and \cref{alg:01} discussed in \cref{Chap:ISNSA} specializes \cite{meissen2015compositional} from \glspl{iqc} to dissipativity, but leaves the network characterization as a design variable, where it was prescribed a priori in \cite{meissen2015compositional}. However, this doesn't totally remove the large-scale compatibility issues, so we further explore how and when network structure can enable the fully decomposed problem.

This work tackles a distributed stability analysis problem by integrating distributed optimization and \gls{vndt} to compile global dissipativity from local dissipativity. This approach is implemented using the \gls{admm}, which is well-suited for distributed algorithms with global constraints without requiring specific network topology \cite{boyd2011distributed}. Through this work, each agent independently analyzes its own dissipativity, shares its identified dissipativity parameters with neighbors, and iteratively optimizes its dissipativity parameters to ensure network-wide stability constraints. However, this integration alone, as in \cite{meissen2015compositional}, results in computational challenges and requires dissipativity parameter sharing due to the large-scale global constraint. To overcome this, we apply network decomposition, enabling a fully distributed implementation with significantly smaller, parallelizable subproblems.

The integrated approach comprises two algorithms. 
As in \cite{meissen2015compositional,ingyu2025consensus}, the first algorithm directly combines the methods, while preserving agent privacy. 
Despite enabling information secure stability analysis of large-scale systems, it leads to a network-scale matrix inequality problem to impose compatibility between agents.
The second algorithm extends \cite{meissen2015compositional,ingyu2025consensus} by using graph-theoretic concepts to lower the high computational cost, accelerate convergence speed, and reduce communication bandwidth.


Following this introduction, \cref{Chap:Preliminaries} outlines the necessary preliminaries. 
\cref{Chap:Problem Statement} introduces the main problem addressed in this paper.
\cref{Chap:ISNSA} discusses the direct distributed solution, and critically the bottlenecks in computation time.
\cref{Chap:Chordal_Decomposition} provides the methods for decomposing the high-dimensional \gls{lmi} and the second algorithm that enables a more scalable distributed implementation.
\cref{Chap:Extension to Nonlinear Systems} discusses applications to non \gls{lti} agents.
\cref{Chap:Numerical Example} demonstrates the feasibility of the proposed algorithm using a large-scale 2D swarm \glspl{uav} and nonlinear robot manipulator network.
The paper concludes with a summary and discussion of future work in \cref{Chap:Conclusion}.

\section{PRELIMINARIES} \label{Chap:Preliminaries}

\subsection{Notation}
The sets of real, natural numbers, and natural numbers up to $n$ are denoted by $\bbR$, $\bbN$, $\bbN_n$, respectively. The set of real $n{\times}m$ matrices is $\bbR^{n{\times} m}$. The set of $n{\times} n$ symmetric matrices is $\bbS^n$, with $\bbS^n_-$ denoting its negative semi-definite subset. The notation $\bA{\prec}0$ indicates that $\bA$ is negative-definite. 
The cardinality of a set $\bbA$ is denoted by $|\bbA|$. 
The $n{\times} n$ identity matrix, $n{\times} m$ matrix of ones, and $n\times m$ zero matrix are denoted $\bI_n$, $\bone_{n\times m}$, and $\bzero_{n\times m}$, respectively.

The set of square integrable functions is $\cL_{2}$. The Frobenius norm and $\cL_2$ norm are denoted by ${\|}{\cdot}{\|}_F$ and ${\|}{\cdot}{\|}_2$, respectively.  The truncation of a function $\textbf{y}(t)$ at $T$ is denoted by $\by_T(t)$, where $\by_T(t){=}\by(t)$ if $t{\leq}T$, and $\by_T(t){=}0$ otherwise. If $\|\by_T\|_2^2=\langle\by_T,\by_T\rangle{=}\int_0^{\infty}\by_T^T(t)\by_T(t)dt{<}\infty$ for all $T{\geq}0$, then $\by{\in}\cL_{2e}$, where $\cL_{2e}$ is the extended $\cL_2$ space.
The indicator function is denoted by $I_\bbA:\Omega\to\set{0,+\infty}$,where $I_\bbA(x)=0$ if $x\in\bbA$ and $I_\bbA(x)=+\infty$ otherwise for all $x\in\Omega$. 

\subsection{Vectorization}

Matrix equations can be reformulated as linear equations through vectorization. For $\bA{=}[\ba_1{\cdots} \ba_n]{\in}\bbR^{m{\times}n}$ with $\ba_i{\in}\bbR^{m\times1}$, the vectorization is given by $\ve(\bA){=}[\ba_1^T{\cdots}\ba_n^T]^T{\in}\bbR^{mn\times1}$. The inverse is ${\ive}({\ve}({\bA})){=}\bA$. 
When $\bA$ is a block diagonal matrix composed of $\bA_i$ for $i{\in}\bbN_n$, the block vectorization, $\ve_b(\bA)$, is $[\ve(\bA_1)^T{\cdots}\ve(\bA_n)^T]^T$, with inverse, $\ive_b({\cdot})$, following the same principle as the inverse of vectorization. Matrix multiplication of $\bA$, $\bB$, and $\bC$ satisfies $\ve(\textbf{ABC})=(\textbf{C}^T{\otimes}\bA)\ve(\textbf{B})$, where $\otimes$ is Kronecker product. Vectorization of the matrix transpose, $\ve(\bA^T)$, is equivalent to $\bP{\cdot}\ve(\bA)$, where $\bP$ is a permutation matrix reversing vector order.

\subsection{Graph Structure and Chordal Decomposition}

Graphs will be useful throughout this work to represent square matrix structures. 
Two types of graphs are considered, undirected graphs $\cG$ and directed graphs $\cD$.
An undirected graph, $\cG(\vertex{\cG}{,}\edge{\cG})$, is defined by its vertex set, $\vertex{\cG}{=}\bbN_N$, and edge set, $\edge{\cG}{\subseteq}\vertex{\cG}{\times}\vertex{\cG}$. 
In contrast, a directed graph, $\cD(\vertex{\cD}{,}\arc{\cD})$, is defined by its vertex set, $\vertex{\cD}{=}\bbN_N$, and arc set, $\arc{\cD}{\subseteq}\vertex{\cD}{\times}\vertex{\cD}$.

In a directed graph, each arc has a direction, meaning $(i,j){\in}\cA(\cD)$ does not imply $(j,i){\in}\arc{\cD}$, while edges $(i,j)\in\edge{\cG}$ have no directionality.
The vertices of an edge $(i,j)\in\edge{\cG}$ or an arc $(i,j)\in\arc{\cD}$ are referred to as the ends of the edge or the arc, respectively.
Vertices $i$ and $j$ of $(i,j){\in}\cE(\cG)$ are called neighbors. $\cN_\cG(i)$ denotes The set of neighbors of vertex $i$ in an undirected graph $\cG$.
If $(i,j){\in}\cA(\cD)$, vertex $i$ is said to dominate vertex $j$. In this case, vertex $i$ is the tail of the arc $(i,j)$ and vertex $j$ is its head.
The vertices that dominate a vertex $i$ are called its in-neighbors, while those dominated by the vertex are its out-neighbors.
These sets denoted by $\inneigh{\cD}{i}$ and $\outneigh{\cD}{i}$, respectively.

For any directed graph $\cD$, its underlying undirected graph $G(\cD)$ is defined on the same vertex set by replacing each arc with an edge having the same ends.
Conversely, any undirected graph $\cG$ can be converted into its associated directed graph $D(\cG)$ by replacing each edge with two oppositely oriented arcs connecting the same ends.
In this paper, the type of the graph will be specified when necessary.
Otherwise, a general graph $\cG$ is denoted as a graph with triple $(\cV(\cG),\cA(\cG),\cE(\cG))$, where $\cV(\cG)$ is the vertex set, $\cA(\cG)$ is the arc set, including both the arc from the directed and associated undirected parts, and $\cE(\cG)$ is the edge set of the undirected parts.



Set operations can be applied to graphs \cite{diestel2024graph}. For two graphs $\cG$ and $\cG'$, $\cG{\cap}\cG'{=}\big(\vertex{\cG}{\cap}\vertex{\cG'},\arc{\cG}{\cap}\arc{\cG'},\edge{\cG}\cap\edge{\cG'}\big)$. 
If $\cG{\cap}\cG'{=}\emptyset$, then $\cG$ and $\cG'$ are disjoint.
If $\vertex{\cG'}{\subseteq}\vertex{\cG}$, $\arc{\cG'}{\subseteq}\arc{\cG}$, and $\edge{\cG'}{\subseteq}\edge{\cG}$, then $\cG'$ is a subgraph of $\cG$.
$\cG$ is connected if any two vertices are linked by a sequence of arcs or edges from $\cG$, otherwise, disconnected.
A maximal connected subgraph of $\cG$ is a component of $\cG$. A disconnected graph can have multiple components.

For an undirected graph $\cG$, a clique, $\cC{\subseteq}\cG$, is a set of $i,j{\in}\cV(\cG)$ satisfying $i{\neq} j$, $(i,j){\in}\cE(\cG)$. 
It is represented as a maximal clique, $\cC_p$, if it is not a subgraph of another clique.
A cycle of length $\alpha$ is a set of pairwise distinct vertices $\set{v_1,{\dots},v_\alpha}{\subseteq}\cV(\cG)$ such that $(v_\alpha,1){\in}\cE(\cG)$ and $(v_i,v_{i+1})\in\cE(\cG)$ for $i\in\bbN_{\alpha-1}$. 
A chord is an edge connecting non-consecutive vertices within a cycle. 
An undirected graph $\cG$ is chordal if every cycle of length greater than three contains a chord.
For a directed graph $\cD$, concepts of cliques or chords are not defined.
Instead, a directed cycle of length $\alpha$ is a set of pairwise distinct vertices $\set{v_1,{\dots},v_\alpha}{\subseteq}\cV(\cG)$ such that $(v_\alpha,v_1){\in}\arc{\cD}$ and $(v_i,v_{i+1})\in\cA(\cG)$ for $i\in\bbN_{\alpha-1}$.

Breaking matrices into blocks is a key tool to link local and network-wide dynamics. The $(i,j)^\thh$ block of a matrix $\bA$ is $(\bA)_{i,j}$. If $(\bA)_{i,j}{\in}\bbR^{n_i{\times} m_j}$ and $\bA{\in}\bbR^{\sum_{i=1}^N n_i {\times} \sum_{j=1}^M m_j}$, then $\bA$ is said to be in $\bbR^{N{\times} M}$ block-wise. 
The block diagonal matrix formed by $\bA_i$ for all $i{\in}\bbX$ is $\diag(\bA_i)_{i\in{\bbX}}$. 
The row operator is defined as $\row(\bA_i)_{i{\in}\bbX}{=}[\bA_i\,{\cdots}]$, and the column operator as $\col(\bA_i)_{i{\in}\bbX}{=}[\bA_i^T\,{\cdots}]^T$, for $\bA_i$ of compatible dimensions.

A graph, $\cG$, is denoted with a script letter to indicate the (block-wise) structure of matrix $\bG$, represented in boldface using the same letter. 
Specifically, the block $(\bG)_{i,j} \neq \bzero$ if and only if $(i,j)\in\arc{\cG}$.
In this paper, this relationship is referred to as the graph $\cG$ of the matrix $\bG$.
Symmetric matrices are represented by undirected graphs.
Let $\bbS_-^n(\edge{\cG},0){=}\set*{\bG{\in}\bbS_-^n|(\bG)_{i,j}{=}\bzero\text{ if }(i,j){\notin}\edge{\cG}}$ be the set of negative semi-definite matrices structured according to an undirected graph $\cG$. The following theorem provides a useful tool for decomposing negative semi-definite block matrices.
\begin{theorem} [Chordal Block-Decomposition \cite{zheng2021chordal}] \label{thm:Chordal Decomposition}
    Let $\cZ$ be a chordal graph with maximal cliques $\set{\cC_p}_{i=1}^M$. Then, $\bZ\in\bbS_{-}^N(\edge{\cZ},0)$ (block-wise) if and only if there exist $\bZ_p\in\bbS_-^{|\vertex{\cC_p}|}$ (block-wise) for $p\in\bbN_M$ such that
    \begin{align} \label{eq:Chordal Decomposition}
        \bZ = \sum_{p=1}^M\bE_{\cC_p}^T\bZ_p\bE_{\cC_p},
    \end{align}
    where $\bE_{\cC_p}{\in}\bbR^{|\vertex{\cC_p}|\times N}$ (block-wise) is defined as $(\bE_{\cC_p})_{i,j}{=}\bI$ if $\cC_p(i){=}j$ and $(\bE_{\cC_p})_{i,j}{=}\bzero$ otherwise,
    and $\cC_p(i)$ is the $i^\thh$ vertex of $\cC_p$, sorted in natural ordering.
\end{theorem}

\subsection{\texorpdfstring{\gls{admm}}{ADMM}} \label{Chap:ADMM}

\gls{admm} can be used for distributed optimization\cite{boyd2011distributed}. Consider the constrained optimization problem,
\begin{align} \label{eq:Constrained Optimization Problem}
\begin{split}
    \min_{\bX}\quad&f(\bX) \qquad
    \text{s.t.}\quad\bX\in\Omega,
\end{split}
\end{align}
where $\bX\in\bbR^{n\times m}$ is the ``primal" variable, $f:\bbR^{n\times m}\to \bbR$ is the objective function, and $\Omega\subseteq\bbR^{n\times m}$ is the constraint set. 
\cref{eq:Constrained Optimization Problem} can be reformulated as
\begin{align} \label{eq:Equivalent Constrained Optimization Problem}
    \begin{split}
        \min_{\bX,\bZ}\quad&f(\bX)+I_\Omega(\bZ) \qquad
        \text{s.t.}\quad\bX-\bZ=0,
    \end{split}
\end{align}
with ``clone" variable, $\bZ\in\bbR^{n\times m}$. \gls{admm} solves \cref{eq:Equivalent Constrained Optimization Problem} by iteratively solving
\begin{subequations} \label{eq:General ADMM}
    \begin{align}
        \bX^{k+1}
            &{=}\arg\min_\bX\Big(f(\bX)
                +\frac{\rho}{2}\|\bX
                    -\bZ^k
                    +\bT^k\|_F^2\Big), 
                    \label{eq:General ADMM x Update} \\
        \bZ^{k+1}
            &{=}\arg\min_\bZ\Big(I_\Omega(\bZ)
                +\frac{\rho}{2}\|\bX^{k+1}
                    -\bZ
                    +\bT^k\|_F^2\Big) \nonumber \\ 
            &{=} \Pi_\Omega(\bX^{k+1}+\bT^k), 
                    \label{eq:General ADMM z Update} \\
        \bT^{k+1}
            &{=}\bT^k+(\bX^{k+1}
                -\bZ^{k+1})
                {=}\bT^k+\bR^{k+1}
                \label{eq:General ADMM u Update}
    \end{align}
\end{subequations}
where $\bT\in\bbR^{n\times m}$ is the ``dual" variable, $k\in\bbN$ is the iteration number, $\rho>0$ is the augmented Lagrangian parameter, $\Pi_\Omega:\bbR^{n\times m}\rightarrow \Omega$ is the projection operator, and $\bR^{k+1}=\bX^{k+1}-\bZ^{k+1}$ is the residual at iteration $k+1$ \cite{boyd2011distributed}.


The following theorem establishes the convergence of \gls{admm}.
\begin{theorem}[\cite{boyd2011distributed}]\label{thm:ADMM}
    Suppose that the following hold. 
    \begin{assumption} \label[ass]{ass:As 01}
        The function $f:\bbR^{n\times m}\to\bbR\cup\set*{+\infty}$ is closed, proper, and convex, and the constraint $\Omega$ is convex.
    \end{assumption}
    \begin{assumption} \label[ass]{ass:As 02}
        The Lagrangian $\cL$ of \cref{eq:Equivalent Constrained Optimization Problem} has a saddle point. Explicitly, there exist $(\bX^\star,\bZ^\star,\bT^\star)$, for which 
        \begin{align} \label{eq:saddle_point}
            \cL(\bX^\star,\bZ^\star,\bT)\leq\cL(\bX^\star,\bZ^\star,\bT^\star)\leq\cL(\bX,\bZ,\bT^\star)
        \end{align}
    holds for all $\bX$, $\bZ$, and $\bT$.
    \end{assumption}
    If \cref{eq:Equivalent Constrained Optimization Problem} is feasible, the \gls{admm} iterates satisfy the following:
    \begin{itemize}
        \item $\bR^k\to0$ as $k\to\infty$, i.e., the iterates approach feasibility.
        \item $f(\bX^k)+I_\Omega(\bZ^k)\to f(\bX^\star)$ as $k\to\infty$, where $\bX^\star$ is a primal optimal point, i.e., the objective function of the iterates approaches the optimal value.
        \item $\bT^k\to\bT^\star$ as $k\to\infty$, where $\bT^\star$ is a dual optimal point.
    \end{itemize}
\end{theorem}

\cref{ass:As 01} can be verified using the following lemma.
\begin{lemma} [\cite{boyd2011distributed}] \label{lem:closed}
    The function $f{:}\Omega{\to}\bbR$ satisfies \cref{ass:As 01} if and only if its epigraph, $\text{epi}f{=}\set{(x,t){\in}\Omega{\times}\bbR{|}f(x)\leq t}$,
    is a closed, nonempty, convex set.
\end{lemma}

\subsection{QSR-Dissipativity of Large-Scale, Multi-Agent Systems}

QSR-dissipativity, defined below, quantifies a relationship between system inputs and outputs.
\begin{definition} [QSR-Dissipativity \cite{vidyasagar1981input}] \label{def:QSR}
    Let $\bQ\in\bbS^{l}$, $\bR\in\bbS^{m}$, $\bS\in\bbR^{l\times m}$. The system $\mathscr{G}:\cL_{2e}^m\to\cL_{2e}^l$ is \textit{QSR-dissipative} if there exists $\beta\in\bbR$ such that for all $\bu\in\cL_{2e}^m$ and $T>0$,
    \begin{align} \label{eq:QSR Dissipativity}
        \int_0^T\begin{bmatrix}
            \mathscr{G}^T(\bu(t)) & \bu^T(t)
        \end{bmatrix}
        \begin{bmatrix}
            \bQ & \bS \\ \bS^T & \bR
        \end{bmatrix}
        \begin{bmatrix}
            \mathscr{G}(\bu(t)) \\ \bu(t)
        \end{bmatrix}dt\geq\beta.
    \end{align}
\end{definition}

For \gls{lti} systems, \cref{lem:KYP Lemma} can be used to prove the QSR-dissipativity of the system.
\begin{lemma}[Dissipativity Lemma \cite{gupta1996robust}] \label{lem:KYP Lemma}
    Consider an \gls{lti} system with minimal state-space realization $\Sigma{:}\dbx{=}\bA\bx{+}\bB\bu$, $\by{=}\bC\bx{+}\bD\bu$. The system is QSR-dissipative if there exists a matrix $\bP{\succ}0$ and matrices $\bQ$, $\bS$, $\bR$ such that
    \setlength{\arraycolsep}{2pt}
    \begin{align} \label{eqn:KYP Lemma}
        \begin{bmatrix}
            \bA^T\bP{+}\bP\bA{-}\bC^T\bQ\bC & \bP\bB{-}\bC^T\bS{-}\bC^T\bQ\bD \\
            \bB^T\bP{-}\bS^T\bC{-}\bD^T\bQ\bC & {-}\bR{-}\bS^T\bD{-}\bD^T\bS{-}\bD^T\bQ\bD
        \end{bmatrix}\preceq0.
    \end{align}
\end{lemma}

QSR-dissiaptivity is useful for ensuring $\cL_2$-stability, defined below, 
\begin{definition}[$\cL_2$-stability \cite{vidyasagar1981input}] \label{def:L2 Stable}
    An operator $\mathscr{H}{:}\cL_{2e}{\mapsto}\cL_{2e}$ is $\cL_2$-stable if there exists a constant $\gamma{>}0$ and $\beta {\in} \bbR$ such that for all $\bu{\in}\cL_2$ and $T{>}0$,
    \begin{align}
        \|(\mathscr{H}\bu)_T\|_2\leq\gamma\|\bu_T\|_2+\beta.
    \end{align}
\end{definition}

For multi-agent systems, \gls{vndt} relates the dissipativity of each agent to the $\cL_2$ stability of the entire system.

\begin{theorem}[\gls{vndt} \cite{vidyasagar1981input}] \label{thm:VNDT}
    Consider $N$ agents, $\mathscr{G}_i{:}\cL^{m_i}_{2e}{\mapsto}\cL^{l_i}_{2e}$, where $\by_i{=}\mathscr{G}_i\be_i$ and $\be_i{=}\bu_i{+}\sum_{j=1}^N(\bH)_{i,j}\by_j$ for $i{\in}\bbN_N$, where $\bu_i{\in}\cL_{2e}^{m_i}$ and $\bH{\in}\bbR^{m_i{\times} l_i}$ satisfying $(\bH)_{i,i}{=}\bzero$ for all $i{\in}\bbN_N$. Further, suppose each agent $\mathscr{G}_i$ is $\bQ_i\bS_i\bR_i$-dissipative, let $\bu{=}\col(\bu_i)_{i\in\bbN_N}$, and $\by{=}\col(\by_i)_{i\in\bbN_N}$. Then the multiagent system $\mathscr{G}{:}\bu{\mapsto}\by$ is $\cL_2$ stable if $\barbQ(\bX){\prec} 0$, where
    \begin{align}\label{eq:barbQ}
        \barbQ(\bX)&=\bQ+\bS\bH+\bH^T\bS^T+\bH^T\bR\bH,
    \end{align}
    with $\bX{=}\diag(\bX)_{i\in\bbN_N}$, $\bX_i{=}\diag(\bQ_i,\bS_i,\bR_i)$, $\bQ{=}\diag(\bQ_i)_{i\in\bbN_N}$, $\bR{=}\diag(\bR_i)_{i\in\bbN_N}$, and $\bS{=}\diag(\bS_i)_{i\in\bbN_N}$.
\end{theorem}

\subsection{Network Matrix Calculations}

We will later break $\barbQ$ from \cref{thm:VNDT} into blocks related to the network interconnections in $\bH$. The following two theorems establish useful relationships for this.

\begin{lemma} \label{thm:barbQ}
    In \cref{thm:VNDT}, $\barbQ(\bX)$ is composed of blocks associated with the $N$ agents given by
    \begin{align} 
        (\barbQ)_{i,i}&=\bQ_i{+}\hspace{-8
        pt}\sum_{k\in\outneigh{\cH}{i}}\hspace{-8
        pt}(\bH)_{k,i}^T\bR_k(\bH)_{k,i}\text{, and} \label{eq:barbQ_ii} \\
        (\barbQ)_{i,j}&=\bS_i(\bH)_{i,j}
            {+}(\bH)_{j,i}^T\bS_j^T{+}\hspace{-10pt}\sum_{k{\in}\cN_\cH^+(i,j)}\hspace{-10
        pt}(\bH)_{k,i}^T\bR_k(\bH)_{k,j}, \label{eq:barbQ_ij}
    \end{align}
    for all $i,j{\in}\bbN_N$ with $i{\neq}j$, where $\cN_\cH^+(i,j){=}\outneigh{\cH}{i}{\cap}\outneigh{\cH}{j}$.
\end{lemma}
\begin{proof}
    To establish \cref{eq:barbQ_ii}, multiplying out the terms in \cref{eq:barbQ} yields
    \begin{align*}
        (\barbQ)_{i,i}&=\bQ_i {+} \bS_i(\bH)_{i,i}{+} (\bH)_{i,i}^T\bS_i^T {+} \sum_{k\in\bbN_N}(\bH)_{k,i}^T\bR_k(\bH)_{k,i}.
    \end{align*}
    Recalling that $(\bH)_{k,i}{=}\bzero$, if $k{\notin}\outneigh{\cH}{i}$, and in particular $i\notin \outneigh{\cH}{i}$, gives \cref{eq:barbQ_ii}.

    To establish \cref{eq:barbQ_ij}, recalling that $\bQ$ is block-diagonal, and multiplying out the terms in \cref{eq:barbQ} yields
    \begin{align*}
        (\barbQ)_{i,j}&=
        \bS_i(\bH)_{i,j}{+} (\bH)_{j,i}^T\bS_j^T {+} \sum_{k\in\bbN_N}(\bH)_{k,i}^T\bR_k(\bH)_{k,j}.
    \end{align*}
    Recalling that $(\bH)_{k,\ell}{=}\bzero$, if $k{\notin}\outneigh{\cH}{\ell}$ for $l{=}i$ or $l{=}j$, gives \cref{eq:barbQ_ij}.
\end{proof}

\begin{lemma} \label{thm:supergraph}
    Under the conditions in \cref{thm:barbQ}, 
    let $\barcQ$ denote the graph of $\barbQ(\bX)$. Then, $\barcQ$ is an undirected supergraph of $\cH$, meaning $\barcQ\supset\cH$, whose edges $\edge{\barcQ}$ are
    \begin{align} \label{eq:Edge_barQ}
    \begin{split}
        \edge{\barcQ}&{=}\edge{G(\cH)}{+}\set{(i,i)|i{\in}\vertex{\cH}} \\
        &\quad{+}\set{(i,j)|i,j{\in}\vertex{\cH},\cN_\cH^+(i,j){\neq}{\emptyset}}.
    \end{split}
    \end{align}
    where $G(\cH)$ is the underlying undirected graph of $\cH$.
\end{lemma}

\begin{proof}
  %
    All statements in this theorem follow from the structure of \cref{eq:barbQ_ii,eq:barbQ_ij}. Most straightforwardly, $\barcQ$ is undirected due to the symmetricity of $\barbQ(\bX)$. To see that $\barcQ{\supset}\cH$, first note that \cref{eq:barbQ_ii,eq:barbQ_ij} do not introduce any new agents, so $\vertex{\cH}{=}\vertex{\barcQ}$. Due to $\edge\cH{\subseteq}\edge{G(\cH)}$, $\cH{\subset}\barcQ$ will follow from \cref{eq:Edge_barQ}, which is the only thing that remains to be established.
    


    \cref{eq:Edge_barQ} can be constructed by finding what elements in $\edge{\barcQ}$ correspond to the variables in \cref{eq:barbQ_ii,eq:barbQ_ij}.
    The term $\bQ_i$ in \cref{eq:barbQ_ii} corresponds to the self-loop $(i,i){\in}\edge{\barcQ}$.
    Terms $\bS_i(\bH)_{i,j}$ or $(\bH)_{j,i}^T\bS_j^T$ in \cref{eq:barbQ_ij} imply that $(i,j){\in}\edge{G(\cH)}$ contributes to $\edge{\barcQ}$.
    The remaining terms in \cref{eq:barbQ_ii,eq:barbQ_ij} are nonzero only when $\cN_\cH^+(i,j){\neq}\emptyset$.
    Combining all contributions yields \cref{eq:Edge_barQ}.
    Therefore, the last statement is also true, completing the proof.
    \cref{fig:Augmented Graph} shows the relationship between the edge sets of $\cH$ and $\barcQ$.
\end{proof}

\begin{figure}
    \centering
    \includegraphics[width = 0.35\textwidth]{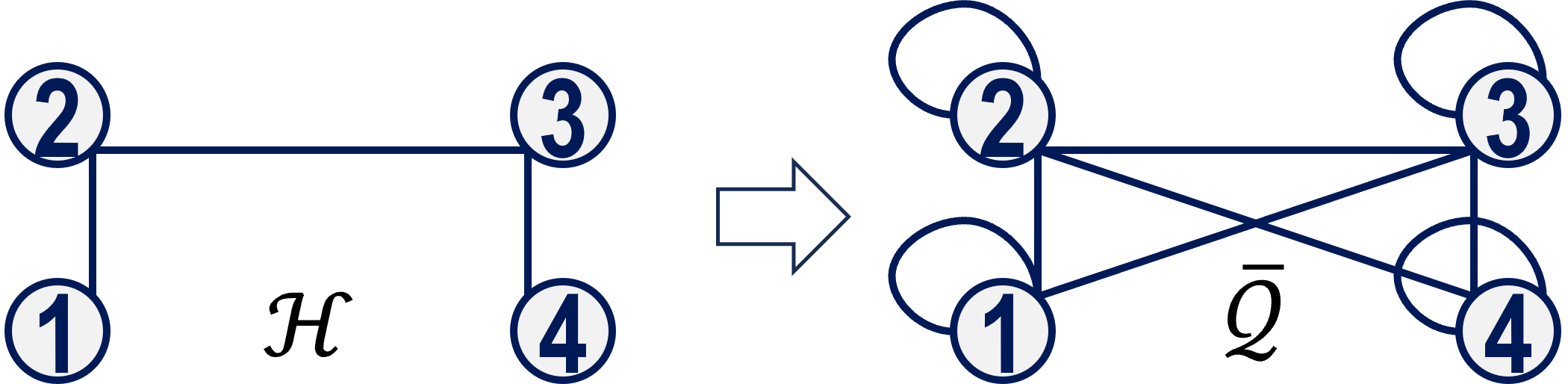}
    \caption{Example of graph $\cH$ and $\barcQ$; $\cH\subset\barcQ$.}
    \label{fig:Augmented Graph}
    \vspace*{-1.25\baselineskip} 
\end{figure}

\subsection{Complexity of Interior Point Methods}

Since all optimization problems in this work are formulated as a quadratic programs with \gls{lmi} constraints, it is helpful to compare their estimated computational complexity. Interior point methods are widely employed to solve \glspl{sdp}, and have complexity  
\begin{align} \label{eq:interor_point_method}
    \cO\brackets{\overline{n}^2\sum_{i=1}^{\overline{N}}\overline{m}_i^2+\overline{n}\sum_{i=1}^{\overline{N}}\overline{m}_i^3+\overline{n}^3},
\end{align}
where $\overline{n}$ is the size of the variables, $\overline{N}$ is the number of uncoupled \gls{lmi} constraints, and $\overline{m}_i$ is the row sizes of $i^\thh$ \gls{lmi} constraint \cite[Sec. 6.4.1]{nesterov1994interior}. 

\section{PROBLEM STATEMENT} \label{Chap:Problem Statement}
Consider $N$ agents, $\mathscr{G}_i{:}\cL^{m_i}_{2e}{\rightarrow}\cL^{l_i}_{2e}$, interconnected as
\begin{align}\label{eq:Interconnected Subsystem}
    \by_i=\mathscr{G}_i\be_i, \qquad \be_i=\bu_i+\sum_{j=1}^N(\bH)_{i,j}\by_j,
\end{align} 
where $\bu_i\in\cL_{2e}^{m_i}$, $(\bH)_{i,i}=\bzero$, and $(\bH)_{i,j}\in\bbR^{m_i\times l_j}$. Let $\bu = \col(\bu_i)_{i\in\bbN_N}$, $\be = \col(\be_i)_{i\in\bbN_N}$, and $\by = \col(\by_i)_{i\in\bbN_N}$. 
Then, the multiagent system $\mathscr{G}{:}\bu{\mapsto}\by$ is expressed as
\begin{align}\label{eq:Interconnected system}
    \begin{split}
        \by=\mathscr{G}\bu, \qquad \be=\bu+\bH\by, 
    \end{split}
\end{align}
where $\be$ is the interconnection signal, $\by$ is the output, and $\bu$ is the exogenous input. The ``interconnection matrix'' $\bH$ encodes the network structure. If each $\mathscr{G}_i$ is $\bQ_i\bS_i\bR_i$-dissipative, then by \Cref{thm:VNDT} the system $\mathscr{G}$ is $\cL_2$ stable if $\barbQ(\bX){\prec}0$, where $\barbQ(\bX)$ is defined in \cref{eq:barbQ}. 
The $\cL_2$ stability can be verified by solving the feasibility problem
\begin{subequations} \label{opt:Main Problem}
    \begin{align}
        \text{Find}\quad&\bX_i \quad i\in\bbN_N,\label{opt:Main Problem: Objective} \\
        \text{s.t.}\quad&
        \bX_i\in\bbP_i\text{, and} \label{opt:Main Problem: KYP}  \\
            &\bX\in\bbQ, \label{opt:Main Problem: Stability} 
    \end{align}
\end{subequations}
where $\bX=\diag(\bX_i)_{i\in\bbN_N}$, $\bX_i = \diag(\bQ_i,\bS_i,\bR_i)$, $\bbP_i=\{\bX_i\;|\;\text{\cref{eq:QSR Dissipativity} holds with }\bQ_i,\bS_i,\bR_i\}$, and  $\bbQ{=}\{\bX\;|\;\barbQ(\bX)\prec0\}$. The set $\bbP_i$ in \cref{opt:Main Problem: KYP} depends on the dynamics of the $i^\thh$ agent. For instance, if $\mathscr{G}_i$ is an \gls{lti} system, then $\bbP_i=\{\bX_i\;|\;\exists\bP_i\succ0$ such that \cref{lem:KYP Lemma} holds $\}$.

A simple, but computationally prohibitive, approach is to construct a state-space realization for the entire network, and then apply the Dissipativity Lemmas in \cite{gupta1996robust}. For $N$ agents with $n$ inputs and outputs each, $\bQ,\bS,\bR$ contain $\mathcal{O}((Nn)^2)$ elements, creating $\barn{\in}\cO((Nn)^2)$, $\barN{=}N{+}1$, $\barm_{i{\in}\bbN_N}{\in}\cO(n)$, and $\barm_{N+1}{\in}\cO(Nn)$.
This results in a computational complexity of $\cO((Nn)^6)$ according to \cref{eq:interor_point_method}.
An alternative is to apply a scattering transformation to recast the system into a specific dissipativity structure, such as passivity, which facilitates the analysis \cite{usova2018scattering}. 
In contrast, solving \cref{opt:Main Problem} offers key advantages. The number of variables in each local matrix $\bQ_i,$ $\bR_i$, and $\bS_i$ is $\mathcal{O}(n^2)$. 
As a result, it achieves a reduced complexity of $\cO(N^4n^6)$ with $\barn{\in} N\cO(n^2)$, $\barN{=}N{+}1$, $\barm_{i{\in}\bbN_N}{\in}\cO(n)$, and $\barm_{N+1}{\in}\cO(Nn)$, and does not require complex transformations of the original system.

The simplest way to solve \cref{opt:Main Problem} is to find $\bX_i$ satisfying \cref{opt:Main Problem: KYP} for each agent, then check if \cref{opt:Main Problem: Stability} holds. However, this is sub-optimal because a dissipative system satisfies \cref{eq:QSR Dissipativity} with various $(\bQ,\bS,\bR)$ triplets. For example, consider two agents in negative feedback, $\mathscr{G}_1$ and $\mathscr{G}_2$, that are $(-1,\frac{1}{2},-2)$-dissipative and $(-1,\frac{13}{32},-\frac{21}{128})$-dissipative, respectively. Under these dissipativity characterizations, the closed-loop does not satsify \Cref{thm:VNDT}. However, $\mathscr{G}_1$ is also $(\frac{1}{10},-\frac{1}{4},\frac{3}{5})$-dissipative\footnote{The first characterization is in the interior conic sector with bounds $(-1,2)$, while the second is in the exterior conic sector with bounds $(2,3)$. Any operator satisfying the former must satisfy the latter \cite{7286765}.}. With this new description, the closed-loop satisfies \Cref{thm:VNDT}. Finding compatible $(\bQ,\bS,\bR)$ becomes more complicated with more agents. Hence, co-optimizing agents' dissipativity by \cref{opt:Main Problem} is less conservative than checking $\barbQ(\bX){\prec}0$ afterward.


Nonetheless, solving \cref{opt:Main Problem} directly has potential drawbacks. First, it requires agents' dynamics information, such as $(\bA_i,\bB_i,\bC_i,\bD_i)$ for \gls{lti} systems, which may be unacceptable due to intellectual property or cybersecurity concerns. Second, it is still computationally expensive when there is a large number of agents, $N$, due to $\cO(N^4n^6)$ \cite{nesterov1994interior}.

The next two sections develop \cref{alg:01,alg:02} to solve \cref{opt:Main Problem} in a distributed manner. In \cref{Chap:ISNSA}, \cref{alg:01} achieves information security by allowing each agent to calculate its dissipativity parameters independently without revealing its dynamics. These parameters are shared and iteratively adjusted to satisfy \gls{vndt}. However, \cref{alg:01} requires solving a problem with a constraint whose dimension scales with the number of agents. In addition, all agents must share their dissipativity parameters with each other or a centralized computer, posing bandwidth concerns.

To remedy the limitations of \cref{alg:01}, in \cref{Chap:Chordal_Decomposition} \cref{alg:02} applies chordal decomposition to \gls{vndt}, splitting a large \gls{lmi} into smaller \gls{lmi} and equality constraints. \cref{alg:02} allows fully distributed stability analysis of multi-agent systems, enabling each agent to share its dissipativity parameters only with designated connected agents. Furthermore, it reduces computational complexity by decreasing the largest constraint size and providing an exact solution to the optimization problem with an equality constraint.

\section{Information-Secure Network Stability Analysis} \label{Chap:ISNSA}
As in \cref{eq:Constrained Optimization Problem,eq:General ADMM}, \cref{opt:Main Problem} can be solved using \gls{admm} in three steps with the iterations
\begin{subequations} \label{eq:Distributed ADMM}
    \begin{align}
        \bX_i^{k+1}&{=}\Pi_{\bbP_i}(\bZ_i^k{-}\bT_i^k), \quad
            i{\in}\bbN_N, \label{eq:X proj}\\
        \bZ^{k+1}&{=}\Pi_{\bbQ}(\bX^{k+1}{+}\bT^k), \label{eq:Z proj} \\ 
        \bT_i^{k+1}&{=}\bT_i^k{+}(\bX_i^{k+1}{-}\bZ_i^{k+1}){=}\bT_i^k{+}\bR_i^{k+1}{,} \;
            i{\in}\bbN_N, \label{eq:T min}
    \end{align}
\end{subequations}
where $\bX$, $\bZ=\diag(\bZ_i)_{i\in\bbN_N}$, and $\bT=\diag(\bT_i)_{i\in\bbN_N}$ act as the primal, clone, and dual variable, respectively, and $\bR_i^{k+1}$ is the residual of $i^\thh$ agent at iteration $k+1$.

\cref{alg:01} describes the iterative optimization process using \cref{eq:Distributed ADMM}. The sequence starts with $\bX^0{=}\bZ^0$ and $\bT^0{=}\bzero$. The initial point $\bX_i^0$ need not be feasible, but the choice strongly influences the number of iterations required for convergence. A natural choice is $\bX^0_i{=}a_i\bI$ for all $i{\in}\bbN_N$, where $a_i{\in}\bbR$ is a weighting constant. 
Since the purpose of this algorithm is to find a feasible point rather than an optimal one, $\barbQ(\bX){\prec}0$ acts as the stopping criterion. 

\subsection{Convergence Criteria}
The following theorem demonstrates that \cref{eq:Distributed ADMM} iteratively converges to the feasible solution of \cref{opt:Main Problem} where the primal and clone variables equal one-another.
\begin{theorem} \label{thm:isnsa_admm}
    Suppose that \cref{opt:Main Problem} has a feasible solution. If $\bbP_i$ and $\bbQ$ are closed and convex, then the \gls{admm} iterates defined in \cref{eq:X proj,eq:Z proj,eq:T min} guarantee that
    \begin{align}\label{eq:ADMM_conv}
        \lim_{k\to\infty}\bR^k=\bzero,\quad\lim_{k\to\infty}\bT_i^k=\bT_i^\star,\quad\forall i\in\bbN_N
    \end{align}
    where $\bT_i^\star$ is the dual optimal point of $i^\thh$ agent.
\end{theorem}

\begin{proof}
    To begin, we must get \cref{opt:Main Problem} in the same format as \cref{eq:Equivalent Constrained Optimization Problem}, which is
    \begin{subequations} \label{opt:lemma_reformulation}
        \begin{align} 
            \arg\min_{\bX_i,\bZ_i}\quad&\sum_{i\in\bbN_N}I_{\bbP_i}(\bX_i)+I_{\bbQ}(\bZ), \label{eq:lemma_reformulation_obj} \\
            \sthat\quad&\bX_i-\bZ_i=\bzero\quad i\in\bbN_N. \label{eq:lemma_reformulation_const}
        \end{align}
    \end{subequations}
    The equivalence holds because the objective of \cref{opt:lemma_reformulation} is smallest when all indicator functions equal zero. This means its solutions, $\bX_i{=}\bZ_i$, satisfy \cref{opt:Main Problem}. Likewise, all solutions to \cref{opt:Main Problem} render the indicator functions in \cref{eq:lemma_reformulation_obj} zero when setting $\bZ_i{=}\bX_i$. Now we must verify the assumptions of \cref{thm:ADMM} with $f=I_{\bbP_i}$ and $\Omega=\bbQ$.

    For any $I_\bbA(\bx)$, if $\bbA$ is closed, nonempty, and convex, then $\epi I_\bbA(\bx)$ is a closed convex set. By the assumptions of \cref{thm:isnsa_admm}, $\bbQ$ is convex and $\epi f{=}\epi I_{\bbP_i}$ is closed, convex, and nonempty, so \cref{lem:closed} implies that \cref{ass:As 01} holds. 

    Since \cref{eq:lemma_reformulation_obj} is a closed, proper, and convex function and \cref{eq:lemma_reformulation_const} is linear, \cref{opt:lemma_reformulation} is a convex optimization problem with linear equality constraints. By Slater's condition, the Lagrangian of \cref{opt:lemma_reformulation} has a saddle point \cite{boyd2004convex}, so \cref{ass:As 02} is satisfied.

    It has now been established that \cref{thm:ADMM} holds with $f{=}I_{\bbP_i}$ and $\Omega{=}\bbQ$, which implies \cref{eq:ADMM_conv}.
\end{proof}

From \cref{thm:isnsa_admm}, the algorithm requires closed and convex constraint sets to converge. If both $\bbP_i$ for all $i{\in}\bbN_N$ and $\bbQ$ consist of \glspl{lmi}, they are convex. While only non-strict \glspl{lmi} are closed, in practice, any strict inequality in \gls{lmi} can be replaced by a non-strict inequality by introducing any positive constant. For instance, imposing $\barbQ(\bX){+}\epsilon\bI{\preceq}0$ for some $\epsilon{>}0$ instead of 
$\barbQ(\bX){\prec}0$. 
Under these conditions, \cref{alg:01} converges to a feasible point of \cref{opt:Main Problem} if one exists, confirming network stability.

\begin{algorithm}[tbp]
    \caption{Information-secure network stability analysis}\label{alg:01}
    \begin{algorithmic}[1]
        \Require Max Iterstions$,\bX_i^0$ for $i\in\bbN_N$
        \Ensure $\bX^k$
        \State Initialize $k=0$, $\bX^0=\diag(\bX_i^0)_{i\in\bbN_N}$, $\bZ^0=\bX^0$, and $\bT^0=\textbf{0}$
        \While {$\barbQ(\bX)\nprec0,k<$Max Iterations}
            \State $k\gets k+1$
            \State Find $\bX_i^k$ by \cref{eq:X proj} in parallel
            \State Find $\bZ^k$ by \cref{eq:Z proj} at a centralized node
            \State Find $\bT_i^k$ by \cref{eq:T min} in parallel
        \EndWhile
        \If {$\barbQ(\bX)\prec0$}
            \State Multi-agent system is stable with $\bX^k=\diag(\bX_i^0)_{i\in\bbN_N}$
        \EndIf
    \end{algorithmic}
\end{algorithm}

\subsection{Computation Time}
\cref{alg:01}'s key advantage over directly solving \cref{opt:Main Problem} is that each agent can verify its own dissipativity without sharing its dynamics information, even with neighboring nodes.
Although \gls{admm} is generally known to converge slowly\cite{boyd2011distributed}, \cref{alg:01} can converge faster than standard \gls{admm}, since \cref{alg:01} focuses on finding a feasible $\bX$ satisfying \cref{opt:Main Problem} in a distributed manner, rather than strictly ensuring \cref{eq:lemma_reformulation_const}.
Thus, it can terminate as soon as \cref{opt:Main Problem} is solved, hastening convergence.


The computation time of $k^\thh$ iteration is $t_{k}{=}t_k^{\bX}{+}t_k^{\bZ}{+}t_k^{\bT}$, where $t_k^{\bX}{=}\max_{i{\in}\bbN_N}(t_k^{\bX_i})$ and $t_k^{\bT}{=}\max_{i\in\bbN_N}(t_k^{\bT_i})$ are the maximum time required for the projections in \cref{eq:X proj} and the $\bT_i$ updates in \cref{eq:T min}, respectively, and $t_k^{\bZ}$ is the time required for the projection in \cref{eq:Z proj}. In multi-agent systems, $t_k^Z$ dominates the overall computational cost, as the size of \cref{opt:Main Problem: Stability} scales with the network size, $N$, while the size of \cref{opt:Main Problem: KYP} scales with the each subsystem size, $n_i$, which is typically much smaller than $N$.
The complexity of solving \cref{eq:Z proj} via interior point methods scales as $\cO(N^4n_m^6)$, where $n_m$ is the maximum row size of matrix variables among agents. Therefore, $t_k{\approx} t_k^Z {\sim} \cO(N^4n_m^6)$.
This makes solving \cref{eq:Z proj} the main bottleneck, as it has high complexity and requires all agents to send their dissipativity parameters to a centralized node.
This bottleneck can be eliminated by decomposing the structure of $\barbQ(\bX)$, as demonstrated next.

\section{Chordal Decomposition of \texorpdfstring{\gls{vndt}}{VNDT}} \label{Chap:Chordal_Decomposition}

This section decomposes the large-scale \gls{lmi} in \cref{opt:Main Problem: Stability} into smaller expressions by dividing associated graphs into cliques. First, the problem is reformulated for distributed stability analysis via chordal decomposition and \gls{admm}. Then, convergence conditions are given, followed by the algorithm for distributed analysis. Finally, conditions are found where the algorithm can be used without computationally costly chordal extensions, and where certain constraints are separable, allowing efficient computations.

\cref{tab:Chordal Decomposition} summarizes the notation arising from dividing graphs into cliques.


\subsection{Problem Reformulation} \label{subChap:Problem_Reformulation}
To begin, \cref{subs:chord1} divides \gls{vndt} based on the network's clique structure. Then \cref{subsubChap:Parallel_Computing_of_LMI} shows that the resulting subdivided \glspl{lmi} can be solved in parallel. Finally, \cref{subs:equiv} reformulates the problem for distributed computations exploiting structure.

\subsubsection{Chordal Structure and LMIs}\label{subs:chord1}

The next theorem shows that if the graph $\barcQ$ is chordal, network-wide stability can be established by applying \cref{thm:Chordal Decomposition}, which results in smaller \glspl{lmi} and a matrix equation.

\begin{table}
\caption{Terms for decomposition of \cref{eq:barbQ}}
\label{tab:Chordal Decomposition}
\begin{center}
\renewcommand{\arraystretch}{1.3}
\begin{tabular}{p{0.6cm}|p{7cm}}
\hline
Terms & Definition \\
\hline\hline
$\bG$, $\cG$ & Script letters denote the graph corresponding to the (block-wise) structure of a matrix named with the same letter in bold.\\
\hline
$\barcQ$ & The graph corresponding to the matrix $\barbQ(\bX)$\\
\hline
$\cH$ & The graph corresponding to the interconnection matrix $\barbH$\\
\hline
$M$ & The number of maximal cliques of $\barcQ$ \\
\hline
$\barcQ_o$ & 
\cref{eq:barcQ_o}; $\edge{\barcQ_o}$ is the overlapped edge from \cref{thm:Chordal Decomposition}. \\
\hline
$L$ & $|\edge{\barcQ_o}|$; the number of overlapped edges from \cref{thm:Chordal Decomposition}. \\
\hline
$\barbX_p$ & \cref{eq:barbX_p}; the block diagonal matrix defined from block diagonal components in $\textbf{X}$ which are used to calculate $(\barbQ)_{i,j}$ in \cref{eq:barbQ_ii,eq:barbQ_ij}, where $(i,j)\in\edge{\cC_p}-\edge{\barcQ_o}$. \\
\hline
$\vertex{\cR_p}$ & 
\cref{eq:cV(cR_p)}; the set of vertices (agents) whose $\bR$ matrix is used to calculate $(\barbQ)_{i,j}$ for $(i,j)\in\edge{\cC_p-\barcQ_o}$
\\ \hline
$\bY_{i,j}^{p}$ & The matrix variables defined from the overlapped position in \cref{eq:Chordal Decomposition}, where $p{\in}\bbN_M$ is the index of maximal cliques and $(i,j){\in}\edge{\barcQ_o}$. Thus, leading to \cref{Qhatdef,eq:bbL_ij} \\
\hline
$\barbY_p$ & \cref{eq:barbY_p}; the block diagonal matrix with block diagonal components in $\bY$, using vertices in $C_p(\barcQ)$. \\
\hline
$\hatbX$ & \cref{eq:hatbX}; the block diagonal matrix defined from dissipativity matrices used to calculate $(\barbQ)_{i,j}$ for $(i,j)\in\edge{\barcQ_o}$. \\
\hline
$\vertex{\cR}$ & 
\cref{eq:cV(cR)}; the set of vertices (agents) whose $\bR$ matrix is used to calculate $(\barbQ)_{i,j}$ for $(i,j)\in\edge{\barcQ_o}$
\\ \hline
$\bY$ & \cref{eq:bY}; the block diagonal matrix defined from all overlapped variables $\bY_{i,j}^p$. \\
\hline
$\bbL_{i,j}$ & \cref{eq:bbL_ij}; the set of maximal clique index whose maximal clique share a common edge $(i,j)$ with $\barcQ_o$\\
\hline
\end{tabular}
\end{center}
\vspace*{-2\baselineskip} 
\end{table}

\begin{theorem} \label{thm:VNDT_Chordal_Decomposition}
    Consider $N$ agents, $\mathscr{G}_i:\cL^{m_i}_{2e}\rightarrow\cL^{l_i}_{2e}$, satisfying the assumptions of \cref{thm:VNDT}, and let $\bu{=}\col(\bu_i)_{i\in\bbN_N}$ and $\by{=}\col(\by)_{i\in\bbN_N}$.
    Furthermore, assume that $\barcQ$ is a chordal graph with maximal cliques $\set{\cC_p}_{p=1}^M$, where $M$ is the number of maximal cliques of $\barcQ$, and define a graph of clique overlaps,
    \begin{align} \label{eq:barcQ_o}
        \barcQ_o=\bigcup_{p,q\in\bbN_M}\cC_p\cap\cC_q.
    \end{align}
    Let $\cC_p(k)$ denote the $k^\thh$ vertex of $\cC_p$, sorted in natural ordering as defined in \cref{thm:Chordal Decomposition}. Then the multi-agent system $\mathscr{G}:\bu\mapsto\by$ is $\cL_2$ stable if there exists $\epsilon>0$ and $\bY_{i,j}^p\in\bbR^{l_i\times l_j}\ \forall p\in\bbN_M$ $\forall(i,j)\in\edge{\barcQ_o}$ such that
    \begin{align}
        \barbQ_p&\in\bbS_-^{\abs{\vertex{\cC_p}}},\; \forall p\in\bbN_M,  \label{eq:barbQ_p}\\
        \hatbQ&=\bzero\in\bbR^{N\times N},\text{(block-wise)}  \label{eq:hatbQ}
    \end{align}
    where
    \begin{align}
        (&\barbQ_p)_{k,l}{=}\nonumber\\
        &\left\{\begin{array}{ll}
            \hspace{-7pt}(\barbQ(\bX) {+} \epsilon\bI)_{\cC_p({k}){,}\cC_p({l})}, & \hspace{-5pt}\text{if }(\cC_p(k){,}\cC_p(l)){\in}\edge{\cC_p{\setminus}\barcQ_o}, \\
            \hspace{-3pt}\bY_{\cC_p(k){,}\cC_p(l)}^p, & \hspace{-5pt}\text{if }(\cC_p(k){,}\cC_p(l)){\in}\edge{\cC_p{\cap}\barcQ_o},\hspace{-5pt}
        \end{array}\right.\hspace{-10pt} \label{QpDef}\\
        %
        (&\hatbQ)_{i,j}{=}
        \left\{\begin{array}{ll}
            \hspace{-7pt}(\barbQ(\bX){+}\epsilon\bI)_{i,j}{-}\sum_{p\in\bbL_{i,j}}\hspace{-5pt}\bY_{i,j}^p  &\text{if }(i,j){\in}\cE(\barcQ_o), \\
            \hspace{-3pt}\bzero  &\text{if }(i,j){\notin}\cE(\barcQ_o),
        \end{array}\right.\hspace{-10pt} \label{Qhatdef} \\
        &\bbL_{i,j}=\set{p\in\bbN_M\,|\,(i,j)\in\edge{\barcQ_o\cap\cC_p}}. \label{eq:bbL_ij}
    \end{align}
\end{theorem}

\begin{proof}
    The proof proceeds by applying \gls{vndt} and \cref{thm:VNDT}, guaranteeing the system is $\cL_2$ stable if $\barbQ(\bX){\prec}0$. This \gls{lmi} holds if and only if there exists $\epsilon{>}0$ such that 
    \begin{align} \label{eq:VNDT_nonstrict}
        \barbQ(\bX)+\epsilon\bI\preceq0.
    \end{align}
    Since $\barcQ$ is chordal and $\barbQ(\bX){+}\epsilon\bI$ is negative semi-definite, \cref{thm:Chordal Decomposition} can be applied to decompose \cref{eq:VNDT_nonstrict} into \cref{eq:barbQ_p,eq:hatbQ}. 
    The remainder of the proof provides a detailed construction of this decomposition.

    \cref{thm:Chordal Decomposition} shows that \cref{eq:VNDT_nonstrict} is equivalent to the existence of $\barbQ_p{\in}\bbS_-^{\abs{\vertex{\cC_p}}}$ for all $p{\in}\bbN_M$ satisfying
    \begin{align} \label{eq:VNDT_Chordal_Decomposition_02}
        \barbQ(\bX)+\epsilon\bI=\sum_{p=1}^M\bE_{\cC_p}^T\barbQ_p\bE_{\cC_p},
    \end{align}
    where $\bE_{\cC_p}$ is defined as in \cref{thm:Chordal Decomposition}. 
    Therefore, we need only demonstrate that \cref{eq:hatbQ,QpDef,Qhatdef,eq:bbL_ij} imply \cref{eq:VNDT_Chordal_Decomposition_02}.

    From the definition of $\bE_{\cC_p}$ and $\cC_p(k)$, $(\bE_{\cC_p}^T\barbQ_p\bE_{\cC_p})_{i,j}=(\barbQ_p)_{k,l}$ with $i{=}\cC_p(k)$ and $j{=}\cC_p(l)$, and $\cC_p{\cap}\cC_q{\neq}\emptyset$ if there exists an edge $(i,j){\in}\edge{\barcQ}$ such that 
    \begin{align*}
        (\bE_{\cC_p}^T\barbQ_p\bE_{\cC_p})_{i,j}\neq\bzero,\quad(\bE_{\cC_q}^T\barbQ_q\bE_{\cC_q})_{i,j}\neq\bzero,\text{ and } p\neq q.
    \end{align*}
    The sum in \cref{eq:VNDT_Chordal_Decomposition_02} therefore reduces to
    \begin{align}
        (\overline{\mathbf{Q}}(\bX){+}\epsilon\bI)_{i,j}{=}
        \begin{cases}
            (\barbQ_p)_{k,l} & (i,j){\in}\edge{\cC_p{\setminus}\barcQ_o},\\ 
            \sum_{p\in\bbL_{i,j}}\hspace{-2pt}(\barbQ_p)_{i,j} & (i,j){\in}\edge{\barcQ_o}.  
        \end{cases} \label{eq:barbQ+e_block in}
    \end{align}
    Therefore, when $(i,j){\notin}\edge{\barcQ_o}$, this follows directly from \cref{QpDef}. 
    Now, it only remains to demonstrate \cref{eq:barbQ+e_block in} for $(i,j){\in}\edge{\barcQ_o}$.

    Using the definition of $\bY_{i,j}^p$, \cref{eq:barbQ+e_block in} becomes
    \begin{align} \label{eq:hatbQ_block}
        (\barbQ(\bX)+\epsilon\bI)_{i,j}-\hspace{-5pt}\sum_{p\in\bbL_{i,j}}\hspace{-5pt}\bY_{i,j}^p=\bzero,\quad(i,j)\in\edge{\barcQ_o}.
    \end{align}
    Due to \cref{eq:hatbQ}, $\hatbQ=\bzero$, so the above equation is implied by \cref{Qhatdef}, completing the proof.
\end{proof}



The above theorem is useful because \cref{eq:barbQ_p,eq:hatbQ} break $\barbQ(\bX)\prec0$ into several smaller \glspl{lmi} that each involve fewer design variables, which is shown explicitly next.

\begin{theorem} 
Under the assumptions of \cref{thm:VNDT_Chordal_Decomposition}, 
$\barbQ_p$ depends only on 
\begin{align}
    \barbX_p=&\diag\left(\begin{array}{l}
             \diag(\bQ_i)_{(i,i)\in\edge{\cC_p\setminus\barcQ_o}}, \\
             \diag(\bS_i,\bS_j)_{(i,j)\in\edge{\cC_p\setminus\barcQ_o}}, \\
             \diag(\bR_k)_{k{\in}\cV(\cR_p)}
        \end{array}\right),  \label{eq:barbX_p}\\
        \cV(\cR_p)=&\set{k{\in}\outneigh{\cH}{i,j}|(i,j)\in\edge{\cC_p\setminus\barcQ_o}}, \text{ and}\label{eq:cV(cR_p)}\\
        \barbY_p=&\diag(\bY_{i,j}^p)_{(i,j)\in\edge{\cC_p\cap\barcQ_o}}, \label{eq:barbY_p}
\end{align}
rather than all of $\bX$ and all $\bY_{i,j}^p$, $p\in\bbN_M$. Likewise, $\hatbQ$ depends only on  
\begin{align}
    \hatbX=&\diag\left(\begin{array}{l}
             \diag(\bQ_i)_{(i,i)\in\edge{\barcQ_o}}, \\
             \diag(\bS_i,\bS_j)_{(i,j)\in\edge{\barcQ_o}}, \\
             \diag(\bR_k)_{k{\in}\vertex{\cR}}
        \end{array}\right), \label{eq:hatbX}\\
        \vertex{\cR}=&\set{k{\in}\outneigh{\cH}{i,j}|(i,j)\in\edge{\barcQ_o}},\text{ and} \label{eq:cV(cR)}\\
        \bY=&\diag(\diag(\bY_{i,j}^p)_{p\in\bbL_{i,j}})_{(i,j)\in\barcQ_o}, \label{eq:bY}
\end{align}
rather than all of $\bX$ and all $\bY_{i,j}^p$, $p\in\bbN_M$.

\end{theorem}

\begin{proof}
    This proof follows directly from \cref{thm:barbQ,thm:VNDT_Chordal_Decomposition}.
    From \cref{QpDef,thm:barbQ}, if $(i,j)\in\edge{\cC_p\setminus\barcQ_o}$, \cref{eq:barbX_p} is sufficient to represent $(\barbQ)_{k,l}$, while if $(i,j)\in\edge{\cC_p\cap\barcQ_o}$, \cref{eq:barbY_p} is sufficient.
    Similarly, from \cref{Qhatdef,thm:barbQ}, if $(i,j)\in\edge{\cC_p\setminus\barcQ_o}$, \cref{eq:hatbX,eq:bY} are sufficient to represent $(\hatbQ)_{i,j}$.
\end{proof}

\begin{figure}
    \centering
    \includegraphics[width = 0.45\textwidth]{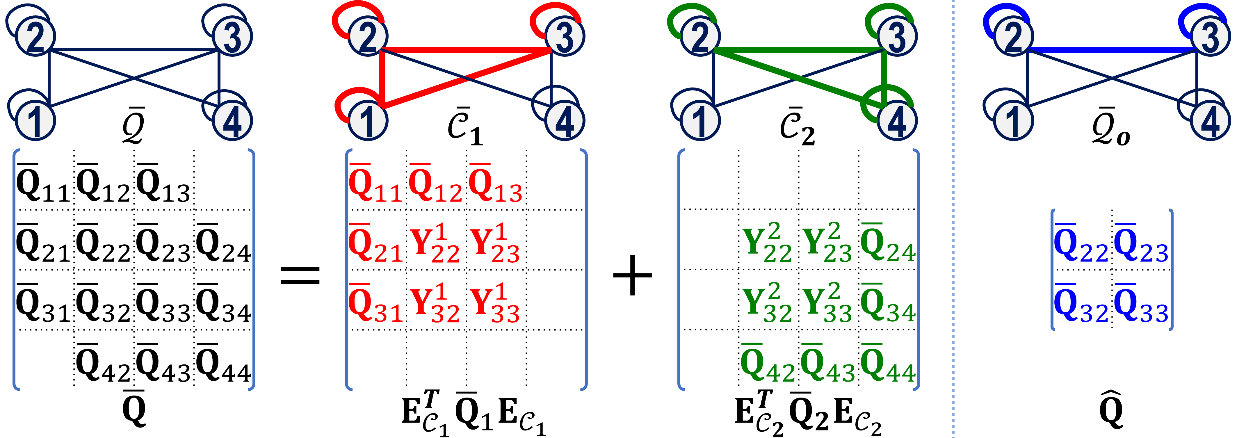}
    \caption{Chordal decomposition of $\barbQ$, corresponding to $\barcQ$ in \cref{fig:Augmented Graph}, so that $\barbQ{+}\epsilon\bI{\preceq} 0$ if and only if $\barbQ_i{\preceq} 0$ for $i{=}1,2$.}
    \label{fig:Chordal Decomposition}
    \vspace*{-1.25\baselineskip} 
\end{figure}

\cref{fig:Chordal Decomposition} illustrates an example of applying \cref{thm:VNDT_Chordal_Decomposition} to a graph and its associated negative semi-definite matrix.
In \cref{fig:Chordal Decomposition}, the matrix corresponding to the leftmost graph can be expressed as the sum of smaller negative semi-definite matrices, derived from the maximal cliques of the graph structure.
Consequently, the \gls{lmi} $\barbQ(\bX){\preceq}0$ is equivalent to $\barbQ_1(\barbX_1{,}\barbY_1){\preceq}{0}$, $\barbQ_2(\barbX_2{,}\barbY_2){\preceq}0$, $(\hatbQ(\hatbX{,}\bY))_{2,2}{=}(\barbQ(\bX))_{2,2}-\bY_{2{,}2}^1{-}\bY_{2{,}2}^2{=}\bzero$, $(\hatbQ(\hatbX{,}\bY))_{3,3}=(\barbQ(\bX))_{3{,}3}-\bY_{3{,}3}^1-\bY_{3{,}3}^2=\bzero$, and $(\hatbQ(\hatbX{,}\bY))_{2{,}3}=(\barbQ(\bX))_{2,3}-\bY_{2,3}^1-\bY_{2,3}^2=\bzero$.
In this case, $\cC_1(1){=}1$, $\cC_1(2){=}2$, $\cC_1(3){=}3$, $\cC_2(1){=}2$, $\cC_2(2){=}3$, and $\cC_2(3){=}4$.

\subsubsection{Parallel Computability of \texorpdfstring{\cref{eq:barbQ_p}}{Equation 20}} \label{subsubChap:Parallel_Computing_of_LMI}

Parallel computability is one of the main advantages of \cref{thm:VNDT_Chordal_Decomposition}. This is demonstrated by \cref{cor:No Overlap}, which shows that \glspl{lmi} in \cref{eq:barbQ_p} are independent from each other for different $p\in\bbN_M$, meaning that $\barbX_p$ and $\barbX_q$ for $p\neq q$ consist of different $(\bQ_i,\bS_i,\bR_i)$-dissipativity matrices. 

\begin{corollary} \label{cor:No Overlap}
    Under the conditions and definitions in \cref{thm:barbQ,thm:supergraph,thm:VNDT_Chordal_Decomposition}, for all $i,j\in \cC_p$ and $k,l\in \cC_q$ with $p\neq q$, if $(i,j),(k,l)\notin\edge{\barcQ_o}$, then 
    \begin{align} \label{eq:No Overlap}
        \cN_\cH^+(i,j)\cap\cN_\cH^+(k,l)=\emptyset.
    \end{align}
\end{corollary}
\begin{proof} 
    \cref{cor:No Overlap} can be proved by its contrapositive. Assume that $\cN_\cH^+(i,j){\cap}\cN_\cH^+(k,l){\neq}\emptyset$; that is, the intersections $\cN_\cH^+(i,k)$, $\cN_\cH^+(i,l)$, $\cN_\cH^+(j,k)$, and $\cN_\cH^+(j,l)$ are non-empty. This implies that $(\barbQ)_{ik}{\neq}\bzero$, $(\barbQ)_{il}{\neq}\bzero$, $(\barbQ)_{jk}{\neq}\bzero$, and $(\barbQ)_{jl}{\neq}\bzero$, meaning that $i,j{\in} \vertex{\cC_q}$ and $k,l{\in} \vertex{\cC_p}$. Since $i,j{\in} \vertex{\cC_p}$ and $k,l{\in} \vertex{\cC_q}$, it follows that $(i,j),(k,l){\in}\edge{\barcQ_o}$.
\end{proof}

\cref{cor:No Overlap} indicates that $\cV(\cR_p){\cap}\cV(\cR_q){=}\emptyset$. 
Since $\barbX_p$ is only defined by using $\bQ_i$ and $\bS_i$ in $\vertex{\cC_p{\setminus}\barcQ_o}$ and $\bR_i$ in $\cV(\cR_p)$, $\barbX_p$ and $\barbX_q$ share no common dissipativity parameters. 
Consequently, the constraints in \cref{eq:barbQ_p} are mutually independent, so they can be solved in parallel, dramatically improving the computational efficiency of \cref{alg:02} over \cref{alg:01}.
In addition, agents in $\vertex{\cC_p}$ need not share their dissipativity information in $\barbX_p$ with agents outside $\vertex{\cC_p}$, thereby offering greater privacy than \cref{alg:01}.



\subsubsection{Equivalent Problem Statement}\label{subs:equiv}

By \cref{thm:VNDT_Chordal_Decomposition}, the $\cL_2$-stability of a multi-agent system can be verified by solving the feasibility problem
\begin{subequations} \label{opt:Main Chordal Problem}
    \begin{align}
        \text{Find}\quad&\bX_i \quad\forall i\in\bbN_N,\\
        \text{s.t.}\quad&\bX_i\in\bbP_i, \label{opt:Main Chordal Problem: KYP} \\ 
            &\diag(\barbX_p,\barbY_p)\in\barbbQ_p,\quad\forall p\in\bbN_M, \label{opt:Main Chordal Problem: Stability}  \\
            &\diag(\hatbX,\bY)\in\hatbbQ, \label{opt:Main Chordal Problem: Equality}
    \end{align}
\end{subequations}
where $\barbbQ_p=\{\diag(\barbX_p,\barbY_p)|\barbQ_p(\barbX_p,\barbY_p)\preceq0\}$ and $\hatbbQ=\{\diag(\hatbX,\bY)|\hatbQ(\hatbX,\bY)=\bzero\}$, as defined in \cref{thm:VNDT_Chordal_Decomposition}.


\subsection{Distributed Stability Analysis} \label{subChap:DNSA}
\begin{table}
\caption{Variables for \gls{admm} of \cref{opt:Main Chordal Problem}}
\label{tab:ADMM Variables}
\begin{center}
\renewcommand{\arraystretch}{1.3}
\begin{tabular}{p{1cm}|p{6.5cm}}
\hline
 Variables & Definition \\
\hline\hline
$\bZ_p$ & 
The clone variable of $\diag(\barbX_p,\barbY_p)$\\
\hline
$\bW$ & 
The clone variable of $\diag(\hatbX,\bY)$\\
\hline
$\bJ$ & 
The global clone variable of $\diag(\bX,\bY)$\\
\hline
$\bT_i$ & The dual variable of $\bX_i$ \\
\hline
$\bU_p$  & The dual variable of $\bZ_p$ \\
\hline
$\bV$    & The dual variable of $\bW$ \\
\hline
$\tilbJ_i$  & 
The block diagonal matrix defined from block diagonal components in $\bJ$, which serve as a clone variable of $\bX_i$ \\
\hline
$\barbJ_p$   & 
The block diagonal matrix defined from block diagonal components in $\bJ$, which serve as a clone variable of $\bZ_p$ \\
\hline
$\hatbJ$      & 
The block diagonal matrix defined from block diagonal components in $\bJ$, which serve as a clone variable of $\bW$ \\
\hline
\end{tabular}
\end{center}
\vspace*{-1.75\baselineskip} 
\end{table}

\gls{admm} can be applied to iteratively solve \cref{opt:Main Chordal Problem}, as outlined in \cref{Chap:ISNSA}. This will be crucial to solving the problem in a distributed manner. 
To apply \gls{admm}, we need to define clone and dual variables of \cref{opt:Main Chordal Problem}, which are summarized in \cref{tab:ADMM Variables}.
With these, \cref{opt:Main Chordal Problem} can be solved iteratively using \gls{admm} in three steps.

\begin{enumerate}
    \item Primal \textbf{X}, Clones \textbf{Z}, \textbf{W} Update
    \begin{subequations} \label{eq:Chordal proj}
        \begin{align}
            \bX_i^{k+1}
                &=\Pi_{\bbP_i}(\tilbJ_i^k-\bT_i^k), 
                \quad i\in\bbN_N, \label{eq:X Chordal proj}\\
            \bZ_p^{k+1}
                &=\Pi_{\barbbQ_p}(\barbJ_p^{k}-\bU_p^k) 
                \quad p\in\bbN_M, \label{eq:Z Chordal proj} \\ 
            \bW^{k+1}
                &=\Pi_{\hatbbQ}(\hatbJ^k-\bV^k). \label{eq:W Chordal min} 
        \end{align}
    \end{subequations}
    \item Clone \textbf{J} Update
    \begin{align} \label{eq:J Chordal bulk}
        \hspace{-7pt}\bJ^{k+1}{=}\arg\min_\bJ
        \left(\begin{array}{l}
             \hspace{-5pt}\sum_{i\in\bbN_N}
                            \|\bX_i^{k+1}-\tilbJ_i+\bT_i^k\|_F^2 \\
                \;+\sum_{p\in\bbN_M}
                            \|\bZ_p^{k+1}-\barbJ_p+\bU_p^k\|_F^2 \hspace{-7pt} \\
                \;\;+\|\bW^{k+1}-\hatbJ+\bV_r^k\|_F^2
        \end{array}\right).
    \end{align}
    \item Dual \textbf{T}, \textbf{U}, \textbf{V} Update
    \begin{subequations}\label{eq:Dual Chordal update}
        \begin{align}
            \bT_i^{k+1}&=\bT_i^k+(\bX_i^{k+1}-\tilbJ_i^{k+1})=\bT_i+\tilbR_i^{k+1} \label{eq:T Chordal update} \\
            \bU_p^{k+1}&=\bU_p^k+(\bZ_p^{k+1}-\barbJ_p^{k+1})=\bU_p^k+\barbR_p^{k+1} \label{eq:U Chordal update} \\
            \bV^{k+1}&=\bV^k+(\bW^{k+1}-\hatbJ^{k+1})=\bV^k+\hatbR^{k+1}, \label{eq:V Chordal update}
        \end{align}
    \end{subequations}
    where $\tilbR_i^{k+1}=\bX_i^{k+1}-\tilbJ_i^{k+1}$, $\barbR_p^{k+1}=\bZ_p^{k+1}-\barbJ_p^{k+1}$, and $\hatbR^{k+1}=\bW^{k+1}-\hatbJ^{k+1}$.
\end{enumerate}

\subsubsection{Convergence}
The convergence of the iterative processes described in \cref{eq:Chordal proj,eq:J Chordal bulk,eq:Dual Chordal update} is established in the following theorem.
\begin{theorem} \label{thm:chordal_admm}
    Suppose that the assumptions of \cref{thm:VNDT_Chordal_Decomposition} hold and that \cref{opt:Main Problem} has a feasible solution. If $\bbP_i$ for all $i{\in}\bbN_N$, $\barbbQ_p$ for all $p{\in}\bbN_M$, and $\hatbbQ$ are closed convex sets, then the \gls{admm} iterates defined in \cref{eq:Chordal proj,eq:J Chordal bulk,eq:Dual Chordal update} satisfy
    \begin{align} \label{eq:chordal converge}
    \begin{array}{lll}
         \displaystyle\lim_{k\to\infty}\tilbR_i^k=\bzero, &
         \displaystyle\lim_{k\to\infty}\barbR_p^k=\bzero, &
         \displaystyle\lim_{k\to\infty}\hatbR^k=\bzero, \\
         \displaystyle\lim_{k\to\infty}\bT_i^k=\bT_i^\star, &
         \displaystyle\lim_{k\to\infty}\bU_p^k=\bU_p^\star, &
         \displaystyle\lim_{k\to\infty}\bV^k=\bV^\star,
    \end{array}
    \end{align}
    for all $i{\in}\bbN_N$ and $p{\in}\bbN_M$, where $\bT_i^\star$, $\bU_p^\star$, and $\bV^\star$ are dual optimal points of \cref{opt:Main Chordal Problem: KYP,opt:Main Chordal Problem: Stability,opt:Main Chordal Problem: Equality}, respectively.
\end{theorem}
\begin{proof}
    Due to \cref{thm:VNDT_Chordal_Decomposition}, \cref{opt:Main Chordal Problem} has a solution whenever \cref{opt:Main Problem} does. The proof proceeds analogously to that of \cref{thm:isnsa_admm}, by reformulating \cref{opt:Main Chordal Problem} and verifying the assumptions of \cref{thm:ADMM}.
\end{proof}

\subsubsection{Vectorization of \texorpdfstring{\cref{eq:hatbQ}}{Equation 21}}\label{subsubChap:Vectorization_of_Matrix_Equation}

All equality constraints in \cref{eq:hatbQ} can be combined into a single equality constraint by vectorization, as described in \cref{cor:Equality_constraints}.
This accelerates \cref{alg:02} by providing an exact solution to the projection step.
Furthermore, in certain networks, the resulting equality constraint is separable; that is, it can be decomposed into multiple smaller equality constraints, as discussed in \cref{subChap:Separable_Equality}.
\begin{corollary} \label{cor:Equality_constraints}
    Under conditions in \cref{thm:VNDT_Chordal_Decomposition}, \cref{eq:hatbQ} can be reformulated into a single linear system equation, 
    \begin{align} \label{eq:Big Equality}
        \bM\begin{bmatrix}
            \bx \\ \by
        \end{bmatrix} = \bzero,
    \end{align}
    where $\bx {=} \ve_b(\hatbX)$, $\by{=}\ve_b(\bY)$, $\hatbX$ and $\bY$ follows the definition in \cref{thm:VNDT_Chordal_Decomposition}, and $\bM$ is a full row rank matrix.
\end{corollary}

\begin{proof}
    To establish \cref{eq:Big Equality}, vectorization operators are applied to \cref{eq:hatbQ}, resulting in
    \begin{align} \label{eq:vec_hatbQ}
        &\ve((\barbQ(\bX)+\epsilon\bI)_{i,j})-\row(\bI_{l_i^2})_{p\in\bbL_{i,j}}\ve_b(\hatbY_{i,j})=\bzero,
    \end{align}
    where $l_i$ is the dimension of the exogenous input $\bu_i$ as defined in \cref{Chap:Problem Statement}, and $\hatbY_{i,j}=\diag(\bY_{i,j}^p)_{p\in\bbL_{i,j}}$. 
    Based on \cref{eq:barbQ_ii,eq:barbQ_ij}, \cref{eq:vec_hatbQ} has two forms,
    \begin{align*}
        &\bI_{l_i^2}\ve(\bQ_i)+\row\big((\bH)_{k,i}^T\otimes(\bH)_{k,i}^T\big)_{k\in\outneigh{\cH}{i}}\ve_b(\bR_k)_{k\in\outneigh{\cH}{i}} \nonumber\\
        &\quad+\epsilon\ve(\bI)-\row(\bI_{l_i^2})_{p\in\bbL_{i,j}}\ve_b(\hatbY_{i,j})=\bzero,\quad\qquad i{=}j, \\
        &\big((\bH)_{i,j}^T\otimes\bI_{l_i}\big)\ve(\bS_i){+}\big(\bI_{m_j}\otimes(\bH)_{j,i}^T\bP\big)\ve(\bS_j) \nonumber\\ 
        &\quad+\row\big((\bH)_{k,j}^T\otimes(\bH)_{k,i}^T\big)_{k\in\cN_\cH^+(i,j)}\ve_b(\bR_k)_{k\in\cN_\cH^+(i,j)} \nonumber\\
        &\quad-\row(\bI_{l_j\times l_i})_{p\in\bbL_{i,j}}\ve_b(\hatbY_{i,j})=\bzero,\qquad\qquad\qquad i{\neq}j.
    \end{align*}
    Each case has the form of a linear system equation $\bM_{i,j}'\begin{bmatrix}
            \bx_{i,j} \\ \by_{i,j}
        \end{bmatrix} = \bzero$,
    where, for $i=j$,
    \begin{align*}
        \bM_{i,j}'&=\begin{bmatrix}
            \bI_{l_i^2}^T \\
            \row\big((\bH)_{k,i}^T\otimes(\bH)_{k,i}^T\big)_{k\in\outneigh{\cH}{i}}^T \\
            -\row(\bI_{l_i^2})_{p\in\bbL_{i,i}}^T
        \end{bmatrix}^T, \\
        \bx_{i,j}&=\begin{bmatrix}
            \ve(\bQ_i) \\ \ve_b(\bR_k)_{k\in\outneigh{\cH}{i}} 
        \end{bmatrix},\quad
        \by_{i,j}=\ve_b(\hatbY_{i,i}),
    \end{align*}
    and for $i\neq j$,
    \begin{align*}
        \bM_{i,j}'&=\begin{bmatrix}
            \big((\bH)_{i,j}^T\otimes\bI_{l_i}\big)^T \\ 
            \big(\bI_{m_j}\otimes(\bH)_{j,i}^T\bP\big)^T \\
            \row\big((\bH)_{k,j}^T\otimes(\bH)_{k,i}^T\big)_{k\in\cN_\cH^+(i,j)}^T \\ 
            -\row(\bI_{l_j\times l_i})_{p\in\bbL_{i,j}}^T
        \end{bmatrix}^T \\
        \bx_{i,j}&=\begin{bmatrix}
            \ve(\bS_i) \\ \ve(\bS_j) \\ \ve_b(\bR_k)_{k\in\cN_\cH^+(i,j)}
        \end{bmatrix}, \quad
        \by_{i,j}=\ve_b(\hatbY_{i,j}).
    \end{align*}
    Each $\bM'_{i,j}\hspace{-2pt}$ is full row rank.
    $\begin{bmatrix}
        \bx_{i,j} \\ \by_{i,j}
    \end{bmatrix}$ can be embedded into 
    $\begin{bmatrix}
        \bx \\ \by
    \end{bmatrix}$ by augmenting each $\bM_{i,j}'$ with appropriate zero columns, yielding $\bM_{i,j}$.
    Therefore, \cref{eq:Big Equality} is constructed by using $\bM=\col(\bM_{i,j})_{(i,j)\in\edge{\barcQ_o}}$.

    To complete the proof, it remains to note that $\bM$ inherits full row rank from the
    $\bM_{i,j}'$ matrices, each of which contains an identity submatrix, either $\bI_{l_i^2}$ or $\bI_{l_j\times l_i}$.
\end{proof}

\subsubsection{Closed Form Solutions} \label{subsubChap:Reformulation}
For practical implementation, \cref{eq:W Chordal min,eq:J Chordal bulk}, can be reformulated in closed form, as presented in the following corollaries. 
\begin{corollary} \label{thm:W_Reformulation}
    Consider the projection in \cref{eq:W Chordal min}. Suppose that $\hatbbQ$ is a hyperplane as defined in \cref{cor:Equality_constraints}, associated with $\bW$. 
    Then, the closed-form solution to \cref{eq:W Chordal min} is given by
    \begin{align} \label{eq:W Chordal proj}
        \bW^{k+1}=\ive_b\big((\bI-\bM^T(\bM\bM^T)^{-1}\bM)(\hatbj^k-\bv^k)\big),
    \end{align}
    where $\bw{=}\ve_b(\bW)$, $\hatbj^k{=}\ve_b(\hatbJ^k)$, and $\bv^k{=}\ve_b(\bv^k)$.
\end{corollary}
\begin{proof}
    By assumption, $\bW{\in}\hatbbQ$ is equivalent to $\bM\bw{=}\bzero$ using $\bw{=}\ve_b(\bW)$. 
    Thus, projecting $\hatbj^k{-}\bv^k$ onto the hyperplane $\bM\bw{=}\bzero$, which is constructed as
    \begin{align} \label{eq:W Chordal projection vec}
    \begin{split}
        \arg\min_{\bw}\;&\|\bw-\hatbj^k+\bv^k\|_2^2, \\
            \sthat\;&\bM\bw=\bzero,
    \end{split}
    \end{align}
    represents $\ve_b(\bW^\star)$, where $\bW^\star$ is the original solution to \cref{eq:W Chordal min}.
    Since $\bM$ has full row rank, $(\bM\bM^T)^{-1}$ exists, and \cref{eq:W Chordal projection vec} admits the closed-form solution \cite{boyd2004convex}, 
    \begin{align}
        \bw^\star&=(\bI-\bM^T(\bM\bM^T)^{-1}\bM)(\hatbj^k-\bv^k).
    \end{align} 
    Accordingly, the solution to \cref{eq:W Chordal min} is $\ve_b^{-1}(\bw^\star)$.
\end{proof}


\begin{figure}
    \centering
    \includegraphics[width = 0.5\textwidth]{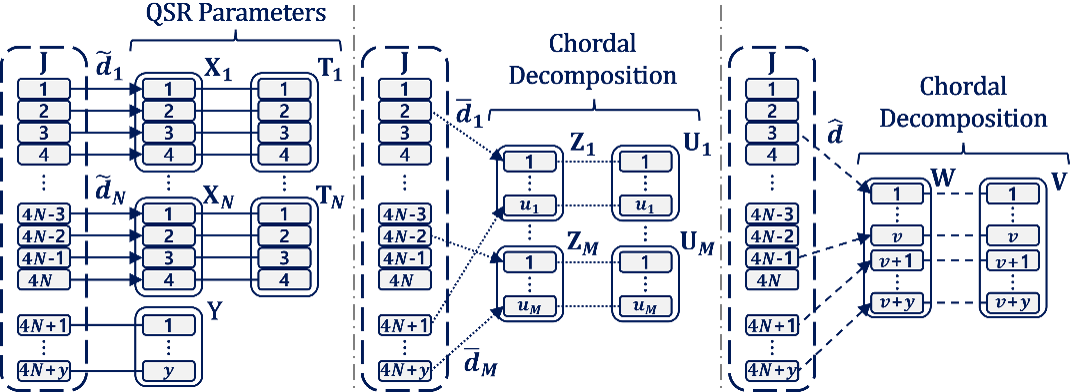}
    \caption{Mapping among all block diagonal components of variables in \cref{tab:ADMM Variables}: The gray rectangles represent the block component indices defining each variable. The QSR parameters are grouped by agent, with blocks for each dissipativity matrix, while the chordal parameters are grouped by clique with blocks for each \gls{lmi} and a matrix equation.}
    \label{fig:Matching}
    \vspace*{-1.25\baselineskip} 
\end{figure}

\begin{corollary} \label{thm:J_Reformulation}
    Consider the projection in \cref{eq:J Chordal bulk}. 
    Define $\bJ_a$ as the $a^\thh$ diagonal element of $\bJ$, and block diagonal matrices $\tilbJ_i$, $\barbJ_p$, and $\hatbJ$ as global clone variables of $\bX_i$, $\bZ_p$, and $\bW$, respectively.
    Further, define 
    \begin{align*}
        \widetilde{d}_a&=\set{ (i,j)\in \bbN_N\times\bbN_3 | \bJ_a=(\widetilde{\bJ}_i)_{j,j} },\quad i\in\bbN_N \\
        \overline{d}_a&=\set{ (p,j)\in \bbN_M\times\bbN_{u_p} | \bJ_a=(\overline{\bJ}_p)_{j,j} },\quad p\in\bbN_M \\
        \widehat{d}_a&=\set{ j\in \bbN_{\nu+y} | \bJ_a=(\hat{\bJ})_{j,j} },
    \end{align*}
    where $u_p$, $v$, and $y$ represents the number of block matrices associated with $\bZ_p$, $\hatbX$, and $\bY$, respectively.
    \cref{fig:Matching} illustrates these relationships of mappings.
    Then, the closed-form solution of \cref{eq:J Chordal bulk} for each block $\bJ_a$ of $\bJ$, is
    \begin{align} \label{eq:J Chordal update}
        \bJ_a^{k+1}&{=}
        \frac{
            \hspace{-7pt}\displaystyle\sum_{(i,j)\in\widetilde{d}_a}\hspace{-10pt}
                        {(\bX_i^{k+1}{+}\bT_i^k)_{j,j}}
            {+}\hspace{-14pt}\sum_{(p,j)\in\overline{d}_a}\hspace{-12pt}
                        {(\bZ_p^{k+1}{+}\bU_p^k)_{j,j}}
            {+}\hspace{-6pt}\sum_{j\in\hat{d}_a}\hspace{-3pt}
                    (\bW^{k+1}{+}\bV^k)_{j,j}
                    }{|\widetilde{d}_a|+|\overline{d}_a|+|\hat{d}_a|},
    \end{align}
    for all $a\in\bbN_{3N+y}$.
\end{corollary}

\begin{proof}
    \cref{eq:J Chordal bulk} is fully separable into smaller subproblems, each computing the diagonal elements $\bJ_a$ of $\bJ$. 
    Specifically, the update can be expressed as
    \begin{align*} 
        \bJ_a^{k+1}{=}\arg\min_{\tilbJ_i,\barbJ_p,\hatbJ}
            \left(
            \begin{array}{l}
            \sum_{(i,j)\in\widetilde{d}_a}
                            \|(\bX_i^{k+1}{-}\tilbJ_i{+}\bT_i^k)_{j,j}\|_F^2 \\
            \quad{+}\sum_{(p,j)\in\overline{d}_a}
                            \|(\bZ_p^{k+1}{-}\barbJ_p{+}\bU_p^k)_{j,j}\|_F^2    \\
            \quad\quad{+}\sum_{j\in\hat{d}_a}\|(\bW^{k+1}{-}\hatbJ{+}\bV_r^k)_{j,j}\|_F^2
            \end{array}
                    \right).
    \end{align*}
    This separability results from the fact that $\bX_i^{k+1}$ for different $i{\in}\bbN_N$ are associated with different dissipativity parameters and do not overlap. Likewise, $\bZ_p^{k+1}$ for different $p{\in}\bbN_M$ are associated with different dissipativity parameters.
    Consequently, the update can be computed block-wise for each $\bJ_a$, and the resulting closed-form solution yields \cref{eq:J Chordal update}.    
\end{proof}

\subsubsection{Algorithm}


\begin{algorithm}[tbp]
    \caption{Distributed network stability analysis}\label{alg:02}
    \begin{algorithmic}[1]
        \Require MaxIterations$,\bX_i^0,\epsilon$ for $i\in\bbN_N$
        \Ensure $\bX^k$
        \State Initialize $k=0$, $\bJ^0=\diag(\diag(\bX_i^0)_{i\in\bbN_N},\bI^\bY)$, where $\bX^0=\tilbJ^0$, $\bZ^0=\barbJ^0$, and $\bW^0=\hatbJ$, 
            $\bT^0=\bU^0=\bV^0=\textbf{0}$
        \While {$\barbQ_p(\barbX_p,\barbY_p){\nprec}0,\hatbQ(\hatbX,\bY){\neq}\bzero,k{<}$MaxIterations}
            \State $k\gets k+1$
            \State Find $\bX_i^k,\bZ_p^k,\bW^k$ by \cref{eq:X Chordal proj,eq:Z Chordal proj,eq:W Chordal proj} in parallel
            \State Find $\bJ_\alpha^k$ by \cref{eq:J Chordal update} in parallel
            \State Find $\bT_i^k,\bU_p^k,\bV^k$ by \cref{eq:T Chordal update,eq:U Chordal update,eq:V Chordal update} in parallel
        \EndWhile
        \If {$\barbQ_p(\barbX_p,\barbY_p){\prec}0,\text{ and }\hatbQ(\hatbX,\bY)=\bzero$}
            \State Multi-agent system is stable with $\bX^k=\diag(\bX_i^k)_{i\in\bbN_N}$
        \EndIf
    \end{algorithmic}
\end{algorithm}

The implementation of these \gls{admm} steps is summarized in \cref{alg:02}.
As in \cref{alg:01}, any initial point such that $\bX^0{=}\tilbJ^0$, $\bZ^0{=}\barbJ$, $\bW^0{=}\hatbJ^0$, and $\bT^0{=}\bU^0{=}\bV^0{=}\bzero$ can be used. A natural choice is $\bJ^0{=}\diag(\diag(\bX_i^0)_{i{\in}\bbN_N},\bI)$ for $\bX^0_i{=}a_iI$.
\cref{alg:02} converges to a feasible point of \cref{opt:Main Chordal Problem} if a feasible point exists. 
Conversely, if the algorithm does not converge, then a feasible point does not exist, which means the \gls{vndt} does not guarantee the stability of the network.

\subsection{Computation Time and Complexity}

The computation time for each iteration is $t_k{=}t_k^{\bP}{+}t_k^{\bJ}{+}t_k^{\bD}$, where $t_k^{\bP}{=}\max_{i{\in}\bbN_N,p{\in}\bbN_M}(t_k^{\bX_i},t_k^{\bZ_p},t_k^{\bW})$, $t_k^{\bJ}{=}\max_{a{\in}3N+y}(t_k^{\bJ_a})$, $t_k^{\bD}{=}\max_{i{\in}\bbN_N,p{\in}\bbN_M}(t_k^{\bT_i},t_k^{\bU_p},t_k^{\bV})$ are similarly defined based on the time variables introduced in \cref{Chap:ISNSA}. 

From \cref{cor:No Overlap}, all projections can be executed in parallel, so the complexity of \cref{alg:02} is dominated by the most intensive step among \cref{eq:X Chordal proj,eq:Z Chordal proj,eq:W Chordal proj}.
The computations in \cref{eq:X Chordal proj} scale as $\cO(n_m^6)$. 
Let $n_p$ be the maximum size of matrix variables related to $\cC_p$ for all $p{\in}\bbN_M$, then the complexity of \cref{eq:Z Chordal proj} is $\cO(n_q^6)$.
\cref{eq:W Chordal proj} involves solving a linear system, which is generally faster than solving \gls{sdp} unless matrices are significantly larger than those in \cref{eq:X Chordal proj,eq:Z Chordal proj}.
Thus, the overall per iteration complexity of \cref{alg:02} is determined by either $\cO(n_m^6)$ or $\cO(n_q^6)$, depending on the network's clique structure.

Assuming $n_q{\approx} n_m$ and the required iterations of \cref{alg:01,alg:02} are $k_1$ and $k_2$, respectively, their total complexities are $k_1\cO(N^4n_m^6)$ and $k_2\cO(n_m^6)$.
Hence, \cref{alg:02} converges faster than \cref{alg:01} whenever $k_2{<}k_1N^4$, which typically holds for large-scale systems with $N{>}5$.
This results in the superior computational efficiency of \cref{alg:02} over \cref{alg:01}, which is seen in the numerical examples.

\subsection{Chordal Networks}
To determine when chordal decomposition is potentially useful, this section explores when it can be employed without modifying the underlying system interconnection structure. Although some $\barcQ$ are non-chordal, \cref{thm:VNDT_Chordal_Decomposition} is always applicable to any graph because any non-chordal graph can be extended to a chordal graph by adding extra edges to the original graph. However, finding the chordal extension with the fewest additional edges is an NP-hard problem \cite{yannakakis1981computing}. 
Fortunately, for certain graph structures of $\cH$, $\barcQ$ is inherently chordal without requiring chordal extension. 
The next two theorems give conditions under which $\barcQ$ is chordal and no extension is needed to apply \cref{thm:VNDT_Chordal_Decomposition}.

\begin{corollary} \label{cor:Chordal_graph_undirected_no_cycle}
    Under the conditions specified in \cref{thm:VNDT_Chordal_Decomposition}, assume that $\cH$ is a connected, undirected graph. If $\cH$ does not have a cycle, $\barcQ$ is a chordal graph.
\end{corollary}
\begin{proof}
    The proof proceeds by contradiction, demonstrating that the negation of the statement is always false. 
    Since $\cH$ is undirected, \cref{eq:Edge_barQ} is reformulated as 
    \begin{align} \label{eq01:Chordal_graph_undirected_no_cycle} 
    \begin{split}
        \edge{\barcQ}&{=}\edge{\cH}{+}\set{(i,i)|i\in\vertex{\cH}}\\
        &\quad{+}\set{(i,j)|i,j{\in}\vertex{\cH},\cN_\cH(i,j){\neq}{\emptyset}},
    \end{split}
    \end{align}
    where $\cN_\cH(i,j)=\neigh{\cH}{i}\cap\neigh{\cH}{j}$.
    
    Assume that $\barcQ$ is chordal. 
    Then, $\barcQ$ contains a cordless cycle of length greater than 3.
    Let $\barcQ'$ denote such a chordless cycle, and $\vertex{\barcQ'}{=}\set{i_k}_{k{\in}\bbN_n}$.
    It follows that  
    \begin{align*}
        \set{(i_1,i_2),(i_2,i_3),\dots,(i_{n-1},i_n),(i_n,i_1)}=\edge{\barcQ'}.
    \end{align*}
    
    If $\cH$ does not have a cycle, there exists $(i_k,i_{k+1}){\in}\edge{\barcQ'\setminus\cH}$.
    It implies that $(i_k,i_{k+1}){\in}\set{(i,j)|i,j{\in}\vertex{\cH},\cN_\cH(i,j){\neq}\emptyset}$ from \cref{eq01:Chordal_graph_undirected_no_cycle}. 
    Equivalently, there exist $l{\in}\vertex{\cH}$ such that $l{\in}\cN_\cH(i_k,i_{k+1})$.
    It means that vertex $l{\in}\vertex{\cH}$ connects vertices $i_k$ and $i_{k+1}$.
    
    As a result, for all $(i_k,i_{k+1})\in\edge{\barcQ'\setminus\cH}$, there exists a vertex $l\in\vertex{\cH}$ such that $l$ connects $i_k$ and $i_{k+1}$. It defines a cycle $\cH'$ with an edge set as
    \begin{align*}
        &\edge{\cH'}{=}\edge{\barcQ'{\cap}\cH} \\
            &\;\;{+}\set{(i_k{,}l){,}(l{,}i_{k+1}){\in}\edge{\cH}\,|\,(i_k,i_{k+1}){\in}\edge{\barcQ'\setminus\cH}\,{\forall} k{\in}\bbN_{n-1}}.
    \end{align*}
    Since $\edge{\cH'}\subseteq\edge{\cH}$, $\cH'$ is the subgraph of $\cH$.
    It implies that $\cH$ has a cycle, which violates the negation of the statement.
    Therefore, if $\cH$ is acyclic, then $\barcQ$ must be chordal.
\end{proof}

\begin{corollary} \label{cor:Chordal_graph_directed}
    Under the conditions specified in \cref{thm:VNDT_Chordal_Decomposition}, let $\cH$ be a directed graph with no directed cycles. Assume that for all $i\in\vertex\cH$, $i$ does not have more than one in-neighbor, meaning $|\inneigh{\cH}{i}|=1$. Then, $\barcQ$ is a chordal graph.
\end{corollary}

\begin{proof}
    The claim is verified by establishing $\edge{\barcQ}=\edge{G(\cH)}$ using \cref{thm:supergraph}.
    Suppose that there exist $i,j,k\in\vertex{\cH}$ such that $k\in\cN_\cH^+(i,j)$.
    This implies $\set{i,j}\subseteq\inneigh{\cH}{k}$, which violates the assumption that $|\inneigh{\cH}{k}|=1$.
    Therefore, for all $i,j\in\vertex{\cH}$, $\cN_\cH^+(i,j)=\emptyset$, implying $\edge{\barcQ}=\edge{G(\cH)}$ from \cref{thm:supergraph}.

    Since $\cH$ does not contain any directed cycles, its associated undirected graph $G(\cH)$ is acycle.
    It follows that $\barcQ$ does not have any cycle either. Hence, $\barcQ$ is a chordal graph.
\end{proof}

\subsection{Networks with Separable Equality Constraints} \label{subChap:Separable_Equality}
As demonstrated in \cref{subChap:Problem_Reformulation,subChap:DNSA}, applying the stability conditions in \cref{thm:VNDT_Chordal_Decomposition} yields smaller \glspl{lmi} in \cref{opt:Main Chordal Problem: Stability}, and a single matrix equation in \cref{opt:Main Chordal Problem: Equality}.
Although the computation time for solving \cref{eq:W Chordal min} can be accelerated using \cref{eq:W Chordal proj}, making it much faster than solving \cref{eq:Z proj}, \cref{eq:W Chordal proj} can be a bottleneck in \cref{alg:02}, particularly in extremely large-scale network systems.
Additionally, it may require network-wide communication of certain dissipativity parameters.
These challenges stem from the fact that the matrix equation in \cref{opt:Main Chordal Problem: Equality} can be substantially larger than the \glspl{lmi} in \cref{opt:Main Chordal Problem: Stability}.

Fortunately, for certain networks, 
the structure of $\barcQ_0$ allows the large matrix equation to be decomposed into smaller, independent equations.
This decomposition reduces both computation time and the communication of dissipativity parameters among agents.
This section establishes the condition under which the network graph has a separable matrix equality.

In particular, \cref{thm:Separable_Equality} demonstrates that if $\barcQ_o$ is a disconnected graph, then the single equation in \cref{opt:Main Chordal Problem: Equality} can be divided into several smaller matrix equations.

\begin{theorem} \label{thm:Separable_Equality}
    Suppose that \cref{thm:VNDT_Chordal_Decomposition} holds, and let the graph $\barcQ_o$ be the union of $X$ components $\cO_r$, meaning that 
    \begin{align*}
        \bigcup_{r\in\bbN_X}\cO_r=\barcQ_o,\quad
        \cO_r\cap\cO_l=\emptyset,\quad r,l\in\bbN_X,r\neq l.
    \end{align*}
    Further, define the sets 
    \begin{align*}
        \vertex{\cR_{o,r}}=\set{k\in\cN_\cH(i,j)|(i,j)\in\edge{\cO_r}},\quad\forall r\in\bbN_X.
    \end{align*}
    Then, \cref{eq:hatbQ} is equivalent to
    \begin{align} \label{eq:hatbQ_r}
        \hatbQ_r(\hatbX_r,\hatbY_r){=}\bzero{\in}\bbR^{\abs{\vertex{\cO_r}}{\times}\abs{\vertex{\cO_r}}},\;{\forall}r{\in}\bbN_X\;\text{(block-wise)},
    \end{align}
    where
    \begin{align}
        (&\hatbQ_r(\hatbX_r{,}\hatbY_r))_{v,w}{=}(\barbQ(\bX){+}\epsilon\bI)_{i,j}{-}\hspace{-5pt}\sum_{p\in\bbL_{i,j}}\hspace{-5pt}\bY_{i,j}^p{,}\quad\forall(i,j){\in}\edge{\cO_r}, \nonumber \\
        &\hatbX_r{=}\diag\left(\begin{array}{l}
             \diag(\bQ_i)_{(i,i)\in\edge{\cO_r}}, \\
             \diag(\bS_i,\bS_j)_{(i,j)\in\edge{\cO_r}}, \\
             \diag(\bR_k)_{k{\in}\vertex{\cR_{o,r}}}
        \end{array}\right),\;\hatbY_r{=}\diag(\hatbY_{i,j}), \label{eq:hatX_r}
    \end{align}
    for all $r{\in}\bbN_X$, $(\hatbQ_r(\hatbX_r,\hatbY_r))_{v,w}$ is $(v,w)$ block of $\hatbQ_r(\hatbX_r,\hatbY_r)$, and $\hatbY_{i,j}$ and $\bbL_{i,j}$ follows the definition in \cref{thm:VNDT_Chordal_Decomposition}.
\end{theorem}

\begin{proof}
    \cref{eq:hatbQ_r} follows directly from the disconnected property of $\barcQ_o$.
    By leveraging the structure of each $\cO_r$ for $r{\in}\bbN_X$, $\hatbQ(\hatbX{,}\bY)$ in \cref{eq:hatbQ} is written as
    \begin{align*}
        \hatbQ(\hatbX,\bY)=\diag(\hatbQ_r(\hatbX,\bY))_{r\in\bbN_X}
    \end{align*}
    where 
    \begin{align*}
        \hatbQ_r(\hatbX,\bY)&\in\bbR^{\abs{\vertex{\cO_r}}\times\abs{\vertex{\cO_r}}},\quad\text{(block-wise)},\\
        (\hatbQ_r(\hatbX,\bY))_{v,w}&=(\hatbQ(\hatbX,\bY))_{i,j}\quad\forall(i,j)\in\edge{\cO_r},
    \end{align*}
    for all $v,w{\in}\bbN_{|\vertex{\cO_r}|}$, associated to the index induced by $\vertex{\hspace{-1pt}\cO_r\hspace{-1pt}}$.
    Consequently, for all $(i,j){\in}\edge{\cO_r}$, $\hatbQ(\hatbX,\hatbY)_{i,j}$ depends only on the elements in the set
    \begin{align*}
        \set{\bQ_i\,|\,i{\in}\vertex{\cO_r}}{+}\set{\bR_k\,|\,k{\in}\vertex{\cR_{o,r}}}{+}\set{\bY_{i,j}^p\,|\,(i,j){\in}\edge{\cO_r}}
    \end{align*}
    when $i=j$, and
    \begin{align*}
        \set{\bS_i,\bS_j|(i,j){\in}\edge{\cO_r}}{+}\set{\bR_k|k{\in}\vertex{\cR_{o,r}}}{+}\set{\bY_{i,j}^p|(i,j){\in}\edge{\cO_r}}
    \end{align*}
    when $i\neq j$.
    Thus, $\hatbQ_r$ is a function only of $\hatbX_r$ and $\hatbY_r$, defined as \cref{eq:hatX_r}, rather than $\hatbX$ and $\bY$.
    This confirms that \cref{eq:hatbQ_r} is equivalent to \cref{eq:hatbQ}.
\end{proof}

\cref{thm:Separable_Equality} alone does not give us a significant advantage, as it does not guarantee the full decoupling of the equations in \cref{eq:hatbQ_r}.
In particular, $\bR_k$ used in $\hatbX_r$ may also be required in $\hatbX_l$ for $r{\neq} l$.
This means that \cref{eq:hatbQ_r} for different $r$ must still be calculated simultaneously.
Achieving full decoupling of the equations in \cref{eq:hatbQ_r} requires additional conditions on $\barcQ_o$, as outlined in \cref{thm:Decompose_Equality}.

\begin{corollary} \label{thm:Decompose_Equality}
    Assume that \cref{thm:Separable_Equality} holds and define 
    \begin{align*}
    \begin{array}{ll}
         \bbQ_r=\set{\bQ_i\,|\,(i,i)\in\edge{\cO_r}}, &
         \bbS_r=\set{\bS_i,\bS_j\,|\,(i,j)\in\edge{\cO_r}}, \\
         \bbR_r=\set{\bR_k\,|\,k\in\vertex{\cR_{o,r}}}, &
         \bbY_r=\set{\bY_{i,j}^p\,|\,(i,j)\in\edge{\cO_r}}.
    \end{array}
    \end{align*}
    If $\vertex{\cR_{o,r}}{\cap}\vertex{\cR_{o,l}}{=}\emptyset$ for $r{\neq} l$, then $\bbQ_r{\cap}\bbQ_l{=}\emptyset$, $\bbS_r{\cap}\bbS_l{=}\emptyset$, $\bbR_r{\cap}\bbR_l{=}\emptyset$ and $\bbY_r{\cap}\bbY_l{=}\emptyset$.
\end{corollary}

\begin{proof}
    Since $\barcQ_o$ is disconnected, $\bbQ_r{\cap}\bbQ_l{=}\emptyset$, $\bbS_r{\cap}\bbS_l{=}\emptyset$, and $\bbY_r{\cap}\bbY_l{=}\emptyset$.
    Therefore, only $\bbR_r{\cap}\bbR_l{=}\emptyset$ remains to be shown, which follows directly from the assumption that $\vertex{\cR_{o,r}}{\cap}\vertex{\cR_{o,l}}{=}\emptyset$ for $r{\neq} l$.
\end{proof}


\begin{figure}
    \centering
    \includegraphics[width = 0.425\textwidth]{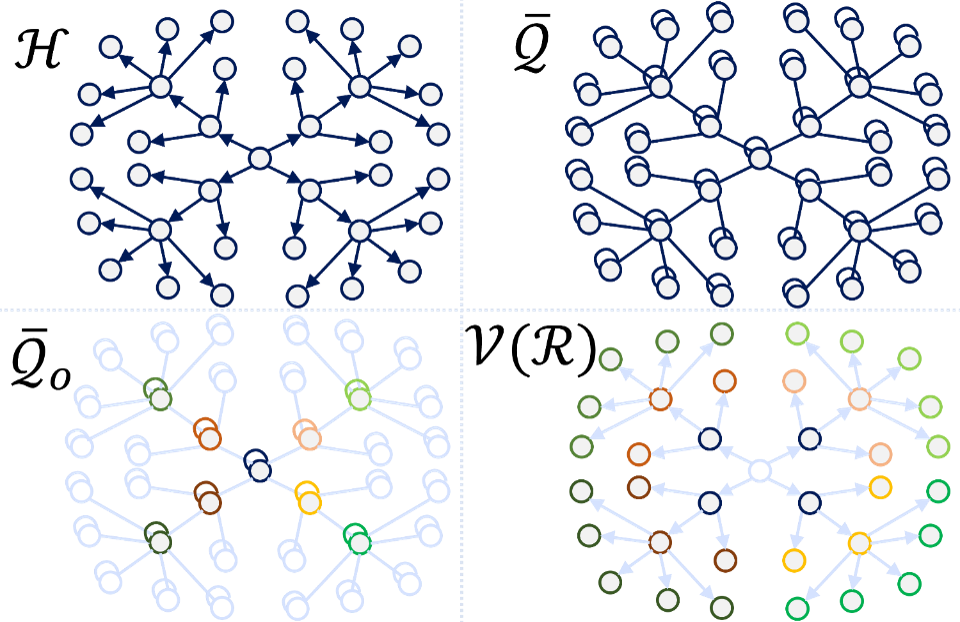}
    \caption{Hierarchical network example: The upper left and right graphs illustrate the graph of the network and $\barbQ(\bX)$. The lower left graph is the graph of overlapped elements in $\barbQ(\bX)$ resulting from \cref{thm:VNDT_Chordal_Decomposition}. The graph is disconnected and 9 components that are denoted in different colors. The lower right figure shows nodes in $\vertex{\bR}$. As in $\barcQ_o$, $\vertex{\bR}$ can be decomposed into 9 subsets.}
    \label{fig:Hierarchical Network}
    \vspace*{-1.45\baselineskip} 
\end{figure}

For multi-agent systems whose networks satisfy the conditions in \cref{thm:Decompose_Equality}, \cref{eq:hatbQ_r} can be fully decoupled using distinct dissipativity parameters, allowing parallel computation of each \cref{eq:hatbQ_r}.
This decoupling significantly reduces both the overall computation and the dissipativity parameter communications.



In conclusion, the $\cL_2$ stability of a network satisfying the conditions in \cref{thm:Separable_Equality} can be confirmed by solving 
\begin{subequations} \label{opt:Main_Separable_Chordal_Problem}
    \begin{align}
        \text{Find}\quad&\bX_i \quad i\in\bbN_N,\\
        \text{s.t.}\quad&\bX_i\in\bbP_i, \label{opt:Main_Separable_Chordal_Problem: KYP} \\ 
            &\diag(\barbX_p,\barbY_p)\in\barbbQ_p,\quad\forall p\in\bbN_M, \label{opt:Main_Separable_Chordal_Problem: Stability}  \\
            &\diag(\hatbX_r,\hatbY_r)\in\hatbbQ_r,\quad\forall r\in\bbN_X. \label{opt:Main_Separable_Chordal_Problem: Equality}
    \end{align}
\end{subequations}
The constraint sets are
\begin{align*}
    \hatbbQ_r&=\set{\diag(\hatbX_r,\hatbY_r)\;|\;\hatbQ_r(\hatbX_r,\hatbY_r)=\bzero}, \\
        &=\set{\diag(\hatbX_r,\hatbY_r)\;|\;\bM_r\ve_b(\diag(\hatbX_r,\hatbY_r))=\bzero},
\end{align*}
where $\bM_r$ is a full row rank matrix defined as in \cref{thm:Separated_Linear_System}.
\cref{opt:Main_Separable_Chordal_Problem} is equivalent to \cref{opt:Main Chordal Problem}, with the only difference being that \cref{opt:Main Chordal Problem: Equality} is replaced by \cref{opt:Main_Separable_Chordal_Problem: Equality}.

Following the same procedure as in \cref{cor:Equality_constraints}, each decoupled constraint of the form in \cref{eq:hatbQ_r} can be rewritten as a single linear system of equations.
\begin{corollary} \label{thm:Separated_Linear_System}
    Under the conditions in \cref{thm:Separable_Equality,thm:Decompose_Equality}, for all $r\in\bbN_X$, \cref{eq:hatbQ_r} can be reformulated into linear system equations
    \begin{align} \label{eq:Separated_linear_system}
        \bM_r\begin{bmatrix}
            \bx_r \\ \by_r
        \end{bmatrix}=\bzero\quad\forall r\in\bbN_X,
    \end{align}
    where $\bx_r{=}\ve_b(\hatbX_r)$, $\by_r{=}\ve_b(\hatbY_r)$, and $\bM_r$ is a full row rank matrix.
\end{corollary}
\begin{proof}
    Applying \cref{cor:Equality_constraints} to each \cref{eq:hatbQ_r} for all $r\in\bbN_X$ leads to \cref{eq:Separated_linear_system}.
\end{proof}

In the same manner, the iteration procedure described in \cref{subChap:DNSA} solve \cref{opt:Main_Separable_Chordal_Problem}, by replacing \cref{eq:W Chordal min,eq:J Chordal bulk,eq:V Chordal update} with
\begin{align}
    \hspace{-7pt}\begin{split}
        \bW_r^{k+1}&=\Pi_{\hatbbQ_r}(\hatbJ_r^k-\bV_r^k),\quad \forall r\in\bbN_X,\\
            &=\ive_b\big((\bI-\bM_r^T(\bM_r\bM_r^T)^{-1}\bM_r)(\hatbj_r^k-\bv_r^k)\big),
    \end{split} \label{eq:W Seperable Chordal proj}\\
    \hspace{-7pt}\bJ_a^{k+1}&{=}
        \frac{
            \hspace{-5pt}\displaystyle\sum_{(i,j)\in\widetilde{d}_a}\hspace{-10pt}
                        {(\bX_i^{k+1}\hspace{-2pt}{+}\bT_i^k)_{j,j}}
            {+}\hspace{-14pt}\sum_{(p,j)\in\overline{d}_a}\hspace{-12pt}
                        {(\bZ_p^{k+1}\hspace{-2pt}{+}\bU_p^k)_{j,j}}
            {+}\hspace{-14pt}\sum_{(r,j)\in\hat{d}_a}\hspace{-10pt}
                    (\bW_r^{k+1}\hspace{-2pt}{+}\bV_r^k)_{j,j}
                    }{|\widetilde{d}_a|+|\overline{d}_a|+|\hat{d}_a|} \label{eq:J Separable Chordal update}\\
    \hspace{-7pt}\bV_r^{k+1}&=\bV_r^k+(\bW_r^{k+1}-\hatbJ_r^{k+1})=\bV_r^k+\hatbR_r^{k+1},\;\;\forall r\in\bbN_X, \label{eq:V Separable Chordal update}
\end{align}
where 
\begin{itemize}
    \item $\bW_r$ is a clone variable of $\diag(\hatbX_r,\hatbY_r)$,
    \item $\hatbJ_r$ is a block diagonal matrix defined from block diagonal components in $\bJ$, serving as a clone variable of $\bW_r$,
    \item $\bV_r$ is a dual variable of $\bW_r$ ,
    \item $\bw_r=\ve_b(\bW_r)$, $\hatbj_r^k=\ve_b(\hatbJ_r^k)$, and $\bv_r^k=\ve_b(\bv_r^k)$,
    \item $\widehat{d}_a=\set{ (r,j)\in \bbN_{v_r+y_r} | \bJ_a=(\hat{\bJ}_r)_{j,j} }$,
    \item $\hatbR_r^{k+1}=\bW_r^{k+1}-\hatbJ_r^{k+1}$.
\end{itemize}

\cref{alg:02} can also be applied to implement the proposed \gls{admm} framework by replacing \cref{eq:W Chordal proj}, \cref{eq:J Chordal update}, and \cref{eq:V Chordal update} with \cref{eq:W Seperable Chordal proj},\cref{eq:J Separable Chordal update}, and \cref{eq:V Separable Chordal update}, respectively.
This enables fully distributed stability analysis for certain networks while significantly reducing the communication of dissipativity parameters.
\cref{fig:Hierarchical Network} illustrates an example that satisfies the assumptions of \cref{thm:Decompose_Equality}. 
The numerical example in \cref{Chap:Numerical Example} performs distributed stability analysis for a network with decomposable equality constraint satistfying \cref{thm:Decompose_Equality}.

\section{Extension to Nonlinear Systems} \label{Chap:Extension to Nonlinear Systems}
There are many variations of \glspl{lmi} similar to \cref{lem:KYP Lemma}, that verify the dissipativity of certain nonlinear systems, such as Euler-Lagrange systems known to be passive \cite[Chapter 6]{lozano2013dissipative}, linear systems with time-delay \cite{li2002delay,bridgeman2016conic}, parametric uncertainty \cite{walsh2019interior}, and stochastic systems \cite{haddad2022dissipativity}. 
Likewise, if $\mathscr{G}_i$ is not \gls{lti} but has some structured nonlinearity, a variation of \cref{lem:KYP Lemma} may be applied with adjusted definitions of $\textbf{X}_i$ and $\mathbb{P}_i$. 
For instance, if $\mathscr{G}_i$ is \gls{lti} with polytopic uncertainty, then $\textbf{X}_i$ remains the same, and $\mathbb{P}_i {=} \{\textbf{X}_i | \hspace{-2.5pt}\text{\cite[Equation 18]{walsh2019interior} holds}\}$. Alternatively, if $\mathscr{G}_i$ is a linear system with input, state, and output delay, then $\textbf{X}_i$ remains the same, and $\mathbb{P}_i = \{\textbf{X}_i |\hspace{-2.5pt}\text{\cite[Theorem 3.1]{bridgeman2016conic} holds}\}$.
Critically, $\bbQ$ only depends on $(\bQ_i,\bS_i,\bR_i)$, so these variations in $\textbf{X}_i$ and $\mathbb{P}_i$ do not affect \cref{opt:Main Chordal Problem: Stability,opt:Main Chordal Problem: Stability,opt:Main Chordal Problem: Equality}. Therefore, the results readily extend to these and other structured nonlinear cases.

\section{NUMERICAL EXAMPLE} \label{Chap:Numerical Example}

\subsection{Large-Scale Network Stability Analysis} \label{subChap:Large Scale Network Stability Analysis}



\begin{figure}
\centering
  \subfloat[Graph $\cH$ of \glspl{uav}]{\includegraphics[width=0.24\textwidth]{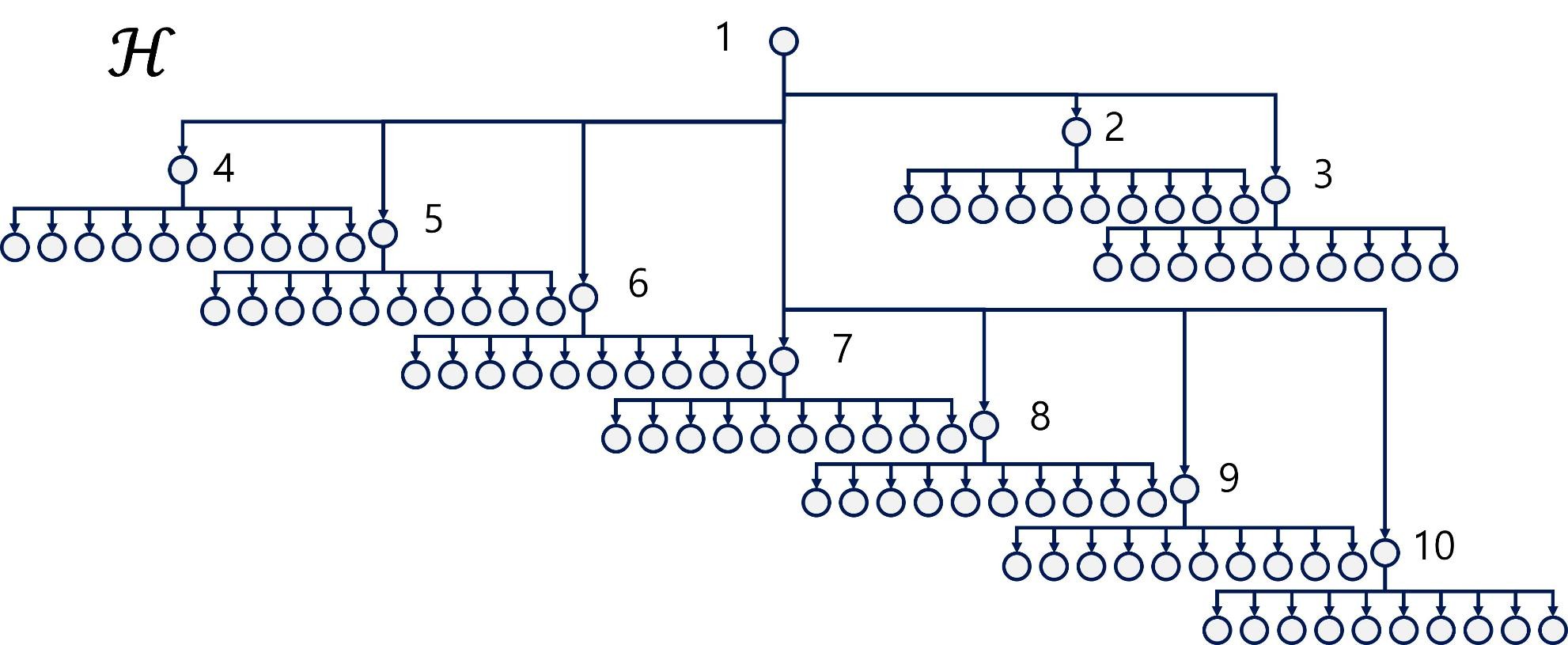}\label{fig1:Interconnection_Graph}}
  \subfloat[Graph $\barcQ$ of \glspl{uav}]{\includegraphics[width=0.24\textwidth]{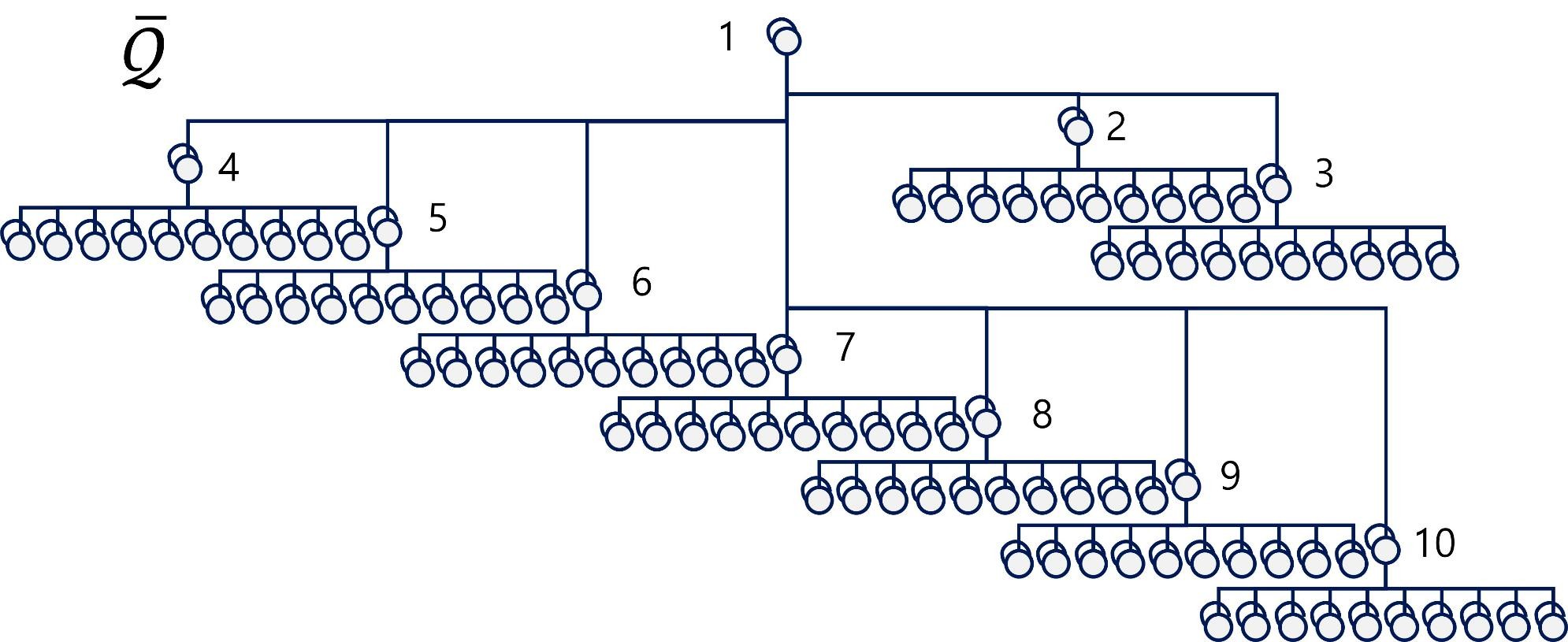}\label{fig1:barQ_Graph}}
  \vspace{0.1cm}
  \subfloat[Overlapped graph $\barcQ_o$ of $\barcQ$]{\includegraphics[width=0.24\textwidth]{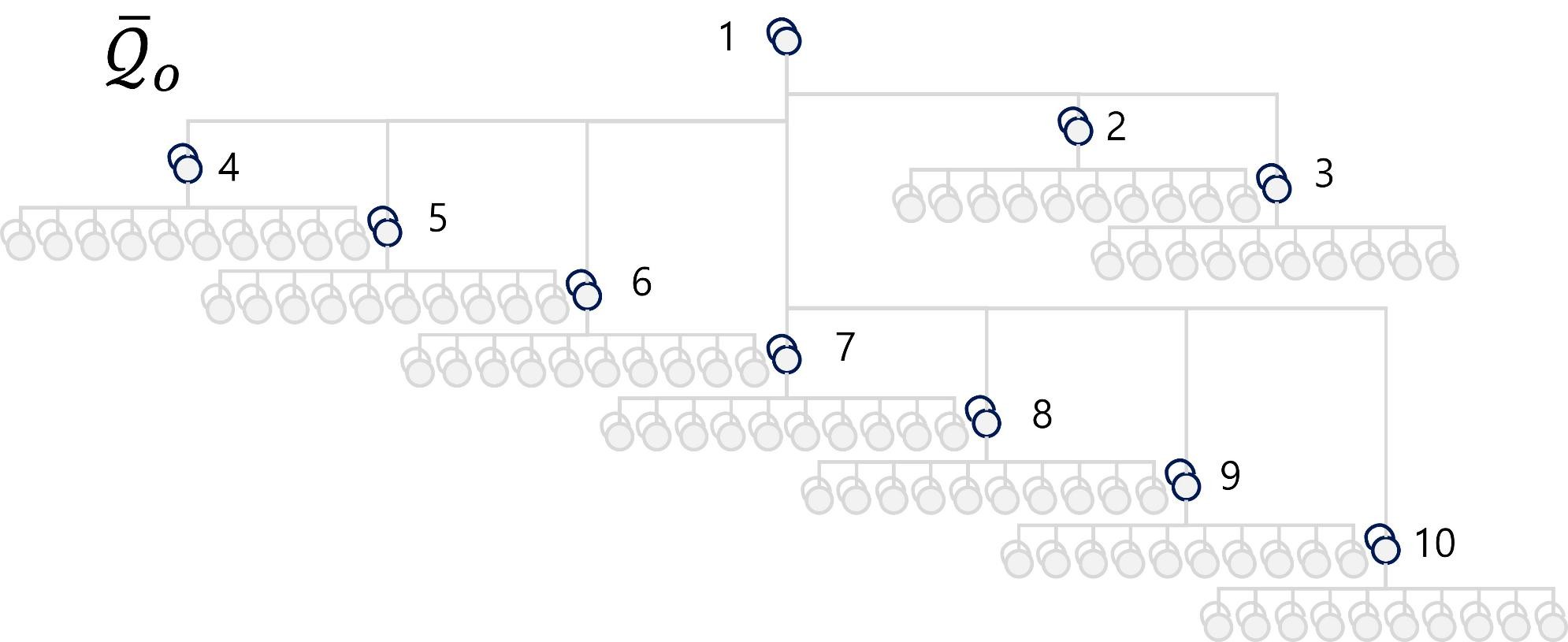}\label{fig1:Overlapped_Graph}}
  \subfloat[$\hspace{-2.5pt}\cV\hspace{-1pt}(\hspace{-1.5pt}\cR_r\hspace{-1.5pt})$ of disconnected graph $\barcQ_o\hspace{-1pt}$]{\includegraphics[width=0.24\textwidth]{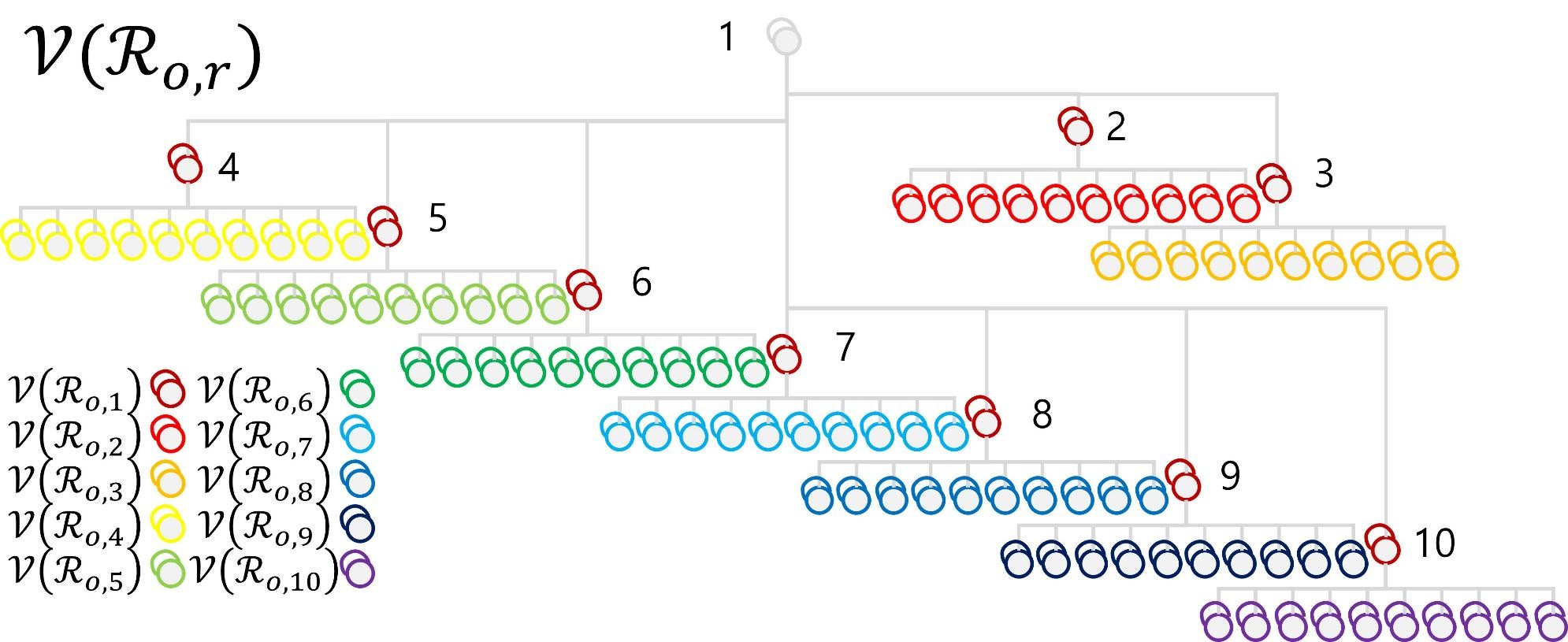}\label{fig1:VR_Graph}}
   \caption{Graph network of a \glspl{uav}}\label{fig1:Network_Structure_of_UAC}
    \vspace*{-1.25\baselineskip} 
\end{figure}

The stability of a 2D swarm of \glspl{uav} is analyzed using \cref{alg:01,alg:02} to demonstrate the proposed approach. 
The mass, moment of inertia, and wing length of each \gls{uav} are independently sampled from Gaussian distribution with mean values $m_{a}{=}3$ kg, $I_{xx,a}{=}1$ kg-m$^2$, and $l_{a}{=}0.2$ m, respectively. 
The standard deviations of all distributions are 10\% of their mean values.
The average open-loop dynamics of the system is given by $\dbx{=}\bA_{a}\bx{+}\bB_{a}\bu$ and $\by{=}\bC\bx$, where
\begin{align*}
    \bA_{a}&= \begin{bmatrix}  \bzero_{3\times 2} & \bzero_{3\times 1} & \bI_3 \\ \bzero_{1\times 2} & -g & \bzero_{1\times 3} \\ \bzero_{2\times 2} & \bzero_{2\times 1} & \bzero_{2\times 3}\end{bmatrix},
    \bB_{a}=\begin{bmatrix}
        \bzero_{4\times 1} & \bzero_{4\times 1} \\
        \frac{1}{m} & \frac{1}{m} \\
        -\frac{l}{I_{xx}} & \frac{l}{I_{xx}}
    \end{bmatrix},
    \bC=\bI_6.
\end{align*}
Each \gls{uav} is equipped with a state feedback controller $\bu{=}{-}\bK\bx$, where $\bK$ is obtained independently for each \gls{uav} using a linear-quadratic regulator. 
The average value is
\begin{align*}
    \bK_{a}&=\begin{bmatrix*}[r]
        7.07 & 7.07 & -49.00 & 8.70 & 5.12 & -15.81 \\
        -7.07 & 7.07 & 49.00 & -8.70 & 5.12 & 15.81 
    \end{bmatrix*}.
\end{align*}


\begin{figure}
    \centering
    \subfloat[Results using \cref{alg:01}]{%
    \begin{minipage}[t]{0.225\textwidth}
        \resizebox{0.85\textwidth}{!}{
%
%
\definecolor{mycolor1}{rgb}{0.00000,0.44700,0.74100}%
\definecolor{mycolor2}{rgb}{0.85000,0.32500,0.09800}%

\begin{tikzpicture}

\begin{axis}[%
width=5in,
height=2.5in,
at={(0.758in,0.509in)},
scale only axis,
bar shift auto,
xmin=0,
xmax=110,
xtick={10, 20, 30, 40, 50, 60, 70, 80, 90, 100},
xlabel style={font=\color{white!15!black}},
xlabel={Number of UAVs},
ymin=0,
ymax=3000,
ylabel style={at={(axis description cs:-0.125,.5)},font=\color{white!15!black}},
ylabel={Iterations},
axis background/.style={fill=white},
title style={font=\bfseries},
title={\textbf{Iterations Required}},
xmajorgrids,
ymajorgrids,
legend style={at={(0.03,0.97)}, anchor=north west, legend cell align=left, align=left, draw=white!15!black},
title style={font=\Huge},xlabel style={font={\Huge}},ylabel style={font=\Huge},legend style={font=\Huge},tick label style={font=\LARGE}
]
\addplot[ybar, bar width=5, fill=blue, draw=black, area legend] table[row sep=crcr] {%
10  280\\
20  1693\\
30  2178\\
40  1408\\
50  1672\\
60  1815\\
70  1964\\
80  2082\\
90  1991\\
100 2072\\
};


\end{axis}
\end{tikzpicture}%
}\\
        \resizebox{0.95\textwidth}{!}{
%
%
\definecolor{mycolor1}{rgb}{0.00000,0.44700,0.74100}%
\definecolor{mycolor2}{rgb}{0.85000,0.32500,0.09800}%
\begin{tikzpicture}

\begin{axis}[%
width=5in,
height=2.5in,
at={(0.758in,0.509in)},
scale only axis,
bar shift auto,
xmin=0,
xmax=110,
xtick={10, 20, 30, 40, 50, 60, 70, 80, 90, 100},
xlabel style={font=\color{white!15!black}},
xlabel={Number of UAVs},
ymin=0,
ymax=25,
ylabel style={at={(axis description cs:-0.075,.5)},font=\color{white!15!black}},
ylabel={Time [s]},
axis background/.style={fill=white},
title style={font=\bfseries},
title={\textbf{Average Computation Time per Iteration}},
xmajorgrids,
ymajorgrids,
legend style={at={(0.03,0.97)}, anchor=north west, legend cell align=left, align=left, draw=white!15!black},
title style={font=\Huge},xlabel style={font={\Huge}},ylabel style={font=\Huge},legend style={font=\Huge},tick label style={font=\LARGE}
]
\addplot[ybar, bar width=5, fill=blue, draw=black, area legend] table[row sep=crcr] {%
10  0.12202428571428571428571428571429\\
20  0.35406810395747194329592439456586\\
30  0.8341995360269375573921028466483\\
40  2.0655918635830965909090909090909\\
50  4.0500448403110167464114832535885\\
60  5.7703209451699173553719008264463\\
70  7.809857110997963340122199592668\\
80  10.550048030739673390970220941402\\
90  18.148664855265344048216976393772\\
100 23.599040015926737451737451737452\\
};


\end{axis}
\end{tikzpicture}%
}\\
        \resizebox{0.85\textwidth}{!}{
%
%
\definecolor{mycolor1}{rgb}{0.00000,0.44700,0.74100}%
\definecolor{mycolor2}{rgb}{0.85000,0.32500,0.09800}%
\begin{tikzpicture}

\begin{axis}[%
width=5in,
height=2.5in,
at={(0.758in,0.509in)},
scale only axis,
bar shift auto,
xmin=0,
xmax=110,
xtick={10, 20, 30, 40, 50, 60, 70, 80, 90, 100},
xlabel style={font=\color{white!15!black}},
xlabel={Number of UAVs},
ymin=0,
ymax=60000,
ylabel style={font=\color{white!15!black}},
ylabel={Time [s]},
axis background/.style={fill=white},
title style={font=\bfseries},
title={\textbf{Total Computation Time}},
xmajorgrids,
ymajorgrids,
legend style={at={(0.03,0.97)}, anchor=north west, legend cell align=left, align=left, draw=white!15!black},
title style={font=\Huge},xlabel style={font={\Huge}},ylabel style={font=\Huge},legend style={font=\Huge},tick label style={font=\LARGE}
]
\addplot[ybar, bar width=5, fill=blue, draw=black, area legend] table[row sep=crcr] {%
10	34.1668\\
20	599.4373\\
30	1816.88658946667\\
40	2908.353343925\\
50	6771.67497300002\\
60	10473.1325154834\\
70	15338.559366 \\
80	21965.2070960001\\
90	36133.9917268333\\
100	48897.2109130002\\
};


\end{axis}
\end{tikzpicture}%
}
    \end{minipage} \label{fig:alg1_result}
    }
    \subfloat[Results using \cref{alg:02}]{%
    \begin{minipage}[t]{0.225\textwidth}
        \resizebox{0.85\textwidth}{!}{
%
%
\definecolor{mycolor1}{rgb}{0.00000,0.44700,0.74100}%
\definecolor{mycolor2}{rgb}{0.85000,0.32500,0.09800}%

\begin{tikzpicture}

\begin{axis}[%
width=5in,
height=2.5in,
at={(0.758in,0.509in)},
scale only axis,
bar shift auto,
xmin=0,
xmax=110,
xtick={10, 20, 30, 40, 50, 60, 70, 80, 90, 100},
xlabel style={font=\color{white!15!black}},
xlabel={Number of UAVs},
ymin=0,
ymax=11000,
ylabel style={at={(axis description cs:-0.125,.5)},font=\color{white!15!black}},
ylabel={Iterations},
axis background/.style={fill=white},
title style={font=\bfseries},
title={\textbf{Iterations Required}},
xmajorgrids,
ymajorgrids,
legend style={at={(0.03,0.97)}, anchor=north west, legend cell align=left, align=left, draw=white!15!black},
title style={font=\Huge},xlabel style={font={\Huge}},ylabel style={font=\Huge},legend style={font=\Huge},tick label style={font=\LARGE}
]
\addplot[ybar, bar width=5, fill=red, postaction={
        pattern=north east lines
    }, draw=black, area legend] table[row sep=crcr] {%
10  408\\
20  4438\\
30  4869\\
40  5526\\
50  6363\\
60  7034\\
70  7688\\
80  8290\\
90  9675\\
100 10216\\
};


\end{axis}
\end{tikzpicture}%
}
        \resizebox{0.95\textwidth}{!}{
%
%
\definecolor{mycolor1}{rgb}{0.00000,0.44700,0.74100}%
\definecolor{mycolor2}{rgb}{0.85000,0.32500,0.09800}%
\begin{tikzpicture}

\begin{axis}[%
width=5in,
height=2.5in,
at={(0.758in,0.509in)},
scale only axis,
bar shift auto,
xmin=0,
xmax=110,
xtick={10, 20, 30, 40, 50, 60, 70, 80, 90, 100},
xlabel style={font=\color{white!15!black}},
xlabel={Number of UAVs},
ymin=0,
ymax=12,
ylabel style={at={(axis description cs:-0.075,.5)},font=\color{white!15!black}},
ylabel={Time [ms]},
axis background/.style={fill=white},
title style={font=\bfseries},
title={\textbf{Average Computation Time per Iteration}},
xmajorgrids,
ymajorgrids,
legend style={at={(0.03,0.97)}, anchor=north west, legend cell align=left, align=left, draw=white!15!black},
title style={font=\Huge},xlabel style={font={\Huge}},ylabel style={font=\Huge},legend style={font=\Huge},tick label style={font=\LARGE}
]
\addplot[ybar, bar width=5, fill=red, postaction={
        pattern=north east lines
    }, draw=black, area legend] table[row sep=crcr] {%
10  9.431863 \\
20  8.19484\\
30  8.5709591\\
40  9.698064\\
50  9.0434543\\
60  8.4613165\\
70  9.1975416\\
80  8.131568\\
90  8.826367\\
100 9.3510082\\
};


\end{axis}
\end{tikzpicture}%
}
        \resizebox{0.85\textwidth}{!}{
%
%
\definecolor{mycolor1}{rgb}{0.00000,0.44700,0.74100}%
\definecolor{mycolor2}{rgb}{0.85000,0.32500,0.09800}%
\begin{tikzpicture}

\begin{axis}[%
width=5in,
height=2.5in,
at={(0.758in,0.509in)},
scale only axis,
bar shift auto,
xmin=0,
xmax=110,
xtick={10, 20, 30, 40, 50, 60, 70, 80, 90, 100},
xlabel style={font=\color{white!15!black}},
xlabel={Number of UAVs},
ymin=0,
ymax=120,
ylabel style={font=\color{white!15!black}},
ylabel={Time [s]},
axis background/.style={fill=white},
title style={font=\bfseries},
title={\textbf{Total Computation Time}},
xmajorgrids,
ymajorgrids,
legend style={at={(0.03,0.97)}, anchor=north west, legend cell align=left, align=left, draw=white!15!black},
title style={font=\Huge},xlabel style={font={\Huge}},ylabel style={font=\Huge},legend style={font=\Huge},tick label style={font=\LARGE}
]
\addplot[ybar, bar width=5, fill=red, postaction={
        pattern=north east lines
    }, draw=black, area legend] table[row sep=crcr] {%
10	3.8482\\
20	36.3687\\
30	41.7320\\
40	53.5915\\
50	57.5435\\
60	59.5169\\
70	70.7107\\
80	67.4107\\
90	85.3951\\
100	95.5299\\
};


\end{axis}
\end{tikzpicture}%
}
    \end{minipage} \label{fig:alg2_result}
    }
    \caption{Analysis results: Although \cref{alg:02} requires more iteration than \cref{alg:01}, its average computation time per iteration is notably lower. Consequently, the total computation time of \cref{alg:02} is much less than \cref{alg:01}. This advantage becomes more pronounced for larger networks.} \label{fig:stability_analysis_result}
    \vspace*{-0.75\baselineskip} 
\end{figure}
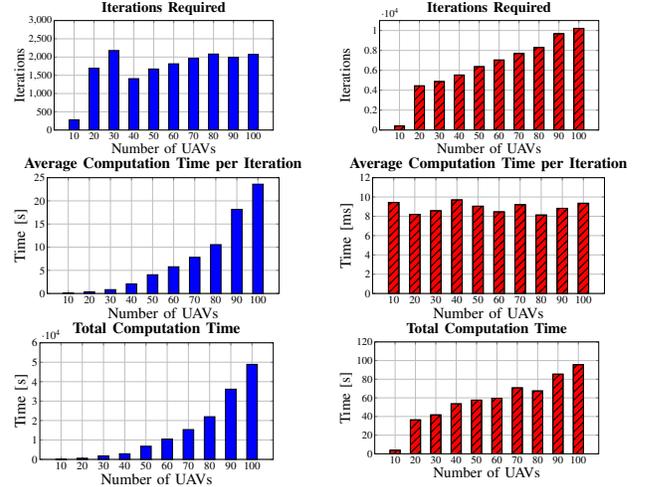

    

\begin{figure}
    \centering
    \includegraphics[width = 0.48\textwidth]{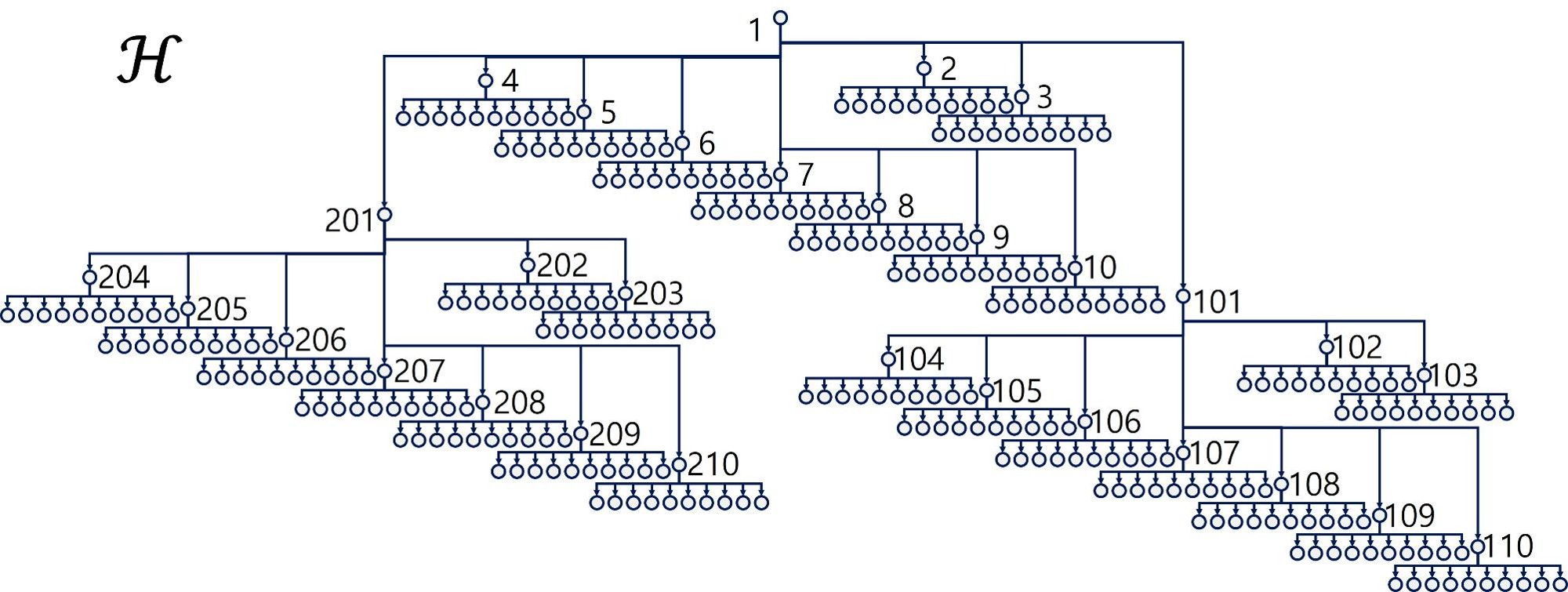}
    \caption{Larger-scale network of \glspl{uav}}
    \label{fig:Extreme Network}
    \vspace*{-1.25\baselineskip} 
\end{figure}

The heterogeneous \glspl{uav} network $\cH$, shown in \cref{fig1:Interconnection_Graph}, consists of 100 \glspl{uav}.
Each node in \cref{fig1:Interconnection_Graph} refers to each \gls{uav}.
The tail of every edge in $\edge{\cH}$ receives state information from the head vertex in the same edge.
Consequently, the closed-loop dynamics of the $i^\thh$ agent is
\begin{align} \label{eq:dotbx_i}
    \dot{\bx}_i&=(\bA_i-\bB_i\bK_i)\bx_i+\bB_i\bK_i\Big(\bx_i^d+\hspace{-8pt}\sum_{j\in\cN_\cH^+(i)}\hspace{-8pt}(\bH)_{i,j}\bx_j\Big),
\end{align}
where $\bx_i^d$ denotes the desired state of $i^\thh$ agent. 
The variables in \cref{eq:dotbx_i} correspond to those in \cref{eq:Interconnected Subsystem} through $\bu_i{=}\bx_i^d$, $\by_j{=}\bx_j$, and $\be_i{=}\bx_i^d{+}\sum_{j\in\cN_\cH^+(i)}(\bH)_{i,j}\bx_j$.

\cref{fig1:barQ_Graph} results from applying \cref{thm:supergraph} to $\cH$.
The graph $\barbQ$ has 99 maximal cliques, each consisting of two vertices. 
Since $\cH$ is a directed graph without a directed cycle, it satisfies \cref{cor:Chordal_graph_directed}, so \cref{thm:VNDT_Chordal_Decomposition} can be applied for stability analysis of the multi-agent \gls{uav} system without using chordal expansion to $\barcQ$.
This yields the overlapping graph $\barcQ_o$ described in \cref{fig1:Overlapped_Graph}.
The graph $\barcQ_o$ is a disconnected graph with 10 components, and its $\vertex{\cR_{o,r}}$ for all $r{\in}\bbN_{10}$ are shown in \cref{fig1:VR_Graph}.
For $r{\neq} l$, $\vertex{\cR_{o,r}}{\cap}\vertex{\cR_{o,l}}{=}\emptyset$, so it satisfies assumptions in \cref{thm:Separable_Equality,thm:Decompose_Equality}.
Therefore, \cref{alg:02} can solve \cref{opt:Main Problem} of this network with decomposable equality constraints.

\cref{alg:01,alg:02} were initialized with $\bX_i{=}100\bI$ for $i{\in}\bbN_{10}$, $\bX_i{=}50\bI$ for $i{\in}\bbN_{100}{\setminus}\bbN_{10}$, and $\bY_{j,j}^p{=}100\bI$ for $p{\in}\bbN_{99}$ and $j{\in}\bbN_{10}$. 
The initial values of $\bZ$, $\bW$, $\bJ$, $\bT$, $\bU$, and $\bV$ followed the initialization rules in \cref{Chap:ISNSA,subChap:DNSA}. Both algorithms successfully found feasible $\bX_i$ with $\barbQ(\bX){\prec}0$ without sharing dynamics information.
Furthermore, when \cref{alg:02} is used, even the communication of dissipativity parameters was restricted based on the graph structure of maximal cliques of $\barcQ$, $\barcQ_o$, and $\vertex{\cR_{o,r}}$ for $r{\in}\bbN_{10}$.

\cref{fig:stability_analysis_result} presents the results of the stability analysis of the \gls{uav} network.
Simulations were performed using \glspl{uav}, with their numbers varying from 10 to 100 in increments of 10. 
In the initial simulation with 10 \glspl{uav}, the 10 vertices numbered in \cref{fig1:Interconnection_Graph} were used.
For each subsequent test, the network was expanded by adding 10 tail vertices to each numbered vertex.
As a result, the network in the final simulation reached the network in \cref{fig1:Interconnection_Graph}.
The stability analysis was carried out using MOSEK \cite{aps2019mosek}, YALMIP \cite{lofberg2004yalmip}, and MATLAB.
In this simulation, only the computation time for solving each \gls{admm} step was measured, as projection time dominate all other operations.

As shown in the results, \cref{alg:02} requires more iterations to converge than \cref{alg:01}.
However, its total computation time is considerably lower, and this advantage does stand out in larger networks.
This efficiency arises because all maximal cliques in the network consist of 4 vertices, so the size of \cref{eq:Z Chordal proj} for all $p{\in}\bbN_{99}$ remains constant with $n_q{=}n_m{=}4$.
Therefore, the average computation time per iteration of \cref{alg:01} is approximately $N^4$ times greater than that of \cref{alg:02}, which far outweighs the difference in required iterations.

\begin{figure}
\centering
    \resizebox{0.4\textwidth}{!}{
%
%
\definecolor{mycolor1}{rgb}{1.00000,0.00000,1.00000}%
\definecolor{mycolor2}{rgb}{0.00000,1.00000,1.00000}%
\begin{tikzpicture}

\begin{axis}[%
width=3.7in,
height=2in,
at={(0.758in,0.592in)},
scale only axis,
xmin=5,
xmax=305,
xlabel style={font=\color{white!10!black}},
xlabel={Number of UAVs},
ymin=-2,
ymax=6,
ylabel style={font=\color{white!10!black}},
ylabel={log(Time) [log(s)]},
axis background/.style={fill=white},
title style={font=\bfseries},
xmajorgrids,
ymajorgrids,
legend columns=-1,
legend style={at={(0.00,1)}, anchor=north west, legend cell align=left, align=left, draw=white!2.5!black},
title style={font=\huge},xlabel style={font={\LARGE}},ylabel style={font=\LARGE},legend style={font=\large},
]
\addplot [color=blue, dotted, line width=1.0pt, mark size=4pt, mark=o, mark options={solid, blue}]
  table[row sep=crcr]{%
10	0.985219831897741\\
20	1.91722555409355\\
30  3.047002130518505\\
40  3.738867725635786\\
50  4.228220648868423\\
60  4.691452827587458\\
};
\addlegendentry{SDP}

\addplot [color=red, dashed, line width=1.0pt, mark size=2pt, mark=*, mark options={solid, red}]
  table[row sep=crcr]{%
10	-0.530914700876880\\ 
20	-0.180238971485901\\ 
30	0.147150525210071\\  
40	0.492033629731959\\  
50	0.838584657086384\\  
60	0.883757934316753\\  
70	1.069449785027674\\  
80	1.286061068353925\\  
90	1.313981049635012\\  
100	1.548605590147786\\  
110 1.603614152872576\\ 
120	1.749551585216718\\ 
130	1.984929917927822\\ 
140	1.98989323016150997\\  
150	2.122962223108981968\\  
160	2.181370463636219\\  
170	2.201684100868195\\  
180	2.464998807016851\\  
190	2.493856888926965\\  
200	2.541804195146257\\  
    210 2.477425588494946\\ 
    220	2.5671292193926903\\ 
    230	2.5671705185850373\\ 
240	2.6743621041059377\\  
250	2.666422406533684\\  
260	2.73202984555469018\\  
270	2.7895632977345997\\  
280	2.8438338318020638\\  
290	2.9414967737059859\\  
300	2.9811353808428838\\  
};
\addlegendentry{VNDT}

\addplot [color=green, dashdotted, line width=1.0pt, mark size=4pt, mark=x, mark options={solid, green}]
  table[row sep=crcr]{%
10   0.121986379050838\\
20   0.274665769081389\\
30   0.370235437283177\\
40   0.448644494226009\\
50   0.507855871695831\\
60   0.572650915252359\\
70   0.623848616671376\\
80   0.669772325387642\\
90   0.711089840140532\\
100   0.748815749085026\\
110   0.784460376298138\\
120   0.817823093451339\\
130   0.848503173886313\\
140   0.877123600647618\\
150   0.904087717349737\\
160   0.929286062374987\\
170   0.953107035221487\\
180   0.975914090049001\\
190   0.997517439414725\\
200   1.017884390216724\\
210   1.037772998391635\\
220   1.056195686916524\\
230   1.073670703612163\\
240   1.090476732698775\\
250   1.106598813212537\\
260   1.122563174958190\\
270   1.137708492852578\\
280   1.152162931391168\\
290   1.166205063173285\\
300   1.179709739929585\\
};
\addlegendentry{SSA\cite{agarwal2020distributed}}

\addplot [color=mycolor1, dash dot dot, line width=1.0pt, mark size=4.0pt, mark=+, mark options={solid, mycolor1}]
  table[row sep=crcr]{%
10	1.5336043054765379900623045936771\\
20	2.7777437634347297622003053334992\\
30	3.2593278193795845\\
40	3.4636471690352146994003517207326\\
50	3.8306961046613297924640894380712\\
60	4.0200765986526337675931294677766\\
70	4.1857845715570669683899454918595\\
80	4.3417353023077389332434990006479\\
90	4.5579159407869196229037095311218\\
100	4.6892840877608670232167666822887\\    
110 4.9923723399260995068668763307728 \\
};
\addlegendentry{alg1}

\addplot [color=orange, loosely dashed, line width=1.0pt, mark size=4pt, mark=asterisk, mark options={solid, orange}]
  table[row sep=crcr]{%
10	0.585257635257462\\
20	1.560727777466126\\
30	1.620469198710031\\
40	1.728760286174553\\
50	1.759996273690194\\
60	1.774640302444487\\
70	1.849485136559961\\
80	1.828728836923080\\
90	1.931432951438255\\
100	1.980139323129708\\
110 2.009926149878038\\
120	2.215018060521491\\ 
130	2.409923828991956\\ 
140	2.449935748850888\\ 
150	2.4431188138831775\\ 
160	2.5080617755711972\\ 
170 2.5652478892877731\\ 
180	2.5710301970200071\\ 
190	2.6105603064578749\\ 
200	2.610102886894875\\ 
210 2.589480067953295\\ 
220 2.582346471545459\\ 
230 2.60833205531362847\\
240 2.6229583190468295319183757750192\\ 
250 2.5680723554711321 \\
260 2.64050250520770980 \\
270 2.6369267833464274391556477585167 \\
280 2.6548218225221809201764066310759 \\ 
290 2.726902683949696829 \\
300 2.7162918030918044 \\
};
\addlegendentry{alg2}

\end{axis}
\end{tikzpicture}%
}
    \caption{Computation time of stability analysis for larger-scale network: \gls{sdp} fails to analyze the stability of networks with more than 60 agents due to the memory issue.} \label{fig:Comparison_Large_Network}
    \vspace*{-1.25\baselineskip} 
\end{figure}
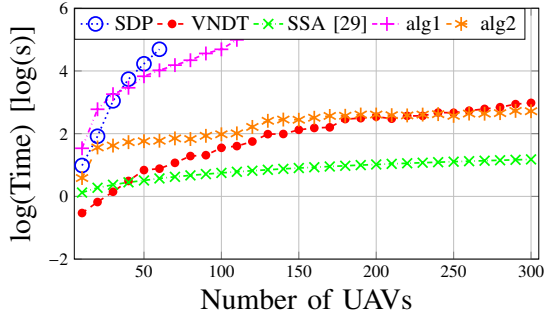

\begin{table}
\caption{Computational complexity of algorithms: $k_1$ and $k_2$ are the required iterations of \cref{alg:01,alg:02}, respectively.}
\label{tab:Complexity}
\begin{center}
\renewcommand{\arraystretch}{1.3}
\begin{tabular}{p{1.3cm}|p{1.3cm}|p{1.3cm}|p{1.3cm}|p{1.3cm}}
\hline
\gls{sdp} & \gls{vndt} & SSA & \cref{alg:01} & \cref{alg:02} \\
\hline\hline
$\cO(N^6n_m^6)$ & $\cO(N^4n_m^6)$ & $N\cO(n_m^6)$ & $k_1\cO(N^4n_m^6)$ & $k_2\cO(n_m^6)$\\
\hline
\end{tabular}
\end{center}
\vspace*{-1.75\baselineskip} 
\end{table}

The convergence results of \cref{alg:01,alg:02} are compared with those of other centralized approaches and the algorithm in \cite{agarwal2020distributed}.
The centralized \gls{sdp} determines the dissipativity of the entire network by applying \cref{lem:KYP Lemma} with a constraint, $\bQ{\prec}0$, treating the network as a single dynamical system.
The centralized \gls{vndt} solves \cref{opt:Main Problem} in a fully centralized manner.
The \gls{ssa} method proposed in \cite{agarwal2020distributed} determines the network dissipativity sequentially.
Specifically, it verifies the $\bQ\bS\bR$-dissipativity of one vertex and then extends the analysis by progressively incorporating additional agents into the network while maintaining the dissipativity parameters from the previous step.
The computational complexities of algorithms are summarized in \cref{tab:Complexity}.

The results presented in \cref{fig:Comparison_Large_Network} illustrate the total computation time of different methods for analyzing the stability of the larger-scale network in \cref{fig:Extreme Network}.
The centralized \gls{sdp} is both computationally expensive and memory-intensive, requiring more than 1TB of RAM for networks with $N{>}60$.
Although \gls{vndt} provides far faster results than the centralized \gls{sdp}, both algorithms force agents to share their dynamics information with the entire network, which reduces information privacy.
\cref{alg:01} ensures information security and outperforms centralized \gls{sdp} for $N{>}30$, but remains consistently slower than \gls{vndt} and becomes impratically slow for $N{>}100$. 
In contrast, \cref{alg:02} achieved acceptable computation times across all tested networks and outperformed \gls{vndt} for $N{>}240$, successfully combining both information security and computational efficiency.


The sequential nature of \gls{ssa} enables a distributed and computationally efficient operation, making it appear to be the most effective method. 
However, it can become conservative when applied to networks with long paths, as the dissipativity parameters fixed in one step are reused to verify the dissipativity of subsequent vertices.
A detailed investigation of this limitation is presented in the next section.

\begin{figure}
    \centering
    \includegraphics[width = 0.25\textwidth]{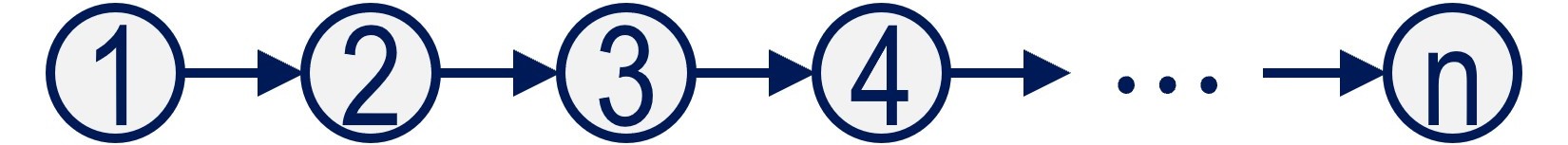}
    \caption{Series network of \glspl{uav}}
    \label{fig:Series Network}
    \vspace*{-1.25\baselineskip} 
\end{figure}

\subsection{Series Network Stability Analysis} \label{subChap:Series Network Stability Analysis}
In this session, the same stability analysis methods are applied to a series network of \glspl{uav}, following the dynamics in \cref{subChap:Large Scale Network Stability Analysis}.
The network is illustrated in \cref{fig:Series Network}.
The computation time is measured for different numbers of \glspl{uav}, denoted by $n$.
Unlike hierarchical large-scale networks, such as \cref{fig1:Network_Structure_of_UAC}, series networks have longer paths, making conservative algorithms like \gls{ssa} in \cite{agarwal2020distributed} unsuitable.

\begin{figure}
\centering
    \resizebox{0.4\textwidth}{!}{
%
%
\definecolor{mycolor1}{rgb}{1.00000,0.00000,1.00000}%
\definecolor{mycolor2}{rgb}{0.00000,1.00000,1.00000}%
\begin{tikzpicture}

\begin{axis}[%
width=3.7in,
height=2in,
at={(0.758in,0.592in)},
scale only axis,
xmin=1,
xmax=11,
xlabel style={font=\color{white!10!black}},
xlabel={Number of UAVs},
ymin=-2,
ymax=6,
ylabel style={font=\color{white!10!black}},
ylabel={log(Time) [log(s)]},
axis background/.style={fill=white},
title style={font=\bfseries},
xmajorgrids,
ymajorgrids,
legend columns=-1,
legend style={at={(0.00,1)}, anchor=north west, legend cell align=left, align=left, draw=white!2.5!black},
title style={font=\huge},xlabel style={font={\LARGE}},ylabel style={font=\LARGE},legend style={font=\large},
]
\addplot [color=blue, dotted, line width=1.0pt, mark size=4pt, mark=o, mark options={solid, blue}]
  table[row sep=crcr]{%
2	-1.16585280354752\\
3	-0.951799106349196\\
4	-0.494504081082235\\
5	-0.347724642517357\\
6	0.0150167117706854\\
7	0.373815352260072\\
8	0.53330880112994\\
9	0.866932805754846\\
10	0.971819397332452\\
};
\addlegendentry{SDP}

\addplot [color=red, dashed, line width=1.0pt, mark size=2pt, mark=*, mark options={solid, red}]
  table[row sep=crcr]{%
2	-0.874310583089405\\
3	-1.11423181962457\\
4	-0.682366996742804\\
5	-0.997890119272839\\
6	-0.84785236535403\\
7	-0.785823995728657\\
8	-0.667832301030643\\
9	-0.503049605870625\\
10	-0.41402744313359\\
};
\addlegendentry{VNDT}

\addplot [color=green, dashdotted, line width=1.0pt, mark size=4pt, mark=x, mark options={solid, green}]
  table[row sep=crcr]{%
2	-0.00747873514286651\\
3	0.145850619045911\\
4	0.205696862618955\\
5	0.398117511977036\\
6	0.423155400153547\\
7	0.442514260397309\\
8	0.460301142192809\\
9	0.477139060428406\\
};
\addlegendentry{SSA\cite{agarwal2020distributed}}

\addplot [color=mycolor1, dash dot dot, line width=1.0pt, mark size=4.0pt, mark=+, mark options={solid, mycolor1}]
  table[row sep=crcr]{%
2	0.616611505014178\\
3	0.933141688856634\\
4	1.48215518956029\\
5	2.25597628540671\\
6	2.51879114183915\\
7	2.7887368857395\\
8	3.29325530366367\\
9	3.86029023945833\\
10	4.24213752664932\\
};
\addlegendentry{alg1}

\addplot [color=orange, loosely dashed, line width=1.0pt, mark size=4pt, mark=asterisk, mark options={solid, orange}]
  table[row sep=crcr]{%
2	0.616611505014178\\
3	0.87796499892224\\
4	1.78877211515704\\
5	1.79780369906744\\
6	2.31733296671536\\
7	2.71139685187603\\
8	2.9922987350676\\
9	3.56331916766448\\
10	3.61954099495138\\
};
\addlegendentry{alg2}

\end{axis}
\end{tikzpicture}%
}
    \caption{Computation time of stability analysis for series network: \gls{ssa} fails to analyze the stability of networks with more than 9 agents because of the conservative nature.} \label{fig:Comparison_Series_Network}
    \vspace*{-1.25\baselineskip} 
\end{figure}
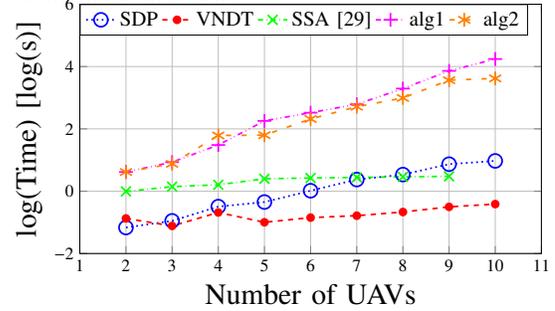

\cref{fig:Comparison_Series_Network} shows the stability analysis results for series networks. 
In this case, the centralized approaches do not exhibit numerical issues, as the network size is relatively small compared to the case in \cref{subChap:Large Scale Network Stability Analysis}.
However, \gls{ssa} shows numerical problems when applied to networks with 10 or more \glspl{uav}.
These arise because fixing the dissipativity parameters from the previous step leads to more conservative constraint sets for subsequent feasibility problems.
As this process continues, the constraints become overly restrictive, resulting in solver failures and infeasible problems.
This example highlights the limitations of \gls{ssa} for networks with long paths.
In contrast, the proposed methods successfully determine the stability of the network without numerical problems while maintaining a distributed implementation.

\subsection{Nonlinear Network Stability Analysis} \label{subChap:Nonlinear Netowrk Stability Analysis}
The proposed algorithms can also be used to analyze networks of nonlinear agents with disparate sources of uncertainty under a single framework. To demonstrate this, the following example considers a network consisting of a \gls{lti} agent, a linear time-delayed agent, and a nonlinear \gls{osp} agent.

All agents are based on the 2D robot manipulator dynamics, which are expressed as
\begin{align*}
    \bM(\bq)\ddot{\bq}+\bC(\bq,\dot{\bq})\dot{\bq}+\nabla\bU(\bq)=\bm{\tau},
\end{align*}
where $\bM(\bq){\in}\bbS^n$, $\bC(\bq,\dot{\bq}){\in}\bbR^{n\times n}$, $\bU(\bq){\in}\bbR^n$, $\bm{\tau}{\in}\bbR^n$, and $\bq\in\bbR^n$ denote the inertia matrix, Coriolis matrix, potential energy matrix, external torque input vector, and state vector, respectively.
It is well known that $\dot{\bM}(\bq){-}2\bC(\bq,\dot{\bq})$ is skew-symmetric, that is $\bx^T(\dot{\bM}(\bq){-}2\bC(\bq,\dot{\bq}))\bx{=}0$ for any $\bx{\in}\bbR^{n}$.

\begin{figure}
    \centering
    \includegraphics[width = 0.4\textwidth]{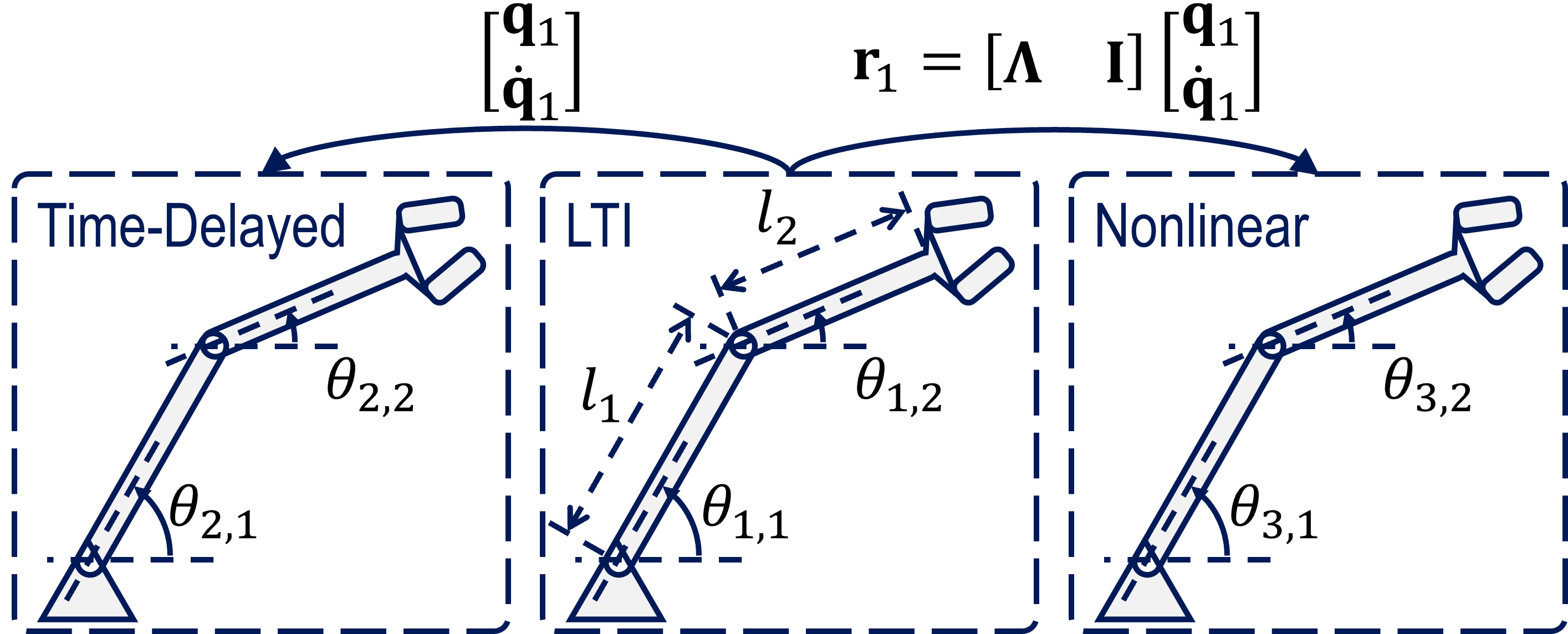}
    \caption{Interconnection of manipulators}
    \label{fig:Manipulator_Network}
    \vspace*{-1.5\baselineskip} 
\end{figure}


The network consists of three 2D robot manipulators, as depicted in \cref{fig:Manipulator_Network}.
The 2nd and 3rd agents receive the state information of the 1st agent as their input, where the state vector is $\bx_i{=}[\bq_i^T\,\dot{\bq}_i^T]^T{\in}\bbR
^4$, $\bq_i{=}[\theta_{i,1}\,\theta_{i,2}]^T{\in}\bbR^2$ for $i{\in}\bbN_3$.

The 1st agent is linearized around the equilibrium point $\theta_{1,1}{=}\theta_{1,2}{=}\frac{\pi}{2}$, and is equipped with LQR controller.
Consequently, its closed-loop dynamics are given by
\begin{align*}
    \dot{\bx}_1=(\bA_1-\bB_1\bK_1)\bx_1+\bB_1\bK_1\bx_e,\;\;
    \by_1=\begin{bmatrix}
        \bI & \bzero \\ \bzero & \bI \\ \bm\Lambda & \bI
    \end{bmatrix}\bx_1=\bC_1\bx_1.
\end{align*}
where $\bx_e$ is the exogenous input, and $\bm\Lambda\in\bbS^2$ is a constant used for the 3rd agent.
The dissipativity of the system can be determined using \cref{lem:KYP Lemma}.

The 2nd agent follows the same linearized dynamics but with input and output delays. 
Its dynamics are described by
\begin{align*}
    \dot{\bx}_2(t)&{=}(\bA_2{-}\bB_2\bK_2)\bx_2(t){+}\bB_2\bK_2\bx_1(t{-}T_i),\; 
    \by_2(t){=}\bx_2(t{-}T_o)
\end{align*}
where $T_i$ and $T_o$ are the maximum input and output delays, respectively.
Since this system is not an \gls{lti} system, distributed methods based on \cref{lem:KYP Lemma} cannot be applied.
Instead, \cite[Theorem 3.1]{bridgeman2016conic} can be used to determine the dissipativity of the time-delayed system.

The 3rd agent follows the original nonlinear manipulator dynamics with the passivation controller presented by Slotine and Li \cite{slotine1987adaptive}.
The dynamics follow
\begin{align*}
    &\bM_3(\bq_3)\dot{\br}_3+\bC_3(\bq_3,\dot{\bq}_3)\br_3+\bK_3\br_3=\bK_3\br_2, 
\end{align*}
where $\br_3=[\bm\Lambda\;\bI_2]\bx_3$, $\br_2=[\bm\Lambda\;\bI_2]\bx_2$, and $\bm\Lambda\in\bbS^2$.
This system known to be $QSR$-dissipative with respect to input $\br_2$ and output $\br_3$ \cite{hatanaka2015passivity}, with $\bQ_3{=}{-}\bK_3$, $\bS_3{=}\frac{1}{2}\bQ_3$, and $\bR_3{=}\bzero$,
where $\bK_3{\in}\bbS^2$ is a design variable, not from an LQR controller.

The mass of the first and second links and final endpoint are $3$ kg, $2$ kg, and $1$ kg, respectively.
The lengths of both links are $1$ m, and the damping coefficients at both joints are $1$ kg/s. 
The maximum input and output delays are $8$ s and $7$ s, respectively.
The matrix $\bm\Lambda$ is set to $\bI_2$. 
The network in \cref{fig:Manipulator_Network} is expressed as
\begin{align*}
    \bH=\begin{bmatrix}
        \bzero_4 & \bzero_{4\times2} & \bzero_{4\times6} \\
        \bI_4 & \bzero_{4\times2} & \bzero_{4\times6} \\
        \bzero_{2\times4} & \bI_2 & \bzero_{2\times6}
    \end{bmatrix}.
\end{align*}
With these parameters, the network is successfully proven to be stable using \cref{alg:01}.
For the linear system, the initial values are $\bP_1^0{=}\bI_4$, $\bQ_1^0{=}{-}\bI_6$, $\bS_1^0{=}\bzero_{6\times4}$, and $\bR_1^0{=}\bI_4$.
For the linear time-delayed system, the initial values are $\bP_2^0{=}\bI_4$, $\bQ_2^0{=}{-}\bI_4$, $\bS_2^0{=}0.5\bI_4$, $\bR_2^0{=}\bI_4$, $\mathscr{R}_1^0{=}\bI_4$, $\mathscr{R}_2^0{=}\bI_4$, and $\lambda{=}1$, where $\mathscr{R}_1$, $\mathscr{R}_2$, and $\lambda$ are additional design variables used to determine the dissipativity of the linear time-delyaed system \cite{bridgeman2016conic}.
For the nonlinear system, the initial values are $\bQ_3^0{=}{-}\bI_4$, $\bS_4^0{=}0.5\bI_4$, and $\bR_3^0{=}\bzero_4$.
\cref{alg:01} terminated after 6 iterations, and the resulting dissipativity parameters are summarized in \cref{tab:Manipulator Dissipativity}, thereby confirming system stability.
This example demonstrates that the proposed approaches are applicable to general dissipative systems, as discussed in \cref{Chap:Extension to Nonlinear Systems}.

\begin{table}
\caption{Dissipativity parameters for \ref{fig:Manipulator_Network}}
\label{tab:Manipulator Dissipativity}
\begin{center}
\renewcommand{\arraystretch}{1.3}
\begin{tabular}{p{.05cm}|p{7.3cm}}
\hline
  & Dissipativity parameters \\
\hline\hline
 & \setlength{\arraycolsep}{2pt}$\hspace{-3pt}\bQ_1\hspace{-1pt}{=}\hspace{-2pt}\left[\hspace{-2pt}\begin{array}{rrrrrr}
    {-}0.6148 & {-}0.0012 & {-}0.0463 & {-}0.0041 & 0.4195 & {-}0.0041\\
    {-}0.0012 & {-}0.8258 & 0.0088 & {-}0.0597 & 0.0167 & 0.2096 \\
    {-}0.0463 & 0.0088 & {-}0.9932 & 0.0011 & {-}0.0344 & 0.0175 \\
    {-}0.0041 & {-}0.0597 &	0.0011 & {-}0.9656 & 0.0066 & {-}0.0531 \\
    0.4195 & 0.0167 & {-}0.0344 & 0.0066 & {-}0.6418 & 0.0200 \\
    {-}0.0041 & 0.2096 & 0.0175 & {-}0.0531 & 0.0200 & {-}0.8180
 \end{array}\hspace{-2pt}\right]${,} \\
1 & \setlength{\arraycolsep}{2pt}$\hspace{-3pt}\bS_1\hspace{-1pt}{=}\hspace{-2pt}\left[\hspace{-2pt}\begin{array}{rrrr}
    0.0907 & {-}0.0169 & \;\;0.0136 & \;\;0.0020 \\
    {-}0.0051 & 0.1569 & \;\;0.0303 & \;\;0.0165 \\
    {-}0.0121 &	0.0328 & \;\;0.0043 & \;\;0.0030 \\
    {-}0.0116 &	0.0304 & \;\;0.0039 & \;\;0.0028 \\
    0.0786 & 0.0159 & \;\;0.0179 & \;\;0.0050 \\
    {-}0.0167 & 0.1873 & \;\;0.0342 & \;\;0.0192
 \end{array}\hspace{-2pt}\right]$, \\
 & \setlength{\arraycolsep}{2pt}$\hspace{-3pt}\bR_1\hspace{-1pt}{=}\hspace{-2pt}\left[\hspace{-2pt}\begin{array}{rrrr}
    1.0049 &	{-}0.0231 &	{-}0.0037 &	{-}0.0023 \\
    {-}0.0231 &	1.1317 &	0.0219 &	0.0131 \\
    {-}0.0037 &	0.0219 &	1.0037 &	0.0022 \\
    {-}0.0023 &	0.0131 &	0.0022 &	1.0013
 \end{array}\hspace{-2pt}\right]$.  \\
\hline
2 & $\bQ_2{=}{-}2.9721{\times}10^{-4}\bI_4$, $\bS_2{=}{-}1.7712{\times}10^{-4}\bI_4$, $\bR_2{=}0.0713\bI_4$; \\
& The system is in cone$(-16.0991,14.9072)$\\
\hline
3 & $
    \setlength{\arraycolsep}{2pt}
    \bQ_3{=}\left[\hspace{-2pt}\begin{array}{rr}
        {-}0.4179 & 0.0354 \\
        0.0354 & {-}0.9938
 \end{array}\hspace{-2pt}\right]$, $\bS_3{=}{-}\frac{1}{2}\bQ_3$, $\bR_3{=}\bzero_2$ \\
\hline
\end{tabular}
\end{center}
\vspace*{-1.75\baselineskip} 
\end{table}

\section{CONCLUSIONS} \label{Chap:Conclusion}

This paper presents a novel method for analyzing the stability of multi-agent systems without requiring agents to share their internal dynamics. The proposed approach allows each agent to independently assess its dissipativity while ensuring overall system stability. 
The first algorithm applies \gls{admm} directly to solve the feasibility problem in a distributed manner.
The second algorithm significantly reduces the computational burden of the first via chordal decomposition.
Additionally, several computational enhancements are incorporated to improve efficiency.
The effectiveness of the proposed algorithms is demonstrated on large-scale systems, involving 2D swarm \glspl{uav}, as well as nonlinear and time-delayed robot manipulator networks.
The results show that both approaches successfully verify the system-wide stability of any dynamics, including nonlinear agents, without sharing its dynamics matrices.
Moreover, combining chordal decomposition with \gls{admm} proved particularly effective for denser and larger-scale networks. Future work will extend the distributed stability analysis framework to enable distributed controller syntheses, ensuring privacy by preventing the exchange of agents' private information.

\section*{References}
\bibliographystyle{IEEEtran}
\bibliography{MyBib}{}

@STRING{IEEE_J_AC         = "{IEEE} Tran. Aut. Ctrl."}

@STRING{J_FI         = "J. Franklin Inst."}

@STRING{IEEE_P         = "Proc. {IEEE}"}

@STRING{Eu_JC         = "Eu. J. Ctrl."}

@STRING{ECCE         = "ICPE 2023-ECCE Asia"}

@STRING{FT_ML         = "F. T. Mach."}

@STRING{ARC         = "Ann. Rev. Ctrl."}

@STRING{ACC         = "Proc. Amer. Ctrl. Conf."}

@STRING{IJRNC         = "I. J. Robu. Nonl. Ctrl."}

@STRING{ARXIV         = "ArXiv"}

@inproceedings{ingyu2025consensus,
  title={Consensus-Based Stabiligy Analysis of Multi-Agent Networks},
  author={Jang, Ingyu and LoCicero, Ethan J and Bridgeman, Leila},
  booktitle={2025 IEEE 64th Conference on Decision and Control (CDC)},
  year={2025},
  organization={IEEE},
  note={Under Review}
}

@article{usova2018scattering,
  title={Scattering-based stabilization of non-planar conic systems},
  author={Usova, Anastasiia A and Polushin, Ilia G and Patel, Rajni V},
  journal={Automatica},
  volume={93},
  pages={1--11},
  year={2018},
  publisher={Elsevier}
}

@article{hirche2008stabilizing,
  title={Stabilizing interconnection characterization for multi-agent systems with dissipative properties},
  author={Hirche, Sandra and Hara, Shinji},
  journal={IFAC Proceedings Volumes},
  volume={41},
  number={2},
  pages={1571--1577},
  year={2008},
  publisher={Elsevier}
}

@article{stan2007analysis,
  title={Analysis of interconnected oscillators by dissipativity theory},
  author={Stan, Guy-Bart and Sepulchre, Rodolphe},
  journal={IEEE Transactions on Automatic Control},
  volume={52},
  number={2},
  pages={256--270},
  year={2007},
  publisher={IEEE}
}

@article{meissen2015compositional,
  title={Compositional performance certification of interconnected systems using ADMM},
  author={Meissen, Chris and Lessard, Laurent and Arcak, Murat and Packard, Andrew K},
  journal={Automatica},
  volume={61},
  pages={55--63},
  year={2015},
  publisher={Elsevier}
}

@article{chopra2006passivity,
  title={Passivity-based control of multi-agent systems},
  author={Chopra, Nikhil and Spong, Mark W},
  journal={Advances in robot control: from everyday physics to human-like movements},
  pages={107--134},
  year={2006},
  publisher={Springer}
}

@article{chopra2012output,
  title={Output synchronization on strongly connected graphs},
  author={Chopra, Nikhil},
  journal={IEEE Transactions on Automatic Control},
  volume={57},
  number={11},
  pages={2896--2901},
  year={2012},
  publisher={IEEE}
}

@article{yao2009passive,
  title={Passive stability and synchronization of complex spatio-temporal switching networks with time delays},
  author={Yao, Jing and Wang, Hua O and Guan, Zhi-Hong and Xu, Weisheng},
  journal={Automatica},
  volume={45},
  number={7},
  pages={1721--1728},
  year={2009},
  publisher={Elsevier}
}

@article{miyano2020continuous,
  title={Continuous-time optimization dynamics mirroring ADMM architecture and passivity-based robustification against delays},
  author={Miyano, Tatsuya and Yamashita, Shunya and Hatanaka, Takeshi and Shibata, Kazuki and Jimbo, Tomohiko and Fujita, Masayuki},
  journal={IEEE Transactions on Control of Network Systems},
  volume={7},
  number={3},
  pages={1296--1307},
  year={2020},
  publisher={IEEE}
}

@article{rojas2009dynamic,
  title={Dynamic operability analysis of nonlinear process networks based on dissipativity},
  author={Rojas, Osvaldo J and Setiawan, Ridwan and Bao, Jie and Lee, Peter L},
  journal={AIChE Journal},
  volume={55},
  number={4},
  pages={963--982},
  year={2009},
  publisher={Wiley Online Library}
}

@article{hioe2014decentralized,
  title={Decentralized nonlinear control of process networks based on dissipativity—A Hamilton--Jacobi equation approach},
  author={Hioe, Denny and Hudon, Nicolas and Bao, Jie},
  journal={Journal of Process Control},
  volume={24},
  number={3},
  pages={172--187},
  year={2014},
  publisher={Elsevier}
}

@article{ghanbari2016large,
  title={Large-scale dissipative and passive control systems and the role of star and cyclic symmetries},
  author={Ghanbari, Vahideh and Wu, Po and Antsaklis, Panos J},
  journal={IEEE Transactions on Automatic Control},
  volume={61},
  number={11},
  pages={3676--3680},
  year={2016},
  publisher={IEEE}
}

@article{liu2023incremental,
  title={Incremental-dissipativity-based synchronization of interconnected systems with MIMO nonlinear operators},
  author={Liu, Tao and Hill, David J and Zhao, Jun},
  journal={IFAC-PapersOnLine},
  volume={56},
  number={2},
  pages={1853--1858},
  year={2023},
  publisher={Elsevier}
}

@inproceedings{anderson2011dynamical,
  title={Dynamical system decomposition using dissipation inequalities},
  author={Anderson, James and Teixeira, Andr{\'e} and Sandberg, Henrik and Papachristodoulou, Antonis},
  booktitle={2011 50th IEEE Conference on Decision and Control and European Control Conference},
  pages={211--216},
  year={2011},
  organization={IEEE}
}

@book{boyd2004convex,
  title={Convex optimization},
  author={Boyd, Stephen and Vandenberghe, Lieven},
  year={2004},
  publisher={Cambridge university press}
}

@book{diestel2024graph,
  title={Graph theory},
  author={Diestel, Reinhard},
  year={2024},
  publisher={Springer (print edition); Reinhard Diestel (eBooks)}
}

@article{yannakakis1981computing,
  title={Computing the minimum fill-in is NP-complete},
  author={Yannakakis, Mihalis},
  journal={SIAM Journal on Algebraic Discrete Methods},
  volume={2},
  number={1},
  pages={77--79},
  year={1981},
  publisher={SIAM}
}

@article{sztipanovits2011toward,
  title={Toward a science of cyber--physical system integration},
  author={Sztipanovits, Janos and Koutsoukos, Xenofon and Karsai, Gabor and Kottenstette, Nicholas and Antsaklis, Panos and Gupta, Vijay and Goodwine, Bill and Baras, John and Wang, Shige},
  journal=IEEE_P,
  volume={100},
  number={1},
  pages={29--44},
  year={2011},
  publisher={IEEE}
}

@book{vidyasagar1981input,
  title={Input-output analysis of large-scale interconnected systems: decomposition, well-posedness and stability},
  author={Vidyasagar, Mathukumalli},
  year={1981},
  publisher={Springer}
}

@article{zheng2021chordal,
  title={Chordal and factor-width decompositions for scalable semidefinite and polynomial optimization},
  author={Zheng, Yang and Fantuzzi, Giovanni and Papachristodoulou, Antonis},
  journal=ARC,
  volume={52},
  pages={243--279},
  year={2021},
  publisher={Elsevier}
}

@article{boyd2011distributed,
  title={Distributed optimization and statistical learning via the alternating direction method of multipliers},
  author={Boyd, Stephen and Parikh, Neal and Chu, Eric and Peleato, Borja and Eckstein, Jonathan and others},
  journal=FT_ML,
  volume={3},
  number={1},
  pages={1--122},
  year={2011},
  publisher={Now Publishers, Inc.}
}

@article{agarwal2020distributed,
  title={Distributed synthesis of local controllers for networked systems with arbitrary interconnection topologies},
  author={Agarwal, Etika and Sivaranjani, S and Gupta, Vijay and Antsaklis, Panos J},
  journal=IEEE_J_AC,
  volume={66},
  number={2},
  pages={683--698},
  year={2020},
  publisher={IEEE}
}

@book{nesterov1994interior,
  title={Interior-point polynomial algorithms in convex programming},
  author={Nesterov, Yurii and Nemirovskii, Arkadii},
  year={1994},
  publisher={SIAM}
}

@article{romer2019one,
  title={One-shot verification of dissipativity properties from input--output data},
  author={Romer, Anne and Berberich, Julian and K{\"o}hler, Johannes and Allg{\"o}wer, Frank},
  journal={IEEE Control Systems Letters},
  volume={3},
  number={3},
  pages={709--714},
  year={2019},
  publisher={IEEE}
}

@inproceedings{isik2023impact,
  title={Impact of Centralized and Distributed Control Structures on the Harmonic Stability of Modular Multilevel Converter Based on DQ Reference Frame Impedance Assessment},
  author={Isik, Semih and Bhattacharya, Subhashish},
  booktitle=ECCE,
  pages={811--816},
  year={2023},
  organization={IEEE}
}

@article{walsh2019interior,
  title={Interior-conic polytopic systems analysis and control},
  author={Walsh, Alex and Forbes, James Richard},
  journal=ARXIV,
  year={2019}
}

@book{lozano2013dissipative,
  title={Dissipative systems analysis and control: theory and applications},
  author={Lozano, Rogelio and Brogliato, Bernard and Egeland, Olav and Maschke, Bernhard},
  year={2013},
  publisher={Springer Science \& Business Media}
}

@techreport{gupta1996robust,
  title={Robust stabilization of uncertain systems based on energy dissipation concepts},
  author={Gupta, Sandeep},
  year={1996},
    institution={NASA}
}

@article{li2002delay,
  title={Delay-dependent dissipative control for linear time-delay systems},
  author={Li, Zhihu and Wang, Jingcheng and Shao, Huihe},
  journal=J_FI,
  volume={339},
  number={6-7},
  pages={529--542},
  year={2002},
  publisher={Elsevier}
}

@article{mahmoud2009dissipativity,
  title={Dissipativity analysis and synthesis of a class of nonlinear systems with time-varying delays},
  author={Mahmoud, Magdi S and Shi, Yan and Al-Sunni, Fouad M},
  journal=J_FI,
  volume={346},
  number={6},
  pages={570--592},
  year={2009},
  publisher={Elsevier}
}

@article{haddad2022dissipativity,
  title={Dissipativity theory for discrete-time nonlinear stochastic dynamical systems},
  author={Haddad, Wassim M and Lanchares, Manuel},
  journal=IJRNC,
  volume={32},
  number={11},
  pages={6293--6314},
  year={2022},
  publisher={Wiley Online Library}
}

@article{bridgeman2016conic,
  title={Conic bounds for systems subject to delays},
  author={Bridgeman, Leila Jasmine and Forbes, James Richard},
  journal={IEEE Transactions on Automatic Control},
  volume={62},
  number={4},
  pages={2006--2013},
  year={2016},
  publisher={IEEE}
}

@article{zames1966input1,
  title={On the input-output stability of time-varying nonlinear feedback systems part one: Conditions derived using concepts of loop gain, conicity, and positivity},
  author={Zames, George},
  journal=IEEE_J_AC,
  volume={11},
  number={2},
  pages={228--238},
  year={1966},
  publisher={IEEE}
}

@book{hatanaka2015passivity,
  title={Passivity-based control and estimation in networked robotics},
  author={Hatanaka, Takeshi and Chopra, Nikhil and Fujita, Masayuki and Spong, Mark W},
  year={2015},
  publisher={Springer}
}

@article{7286765,
  author={Bridgeman, Leila Jasmine and Forbes, James Richard},
  journal=IEEE_J_AC, 
  title={The Extended Conic Sector Theorem}, 
  year={2016},
  volume={61},
  number={7},
  pages={1931-1937},
  doi={10.1109/TAC.2015.2484331}}

@article{jovanovic2016controller,
  title={Controller architectures: Tradeoffs between performance and structure},
  author={Jovanovi{\'c}, Mihailo R and Dhingra, Neil K},
  journal=Eu_JC,
  volume={30},
  pages={76--91},
  year={2016},
  publisher={Elsevier}
}

@inproceedings{lian2018sparsity,
  title={Sparsity-Constrained Mixed {$H_2/H_\infty$} Control},
  author={Lian, Feier and Chakrabortty, Aranya and Wu, Fen and Duel-Hallen, Alexandra},
  booktitle=ACC,
  pages={6253--6258},
  year={2018},
  organization={IEEE}
}

@inproceedings{fardad2011sparsity,
  title={Sparsity-promoting optimal control for a class of distributed systems},
  author={Fardad, Makan and Lin, Fu and Jovanovi{\'c}, Mihailo R},
  booktitle=ACC,
  pages={2050--2055},
  year={2011},
  organization={IEEE}
}

@article{lin2013design,
  title={Design of optimal sparse feedback gains via the alternating direction method of multipliers},
  author={Lin, Fu and Fardad, Makan and Jovanovi{\'c}, Mihailo R},
  journal=IEEE_J_AC,
  volume={58},
  number={9},
  pages={2426--2431},
  year={2013},
  publisher={IEEE}
}

@article{slotine1987adaptive,
  title={On the adaptive control of robot manipulators},
  author={Slotine, Jean-Jacques E and Li, Weiping},
  journal={The international journal of robotics research},
  volume={6},
  number={3},
  pages={49--59},
  year={1987},
  publisher={Sage Publications Sage CA: Thousand Oaks, CA}
}

@article{aps2019mosek,
  title={Mosek optimization toolbox for matlab},
  author={ApS, Mosek},
  journal={User’s Guide and Reference Manual, Version},
  volume={4},
  number={1},
  pages={116},
  year={2019},
  publisher={MOSEK}
}

@inproceedings{lofberg2004yalmip,
  title={YALMIP: A toolbox for modeling and optimization in MATLAB},
  author={Lofberg, Johan},
  booktitle={2004 IEEE international conference on robotics and automation (IEEE Cat. No. 04CH37508)},
  pages={284--289},
  year={2004},
  organization={IEEE}
}

\end{document}